\documentclass[10pt,journal,compsoc]{IEEEtran}
\usepackage{amsmath,amssymb,amsfonts,amsthm}
\usepackage{mathrsfs}
\usepackage{array}
\usepackage{booktabs}
\usepackage[caption=false,font=normalsize,labelfont=sf,textfont=sf]{subfig}
\usepackage{textcomp}
\usepackage{stfloats}
\usepackage{url}
\usepackage{verbatim}
\usepackage{subfig,graphicx}
\usepackage{cite}
\usepackage{xcolor}
\definecolor{color}{rgb}{0, 0, 0}
\definecolor{color1}{rgb}{1, 0, 0}
\usepackage{bm}
\usepackage[ruled,linesnumbered]{algorithm2e}
\usepackage{algpseudocode}
\makeatletter
\newcommand{\removelatexerror}{\let\@latex@error\@gobble}
\makeatother
\usepackage{mathtools}
\usepackage{epstopdf}
\newtheorem{remark}{Remark}
\newtheorem{proposition}{Proposition}
\allowdisplaybreaks[4]
\def\BibTeX{{\rm B\kern-.05em{\sc i\kern-.025em b}\kern-.08em
		T\kern-.1667em\lower.7ex\hbox{E}\kern-.125emX}}

\begin{document}
	\title{UAV Swarm-enabled Collaborative Secure Relay Communications with Time-domain Colluding Eavesdropper}
	
	\author{
		Chuang Zhang,
		Geng~Sun,~\IEEEmembership{Member,~IEEE,}
		Qingqing~Wu,~\IEEEmembership{Senior Member,~IEEE,}\\
		Jiahui~Li,~\IEEEmembership{Student Member,~IEEE,}
		Shuang~Liang,\\
		Dusit~Niyato,~\IEEEmembership{Fellow,~IEEE,}
		and Victor~C.M.~Leung,~\IEEEmembership{Life Fellow,~IEEE}
		\IEEEcompsocitemizethanks{\IEEEcompsocthanksitem Chuang Zhang and Geng Sun are with the College of Computer Science and Technology, Jilin University, Changchun 130012, China, and also with the Key Laboratory of Symbolic Computation and Knowledge Engineering of Ministry of Education, Jilin University, Changchun 130012, China. \protect\\
			E-mail: chuangzhang1999@gmail.com, 
			sungeng@jlu.edu.cn.
			\IEEEcompsocthanksitem Qingqing Wu is with the Department of Electronic Engineering, Shanghai Jiao
			Tong University, Shanghai, China.\protect\\
			E-mail: qingqingwu@sjtu.edu.cn.  
			\IEEEcompsocthanksitem Jiahui Li is with the College of Computer Science and Technology, Jilin University, Changchun 130012, China, and also with Pillar of Engineering Systems and Design, Singapore University of Technology and Design, Singapore 487372.\protect\\
			E-mail: lijiahui0803@foxmail.com.
			\IEEEcompsocthanksitem Shuang Liang is with the School of Information Science and Technology, Northeast Normal University, Changchun, 130024, China.\protect\\
			E-mail: liangshuang@nenu.edu.cn.
			\IEEEcompsocthanksitem Dusit Niyato is with the School of Computer Science and Engineering, Nanyang Technological University, Singapore 639798. \protect\\
			E-mail: dniyato@ntu.edu.sg. 
			\IEEEcompsocthanksitem Victor C. M. Leung is with the College of Computer Science and Software Engineering, Shenzhen University, Shenzhen 518060, China, and also with the Department of Electrical and Computer Engineering, University of British Columbia, Vancouver, BC V6T 1Z4, Canada. \protect\\
			E-mail: vleung@ieee.org.}
		\thanks{(\textit{Corresponding author: Geng Sun})\\
			A small part of this paper appeared in IEEE CSCWD 2023~\cite{CSCWD}.}
	}
	\markboth{Journal of \LaTeX\ Class Files,~Vol.~14, No.~8, August~2021}%
	{Shell \MakeLowercase{\textit{et al.}}: A Sample Article Using IEEEtran.cls for IEEE Journals}
	
	
	\IEEEtitleabstractindextext{%
		%%%
		%   Abstract
		%%%
		\begin{abstract}\label{abstract}
			Unmanned aerial vehicles (UAVs) as aerial relays are practically appealing for assisting Internet of Things (IoT) network. In this work, we aim to utilize the UAV swarm to assist the secure communication between the micro base station (MBS) equipped with the planar array antenna (PAA) and the IoT terminal devices by collaborative beamforming (CB), so as to counteract the effects of collusive eavesdropping attacks in time-domain. Specifically, we formulate a UAV swarm-enabled secure relay multi-objective optimization problem (US$^2$RMOP) for simultaneously maximizing the achievable sum rate of associated IoT terminal devices, minimizing the achievable sum rate of the eavesdropper and minimizing the energy consumption of UAV swarm, by jointly optimizing the excitation current weights of both MBS and UAV swarm, the selection of the UAV receiver, the position of UAVs and user association order of IoT terminal devices. Furthermore, the formulated US$^2$RMOP is proved to be a non-convex, NP-hard and large-scale optimization problem. Therefore, we propose an improved multi-objective grasshopper algorithm (IMOGOA) with some specific designs to address the problem. Simulation results exhibit the effectiveness of the proposed UAV swarm-enabled collaborative secure relay strategy and demonstrate the superiority of IMOGOA.
		\end{abstract}
		
		\begin{IEEEkeywords}
			UAV swarm, collaborative beamforming, collusive eavesdropping, secure communication, multi-objective optimization.
	\end{IEEEkeywords}}
	\maketitle
	\IEEEdisplaynontitleabstractindextext
	\IEEEpeerreviewmaketitle
	%%%
	%	Section:Introduction
	%%%
	\section{Introduction}  
	\label{Section:Introduction}
	
	% 说明UAV的重要作用，转而引入其在无线网络中的应用
	\par \IEEEPARstart{D}{ue} to decreasing cost and advancements in manufacturing technology, unmanned aerial vehicles (UAVs) have a significant impact on military and commercial applications\cite{Zeng2019}, \cite{Xu2021}. Especially in the field of wireless networks, UAVs have created a boom and derived a lot of new application scenarios in industry and academia. Integrating UAVs into the incorporated network system becomes a foregone choice for the space–air–ground–aqua network \cite{Liu2020}, \cite{Xu2022}, \cite{LiJiahui2023}. For example, a UAV can be regarded as an aerial base station to assist Internet of Things (IoT) terminal devices for data upload scenario \cite{Samir2020}, \cite{Pan2023}, wherein these devices have limited transmission power and do not have the ability to communicate over long distances. Moreover, a UAV can also act as an aerial user to access the terrestrial network for environmental monitoring and goods delivery \cite{Zeng2019a}. In addition, extending limited network coverage in post-disaster rescue can be efficiently achieved by UAV-enabled multi-hop relay strategy \cite{Zhang2018}.
	%\IEEEPARstart{D}{ue} to the improvement of manufacturing technology and decrease of its cost, unmanned aerial vehicles (UAVs) have a significant impact on military and commercial applications \cite{Zeng2019}, \cite{Xu2021}. Especially in the field of wireless networks, UAVs have created a boom and derived a lot of new application scenarios in industry and academia. Integrating UAVs into the incorporated network system becomes an inevitable choice for the space–air–ground–aqua network \cite{Liu2020}, \cite{Xu2022}, \cite{LiJiahui2023}. For example, a UAV can be regarded as an aerial base station to assist Internet of Things (IoT) terminal devices for data upload scenario \cite{Samir2020}, \cite{Pan2023}, wherein these devices have limited transmission power and do not have the ability to communicate over long distances. Moreover, a UAV can also act as an aerial user to access the terrestrial network for environmental monitoring and goods delivery \cite{Zeng2019a}. In addition, extending limited network coverage in post-disaster rescue can be efficiently achieved by UAV-enabled multi-hop relay strategy \cite{Zhang2018}.
	
	% UAV中继，中继安全，传统的上层加解密不好用，从而引出CB技术
	\par The UAV-enabled relay communication is a process leveraging UAV relay some information between ground-based communication equipment, which can expand the reach of network. However, a single UAV as a high-rate relay to assist the terrestrial network system is a challenging task due to the restricted battery capacity and limited transmit power. For example, in some long-distance communication settings, the UAV relay must first move to a position near the sender before moving to a position near the receiver, which significantly reduces the network lifetime and efficiency. Moreover, due to the broadcast nature of the wireless channel in UAV relay communications, the security is a key issue that should be taken into account seriously. Although UAV flying at a higher latitude provides line-of-sight (LoS) dominant channels for wireless communications, these links are also more vulnerable to the eavesdropping attacks, especially for UAV swarm-enabled multi-hop relay strategy since the risk of eavesdropping increases with the increase of the number of hops. Generally, the security can be regarded as a higher layer communication protocol stack design concern that could be addressed by using encryption methods. However, this requires high computational ability \cite{GengSun2023}, which is not suitable for UAVs with limited resources.
	
	% 阐明CB技术如何改善安全
	\par Fortunately, collaborative beamforming (CB) \cite{Ochiai2005}, \cite{Ahmed2009}, as a communication technique originally used in wireless sensor networks, can enhance the signal strength and directivity. CB has garnered significant attention from researchers who seek to address the issue of secure and effective communication \cite{Zhang2010}, \cite{Yang2016}. Thus, it is reasonable to introduce CB for UAV swarm to assist terrestrial communications. Specifically, a UAV swarm-enabled virtual antenna array (UVAA) consisting of multiple UAVs can greatly improve the signal strength in a special direction by controlling the radio energy distribution, thereby increasing the transmission rate and enhancing the security of the UAV swarm-enabled relay system. Nevertheless, the UAV swarm-enabled collaborative secure relay communication system based on CB needs to consider several key factors. For example, UAVs in UVAA can move to suitable positions for achieving the higher achievable rate of legitimate user and the lower achievable rate of eavesdropper. However, this significantly causes additional energy consumption because of the movement of UAVs. Moreover, the excitation current weights of UVAA are crucial factors for the beam pattern which should be considered at the same time. Additionally, it needs to adopt the necessary approach to reduce the risk of eavesdropping for the source, e.g., the selection of UAV receiver is an important factor because this can cause different wiretap rates in CB information fusion phase. Thus, obtaining the more proper positions and excitation current weights of UAVs, selecting the appropriate UAV receiver for excellent and secure communication performance, and simultaneously reducing the movement energy consumption of the UAV swarm for the collaborative secure relay communication system are of importance. In this work, we further consider the joint optimization of source MBS and UVAA under the threat of time domain eavesdropper collusion in the complete relay communication process, which is a more practical scenario and a more comprehensive problem compared to \cite{Sun2022}. The major contributions of this paper are summarized as follows.
	
	\begin{itemize}
		\item \textbf{\textit{UAV Swarm-enabled Collaborative Secure Relay System Construction:}} We consider a secure relay communication scenario, where a UAV swarm-enabled collaborative secure relay system is constructed for transmitting confidential information from the source MBS with the planar array antenna (PAA) to the remote IoT terminal devices so as to counteract the threat of eavesdropper colluding in time-domain. To the best of our knowledge, this is the first work that considers the complete secure relay communication process from the source MBS to the remote IoT terminal devices assisted by the UAV swarm-enabled CB under the threat of time-domain collusive eavesdropper. {\color{color}Compared to existing work \cite{Sun2022}, \cite{Li2021}, the considered system is more comprehensive and practical.}  
		
		\item \textbf{\textit{Multi-objective Optimization Problem Formulation:}} We formulate a UAV swarm-enabled secure relay multi-objective optimization problem (US$^2$RMOP)  aiming to cooperatively maximize the achievable sum rate between the MBS and multiple remote IoT terminal devices, minimize the achievable sum rate of eavesdropper, and minimize the traveling energy consumption of the UAV swarm, by jointly optimize excitation current weights of both MBS and UAV swarm, the selection of the UAV receiver, the position of each UAV and user association order of IoT terminal devices. Furthermore, the US$^2$RMOP is proven to be a non-convex, large-scale optimization and NP-hard problem.
		
		\item \textbf{\textit{Algorithm Design:}} {\color{color}Due to the complex constraints and high-dimensional decision space of US$^2$RMOP, reinforcement learning and convex optimization algorithms face significant challenges, e.g., the curse of dimensionality and difficult convex relaxation.} Thus, we design an improved multi-objective grasshopper algorithm (IMOGOA) to solve the formulated US$^2$RMOP. First, IMOGOA adopts half-Halton-half-chaos (H$^3$C) and dynamic elimination-based crowding distance (DCDE) strategies to improve the distribution of the population. Moreover, the non-linear decreasing factor is used to better coordinate exploitation and exploration in IMOGOA. Additionally, we introduce L\'evy flight and archive update strategies to enhance the ability of going beyond the local optimum. {\color{color}The interaction among the aforementioned improvements enables IMOGOA to accomplish better diversity and uniformity when dealing with the formulated US$^2$RMOP.}
		% {\color{color1}Finally, IMOGOA handles the discrete part of solution by designing two points crossover (TPC) and partially-matched crossover (PMX) strategies, which makes it suitable for dealing with the formulated multi-objective optimization problem.}
		
		\item \textbf{\textit{Simulation Validation:}} Simulation results illustrate the performance of the proposed IMOGOA by comparing it with some benchmarks. Moreover, the traditional UAV swarm-enabled multi-hop relay and linear antenna array strategies are introduced to verify the practicability of the UAV swarm-enabled collaborative secure relay communication system. {\color{color}In addition, the performance comparison of the proposed IMOGOA under two situations with multiple eavesdroppers is further analyzed.}
	\end{itemize}
	
	\par The rest of this paper is organized as follows. Section \ref{Section:Related Work} introduces some related work. Section \ref{Section:System Model and Preliminaries} provides the system model and preliminaries. The formulation of US$^2$RMOP is detailed and analyzed in Section \ref{Section:Problem Formulation and Analysis}. Section \ref{Section:Proposed Algorithm} designs the multi-objective optimization algorithm. Section \ref{Section:Simulation Results} shows simulation results and the conclusion of this paper is presented in Section \ref{Section:Conclusion}.
	
	%%%
	%	Section:Related Work
	%%%
	\section{Related Work} 
	\label{Section:Related Work}
	
	\par The UAV-enabled communications have been widely studied in plenty of work. For example, Zeng \textit{et al}. \cite{Zeng2017} investigated a UAV-enabled data collection communication system with the aim of maximizing energy efficiency by optimizing UAV flight path. Zhang \textit{et al}. \cite{Zhang2019} studied a cellular-connected UAV flying mission completion problem for minimizing UAV energy consumption under the constraint quality-of-connectivity. Li \textit{et al}. \cite{Li2018} focused on a single UAV serving multiple ground users scenario to maximize the system throughput by planning UAV trajectory and optimizing subcarrier allocation strategy. Wu \textit{et al}. \cite{Wu2017} considered fair performance among multiple users by jointly optimizing user schedule and the UAV trajectory. Moreover, Hua \textit{et al}. \cite{Hua2020} investigated UAV-enabled simultaneous transmission in both uplink and downlink. Specifically, they considered two types of UAVs, one of which is regarded as an aerial base station that is responsible for receiving data from IoT devices, and the other serves as an aerial user for accessing the terrestrial network. The system throughput maximization is accomplished by corporately optimizing transmission power and trajectory of UAV and communication scheduling. In \cite{Meng2022}, the authors focused on UAV-enabled integrating communication and sensing system, where the user attainable rate is maximized through jointly optimizing transmit precoder, sensing start instant and the trajectory of UAV. In addition, Yang \textit{et al}. \cite{Yang2021} aimed at finding appropriately designed trajectory of a mobile UAV in backscatter communication system.
	
	\par The UAV-enabled secure communications have been the focus of several prior work. For example, in order to prevent a ground eavesdropper, Zhong \textit{et al}. \cite{Zhong2019} made use of the power and trajectory controls of both the UAV transmitter and a friendly UAV jammer. In \cite{Cai2018}, the authors proposed a dual-UAV enabled secure communication network involving multiple legitimate users and ground eavesdroppers, and the minimum worst-case secrecy rate for legitimate users was maximized by jointly optimizing trajectory of UAV and user scheduling. Zhou \textit{et al}. \cite{Zhou2018} investigated how friendly UAV jamming power and the corresponding three-dimensional deployment affected the likelihood of legitimate receivers being interrupted and the likelihood of unknown eavesdroppers being intercepted. Sun \textit{et al}. \cite{Sun2020} analyzed the secure performance of mmWave NOMA systems with both legitimate user and eavesdroppers by taking into account the spatial correlation between the selected legitimate users and eavesdroppers. In \cite{Cheng2019}, the authors used a novel iterative approach to jointly optimize the time schedule and trajectory of UAV to assure the security of UAV-relayed wireless networks. Na \textit{et al}. \cite{Na2022} considered a relay scenario by jointly optimizing resource allocation and UAV trajectory to maximize the minimum average secrecy rate among all IoT terminal devices. Moreover, in \cite{Ji2021}, the authors explored secure transmission in a cache-enabled UAV relay network with D2D communication and eavesdroppers. Specifically, they maximized the minimum secrecy rate between users by concurrently optimizing scheduling, trajectory and transmission power of UAV and user association.
	
	\par Several previous research has devoted into the UAV-enabled CB communications. For instance, Mohanti \textit{et al}. \cite{Mohanti2019} designed a UAV swarm-enabled CB framework and verified the feasibility of this scheme under air-to-ground channel. Mozaffari \textit{et al}. \cite{Mozaffari2019} investigated a UAV swarm-enabled CB technique for providing network service to ground users. Specifically, minimizing service time by reducing wireless transmission time as well as control time for UAV movement and stabilization are considered. Dinh \textit{et al}. \cite{Dinh2019} proposed a communication mode that considered both the flexible deployment and CB transmission of UAVs to maximize the number of admitted users by jointly optimizing the transmit beamforming, user admission decision, position planning and content placement. Zhu \textit{et al}. \cite{Zhu2018} studied a UAV swarm-enabled CB relay system, wherein minimizing the total transmit power of the UAV relays within the interference limits of the primary network and the quality of service (QoS) requirements of the cognitive network is formulated. Furthermore, Li \textit{et al}. \cite{Li2021} investigated a secure communication system where the UAVs communicate with multiple base stations by utilizing CB. In addition, in \cite{Sun2022}, Sun \textit{at al}. made use of CB to realize secure and energy-efficient communications for different terrestrial base stations.

	\par The primary distinctions between this work and the aforementioned research are seen that we consider a complete secure relay communication process between the MBS and remote IoT terminal devices in a UAV swarm-assisted terrestrial IoT network. Moreover, we study the more difficult scenario of security assurance, where the ground eavesdropper adopts a maximal ratio combining (MRC) technology in time domain for the relay process \cite{Yang2017}.
	
	\begin{figure}[!t]	\includegraphics[width=\linewidth,scale=1.00]{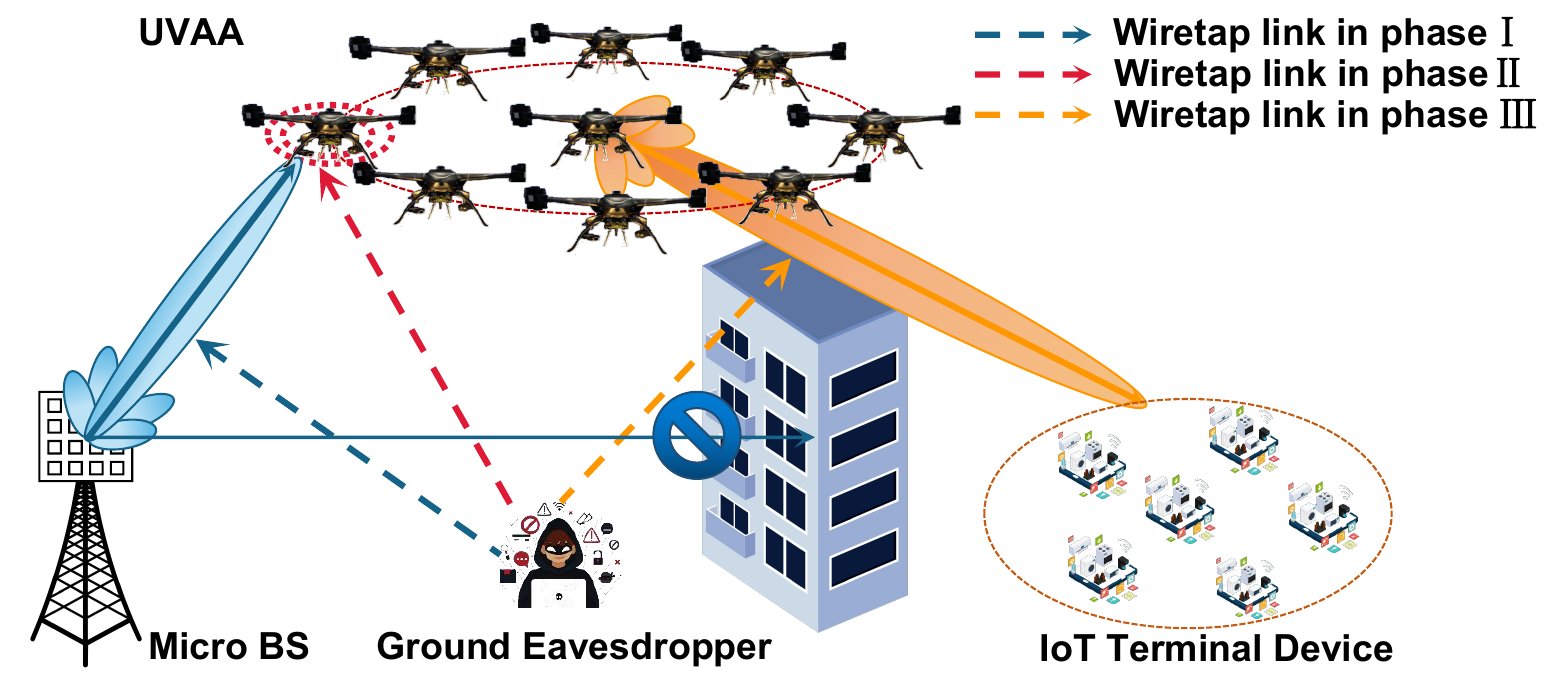}
		\caption{An illustration of UAV swarm-enabled collaborative secure relay communication system, where a UAV swarm is introduced to relay confidential messages between MBS (equipped with a PAA) and multiple IoT terminal devices via CB, and a ground eavesdropper performs eavesdropping in a time-domain colluding manner. Solid and dashed lines represent legitimate communication links and wiretap links, respectively.}
		\label{Figure:SystemModel}
	\end{figure}
	%%%
	%	Section:System Model and Preliminaries
	%%%
	\section{System Model and Preliminaries} 
	\label{Section:System Model and Preliminaries}
	\par As illustrated in Fig.~\ref{Figure:SystemModel}, we consider multiple secure communications from a source MBS $\mathcal{S}$ to $T$ associated IoT terminal devices expressed as $\mathcal{D}=\{D_{1}, D_{2}, ..., D_{T}\}$. Specifically, $\mathcal{S}$ is equipped with an $M \times N$ PAA to enhance the spatial resolution. However, due to the existence of obstacles, all direct communication links from $\mathcal{S}$ to $\mathcal{D}$ are blocked. Moreover, a ground eavesdropper $\mathcal{E}$ potentially intercepts information during the communication in the system. Thus, a UAV swarm consisting of $K$ UAVs, denoted as $\mathcal{U} = \{U_{1}, U_{2}, ..., U_{K}\}$, form a virtual antenna array to relay confidential messages via CB. Here, we assume that {\color{color} each UAV is equipped with an omni-directional antenna with the full-duplex mode and} the exact position of $\mathcal{E}$ can be detected by a radar \cite{Yan2016} or an optical camera \cite{Sun2019}.
	
	\par In the considered system, the $i$th communication process between $\mathcal{S}$ and associated IoT terminal devices has three phases:
	
	\begin{itemize}
		\item \textbf{Phase \uppercase\expandafter{\romannumeral1}}: In this phase, the information is transmitted from $\mathcal{S}$ to the UAV swarm. Specifically, $\mathcal{S}$ employs traditional beamforming to transmit information to the selected UAV denoted as $U_{k}$ in the UAV swarm, and the link between $\mathcal{S}$ and $U_{k}$ is denoted as $\mathcal{S}2U_{k}(i)$.
		
		\item \textbf{Phase \uppercase\expandafter{\romannumeral2}}: In this phase, the information fusion is conducted in the UAV swarm. Specifically, $U_{k}$ serves as the cluster leader, and it broadcasts the received message from $\mathcal{S}$ to other UAVs directly. To simplify this problem, we assume that all individuals in the UAV swarm can communicate with each other at a high rate within the cluster.
		
		\item \textbf{Phase \uppercase\expandafter{\romannumeral3}}: On the basis of the Phase \uppercase\expandafter{\romannumeral2}, the UAV swarm forwards the information to the associated IoT terminal device $D_{i}$ via CB, and the link between UVAA center and $D_{i}$ is denoted as $\mathcal{C}2\mathcal{D}(i)$.
	\end{itemize}
	
	\par Moreover, $\mathcal{E}$ is within the coverage of the MBS and UAV swarm, and it can perform eavesdropping by MRC to maximize the eavesdropping rate during all three phases above. Specifically, the wiretap links in phases \uppercase\expandafter{\romannumeral1}, \uppercase\expandafter{\romannumeral2} and \uppercase\expandafter{\romannumeral3} are denoted as $\mathcal{S}2\mathcal{E}(i)$, $U_{k}2\mathcal{E}(i)$ and $\mathcal{C}2\mathcal{E}(i)$, respectively.
	
	\par Without loss of generality, a 3D Cartesian coordinate system is considered, where the positions of $\mathcal{S}$, IoT terminal devices and $\mathcal{E}$ are fixed. Specifically, the position of the $m$th array element of PAA arranged in rows is denoted as $(x_{m}^{P}, y_{m}^{P}, z_{m}^{P})$. Moreover, the positions of the $i$th IoT terminal device and the $k$th UAV are expressed as $(x_{i}^{D}, y_{i}^{D}, z_{i}^{D})$ and $(x_{k}^{U}, y_{k}^{U}, z_{k}^{U})$, respectively. Moreover, the positions of the PAA and UVAA centers are denoted as $(x_{P}, y_{P}, z_{P})$ and $(x_{C}, y_{C}, z_{C})$, respectively. 
	
	\par To simplify the expression, $[\overline{x_{m}^{P}}, \overline{y_{m}^{P}},\overline{z_{m}^{P}}]$ and $[\overline{x_{k}^{C}},\overline{y_{k}^{C}},\overline{z_{k}^{C}}]$ denote the 3D-component distances of array element $m$ and $k$ to the centers of the PAA and UVAA, respectively. Moreover, we denote the link pair sets of ground-to-air (G2A), ground-to-ground (G2G) and air-to-ground (A2G) as $G2A = \{\mathcal{S}2U_{k}(i) | k = 1,...,K, i = 1,...,T\}$, $G2G = \{\mathcal{S}2\mathcal{E}(i) | i = 1,...,T\}$ and $A2G = \{\mathcal{C}2D(i) | i = 1,...,T\} \cup \{\mathcal{C}2\mathcal{E}(i) | i = 1,...,T\} \cup \{U_{k}\mathcal{E}(i) | k = 1,...,K, i = 1,...,T\}$, respectively.

	%
	%	SubSection:Channel Model
	%
	\subsection{Channel Model} 
	\label{SubSection:Channel Model}
	\par In this section, the corresponding channel models about the UAV swarm-enabled collaborative secure relay communication system are given. 
	%%	SubSubSection:G2A and A2G Channel
	\subsubsection{G2A and A2G Channels} 
	\label{SubSubSection:G2A and A2G Channels}
	
	\par UAVs can bring a higher probability of LoS link for communications compared to ground-based equipment. However, the simplified LoS channel model is not sufficient to accurately characterize the signal propagation in complex environments for G2A and A2G links. In this work, we adopts the angle-dependent probabilistic LoS channel model \cite{Duo2021}, which is described as
	\begin{equation}
		P_{lp}^{LoS} = \frac{1}{1 + ae^{-b(\zeta_{lp} - a)}}, lp \in G2A \cup A2G, 
		\label{Equ:A2G/G2A LoS Probability}
	\end{equation}
	
	\noindent where $a$ and $b$ represent the parameters of the activation function ($S$-curve) as attributed to the environments. Moreover, $\zeta_{lp} = \arctan(dv_{lp}/dh_{lp})$ denotes the elevation angle between the sender and receiver, wherein $dv_{lp}$ and $dh_{lp}$ are the vertical and horizontal distances between the sender and receiver, respectively. Then, the NLoS probability is calculated by $P_{lp}^{NLoS} = 1 - P_{lp}^{LoS}$. 
	
	\par Furthermore, the channel power gain is described as
	\begin{equation}
		h_{lp} = P_{lp}^{LoS}h_{lp}^{LoS} + P_{lp}^{NLoS}h_{lp}^{NLoS}, lp \in G2A \cup A2G, 
		\label{Equ:A2G/G2A Channel Power Gain}
	\end{equation} 
	
	\noindent where $h_{lp}^{LoS} = \beta_{0}d_{lp}^{-\alpha_{LoS}}$ and $h_{lp}^{NLoS} = \mu \beta_{0}d_{lp}^{-\alpha_{NLoS}}$ denote the channel power gains under the conditions of the LoS and NLoS states, respectively. Moreover, $\beta_{0}$ represents the average channel power gain at a reference distance $d_{0}=1$m in the LoS state, $\mu < 1$ is the additional signal attenuation factor due to the NLoS propagation, $\alpha_{LoS}$ and $\alpha_{NLoS}$ indicate the average path loss exponents for the LoS and NLoS states, respectively, and $d_{lp}$ represents the distance between the sender and receiver. 
	
	%%	SubSubSection:G2G Channel
	\subsubsection{G2G Channel} 
	\label{SubSubSection:G2G Channel}
	
	\par For the ground wiretap channel from $\mathcal{S}$ to $\mathcal{E}$, we can express the channel power gain as
	\begin{equation}
		h_{lp} = \beta_{0}d_{lp}^{\alpha_{G}}, lp \in G2G, 
		\label{Equ:G2G Channel Gain}
	\end{equation}
	
	\noindent where $\alpha_{G} > 2$ is the path loss exponent for the G2G link.
	
	%
	%	SubSection:Array Factor Model
	%
	\subsection{Array Factor Model} 
	\label{SubSection:Array Factor Model}
	
	\par In this work, the excitation current weights of the $m$th array element arranged in rows of PAA and the $k$th UAV element in UVAA are denoted as $I^{P}_{m}$ and $I^{U}_{k}$, respectively. Accordingly, the array factor (AF) \cite{Balanis2016} of PAA can be mathematically given as
	\begin{equation}
		\begin{small}
			\begin{aligned}
				&AF_{P}(\theta^{P},\varphi^{P}|\theta_{0}^{P},\varphi_{0}^{P})=\\&\sum_{m=1}^{M \times N} I^{P}_{m}e^{j\Psi^{P}_{m}(\theta_{0}^{P},\varphi_{0}^{P})}e^{j[c_{p}(\overline{x^{P}_{m}}\sin\theta^{P}\cos\varphi^{P}+\overline{y^{P}_{m}}\sin\theta^{P}\sin\varphi^{P}+\overline{z^{P}_{m}}\cos\theta^{P})]},
				\label{Equ:Array Factor of PAA}
			\end{aligned}
		\end{small}
	\end{equation}
	
	\noindent where $\theta^{P} \in [0, \pi]$ and $\varphi^{P} \in [-\pi, \pi]$ are the elevation and azimuth angles at the center of PAA, respectively. Moreover, $c_{p} = 2\pi/\lambda$ represents the phase constant, and $\lambda$ is the wavelength. According to \cite{Sun2022}, $\Psi^{P}_{m}$ represents the initial phase of the $m$th array element of PAA and can be determined as 
	\begin{equation}
		\begin{aligned}
			&\Psi^{P}_{m}(\theta_{0}^{P},\varphi_{0}^{P})  = \\&-\frac{2\pi}{\lambda}(\overline{x^{P}_{m}}\sin\theta_{0}^{P}\cos\varphi_{0}^{P}+\overline{y^{P}_{m}}\sin\theta_{0}^{P}\sin\varphi_{0}^{P}+\overline{z^{P}_{m}}\cos\theta_{0}^{P}),
			\label{Equ:Initial Phase of PAA}
		\end{aligned}
	\end{equation}
	
	\noindent where $(\theta_{0}^{P},\varphi_{0}^{P})$ represents the direction of the designated UAV receiver of PAA. Likewise, the AF of the UVAA can be described as follow:
	\begin{small}
		\begin{equation}
			\begin{aligned}
				&AF_{U}(\theta^{U},\varphi^{U}|\theta_{0}^{U},\varphi_{0}^{U})=\\&\sum_{k=1}^{K}I^{U}_{k}e^{j\Psi^{U}_{k}(\theta_{0}^{U},\varphi_{0}^{U})}e^{j[c_{p}(\overline{x^{U}_{k}}\sin\theta^{U}\cos\varphi^{U}+\overline{y^{U}_{k}}\sin\theta^{U}\sin\varphi^{U}+\overline{z^{U}_{k}}\cos\theta^{U})]},
				\label{Equ:Array Factor of UVAA}
			\end{aligned}
		\end{equation}
	\end{small}
	
	\noindent where $\theta^{U} \in [0, \pi]$ and $\varphi^{U} \in [-\pi, \pi]$ are the elevation and azimuth angles at the center of UVAA, respectively. Moreover, $\Psi^{U}_{k}$ is the initial phase of the $k$th UAV of UVAA, and can be determined by
	\begin{equation}
		\begin{aligned}
			&\Psi^{U}_{k}(\theta_{0}^{U},\varphi_{0}^{U}) = \\&-\frac{2\pi}{\lambda}(\overline{x^{U}_{k}}\sin\theta_{0}^{U}\cos\varphi_{0}^{U}+\overline{y^{U}_{k}}\sin\theta_{0}^{U}\sin\varphi_{0}^{U}+\overline{z^{U}_{k}}\cos\theta_{0}^{U}),
			\label{Equ:Initial Phase of UVAA}
		\end{aligned}
	\end{equation}
	
	\noindent where $(\theta_{0}^{U}, \varphi_{0}^{U})$ is the direction of the designated associated IoT terminal device. 
	
	%
	%	SubSection:Achievable Rate Model
	%
	\subsection{Achievable Rate Model} 
	\label{SubSection:Achievable Rate Model}
	In this section, the achievable rates of IoT terminal Devices and the ground eavesdropper are presented.
	%%	SubSubSection:Achievable Rate of IoT Terminal Device
	\subsubsection{Achievable Rate of the IoT Terminal Device $\mathcal{D}_{i}$} 
	\label{SubSubSection:Achievable Rate of IoT Terminal Device}
	
	\par For the phases \uppercase\expandafter{\romannumeral1} and \uppercase\expandafter{\romannumeral3} of the communication process, the signal-to-noise ratio (SNR) of $\mathcal{S}2U_{k}(i)$ and $\mathcal{C}2D(i)$ can be calculated as
	\begin{equation}
		\gamma_{\mathcal{S}2U_{k}}(i) = \frac{P_{\mathcal{S}}G_{0}^{P}h_{\mathcal{S}2U_{k}}(i)}{\sigma^2}
		\label{Equ:the SNR between BS and UAV k}
	\end{equation}
	\noindent and
	\begin{equation}
		\gamma_{\mathcal{C}2\mathcal{D}}(i) = \frac{P_{U}G_{0}^{U}h_{\mathcal{C}\mathcal{D}}(i)}{\sigma^2}, 
		\label{Equ:the SNR between UVAA and IoT i}
	\end{equation}
	
	\noindent where $P_{\mathcal{S}}$ and $P_{U}$ are the transmission power of PAA and UVAA, respectively. Moreover, $h_{\mathcal{S}2U_{k}}(i)$ represents the channel power gain between $\mathcal{S}$ and designated UAV receiver $U_{k}$ in the $i$th communication process, $h_{\mathcal{C}2D}(i)$ is the channel power gain between the UAV swarm and designated IoT terminal device $\mathcal{D}$ in the $i$th communication process, and the noise power of the channel is represented as $\sigma^2$. In addition, the gain $G_{0}^{P}$ and $G_{0}^{U}$ of PAA and UVAA towards the legitimate receivers can be respectively calculated as
	\begin{equation}
		G_{0}^{P} = \frac{4\pi\left|AF_{P}(\theta_{0}^{P},\varphi_{0}^{P}|\theta_{0}^{P},\varphi_{0}^{P})\right|^2w\left(\theta_{0}^{P}, \varphi_{0}^{P}\right)^{2}}{\int_{0}^{2 \pi} \int_{0}^{\pi}|AF_{P}(\theta^{P}, \varphi^{P})|^{2} w(\theta^{P}, \varphi^{P})^{2} \sin \theta^{P} \mathrm{d} \theta^{P} \mathrm{d} \varphi^{P}} \eta_{P} 
		\label{Equ:the array gain torwards UAV k}
	\end{equation}
	\noindent and 
	\begin{equation}
		G_{0}^{U} = \frac{4\pi\left|AF_{U}(\theta_{0}^{U},\varphi_{0}^{U}|\theta_{0}^{U},\varphi_{0}^{U})\right|^2w\left(\theta_{0}^{U}, \varphi_{0}^{U}\right)^{2}}{\int_{0}^{2 \pi} \int_{0}^{\pi}|AF(\theta^{U}, \varphi^{U})|^{2} w(\theta^{U}, \varphi^{U})^{2} \sin \theta^{U} \mathrm{d} \theta^{U} \mathrm{d} \varphi^{U}} \eta_{U}, 
		\label{Equ:the array gain torwards IoT i}
	\end{equation}
	
	\noindent where $w(\theta^P, \varphi^P)$ and $w(\theta^U, \varphi^U)$ represent the magnitude of the far-field beam pattern of each array element in PAA and UVAA, respectively. Moreover, $\eta_{P}$ and $\eta_{U}$ are the antenna efficiencies of PAA and UVAA, respectively. Note that $w(\theta^P, \varphi^P)$ and $w(\theta^U, \varphi^U)$ are $0$ dB in all directions in this system since we consider each array element of PAA and UVAA equipped with a single isotropic antenna with identical power constraints.
	
	\par Accordingly, the achievable rate between $\mathcal{S}$ and designated IoT terminal device $\mathcal{D}$ in the $i$th communication process can be expressed as
	\begin{equation}
		R_{\mathcal{S}2\mathcal{D}}(i) = B\log_{2}\left(1 + \widehat{\min}\{\gamma_{\mathcal{S}2U_{k}}(i), \gamma_{\mathcal{C}2\mathcal{D}}(i)\}\right), 
		\label{Equ:the rate between BS and IoT i}
	\end{equation}
	
	\noindent {\color{color}where $B$ represents the transmission bandwidth, and $\widehat{\min}\{\cdot\}$ refers to an operator that calculates the minimum value of the two elements. Note that we ignore the SNR limitation of the broadcasting process due to the close individual distance between the UAVs.}
	
	% \begin{remark}
		% 	We ignore the SNR limitation of the broadcasting process due to the close individual distance between the UAVs. {\color{color1}Moreover, three channels are required
			% 	for transmitting one message}. {\color{color}Moreover, due to the half-duplex operation mode of UAVs, three time units are required for serving a single IoT terminal device}. Thus, Eq.~\eqref{Equ:the rate between BS and IoT i} exists a pre-\emph{log} factor $\frac{1}{3}$ \emph{\cite{Ng2011}}.
		% \end{remark}
	
	%%	SubSubSection:Achievable Rate of GE
	\subsubsection{Achievable Rate of the Ground Eavedropper $\mathcal{E}$}
	\label{SubSubSection:Achievable Rate of GE}
	
	\par For the phases \uppercase\expandafter{\romannumeral1}, \uppercase\expandafter{\romannumeral2} and \uppercase\expandafter{\romannumeral3} of the $i$th communication process, the SNRs of three wiretap links, i.e., $\mathcal{S}2\mathcal{E}(i)$, $U_{k}2\mathcal{E}(i)$ and $\mathcal{C}2\mathcal{E}(i)$, can be calculated as
	\begin{equation}
		\label{Equ:phase 1 eavesdropping SNR}
		\begin{aligned}
			\gamma_{\mathcal{\mathcal{S}}2\mathcal{E}}(i) = \frac{P_{S}G_{\mathcal{E}}^{P}h_{\mathcal{S}2\mathcal{E}}(i)}{\sigma^{2}},
		\end{aligned}
	\end{equation}
	\begin{equation}
		\label{Equ:phase 2 eavesdropping SNR}
		\begin{aligned}
			\gamma_{U_{k}2\mathcal{E}}(i) = \frac{P_{U_{k}}h_{U_{k}2\mathcal{E}}(i)}{\sigma^2}
		\end{aligned}
	\end{equation}
	\noindent and
	\begin{equation}
		\label{Equ:phase 3 eavesdropping SNR}
		\gamma_{\mathcal{C}2\mathcal{E}}(i) = \frac{P_{U}G_{\mathcal{E}}^{U}h_{\mathcal{C}2\mathcal{E}}(i)}{\sigma^2},
	\end{equation}
	
	\noindent where $G_{\mathcal{E}}^{P}$ and $G_{\mathcal{E}}^{U}$ can be calculated according to the same principle as Eqs.~\eqref{Equ:the array gain torwards UAV k} and \eqref{Equ:the array gain torwards IoT i}. Moreover, $P_{U_{k}}$ represents the broadcast transmission power of $U_{k}$ in phase \uppercase\expandafter{\romannumeral2}.
	
	\par Accordingly, the achievable rate of $\mathcal{E}$ during the $i$th communication process can be expressed as
	\begin{equation}
		R_{\mathcal{E}}(i) = B\log_{2}\left(1 + \gamma_\mathcal{E}(i)\right), 
		\label{Equ:the rate of GE}
	\end{equation}
	
	\noindent where $\gamma_\mathcal{E}(i)$ is the maximal SNR obtained by using MRC technique during the $i$th communication process, and can be calculated as $\gamma_\mathcal{E}(i) = \gamma_{\mathcal{S}2\mathcal{E}}(i) + \gamma_{U_{k}2\mathcal{E}}(i) + \gamma_{C\mathcal{E}}(i)$.
	
	%
	%	SubSection:Rotary-Wing UAV Energy Consumption Model
	%
	\subsection{Rotary-Wing UAV Energy Consumption Model} 
	\label{SubSection:Rotary-Wing UAV Energy Consumption Model}
	
	\par For the UAV swarm-enabled collaborative secure relay communication system, the energy consumption of UAVs consists of the communication energy consumption generated by transmitting data and the propulsion energy consumption to overcome air drag and gravity. Furthermore, the communication energy consumption is usually two orders of magnitude smaller than the propulsion energy consumption in practical applications \cite{Ding2020}. Therefore, we only concern the propulsion energy consumption of the UAV swarm in this paper. According to \cite{Zeng2019}, the energy consumption for a UAV flying in a straight-and-level manner with the speed $v$ can be modeled as
	\begin{equation}
		\begin{aligned}
			P(v) =&P_{b}\left(1+\frac{3v^{2}}{u_{tips}^2}\right)\,\, +P_{i}\left(\sqrt{1+\frac{v^4}{4u_{0}^4}}-\frac{v^2}{2u_{0}^2}\right)^{\frac{1}{2}}+\\ &\frac{1}{2}d_{0}\rho sAv^{3},
		\end{aligned}   
		\label{Equ:the power estimated about UAV}
	\end{equation}
	
	\noindent where $P_{b}$ and $P_{i}$ are two constants related to the flight speed ${v}$, which denote the blade profile power and induced power under the hovering condition, respectively. $u_{tips}$ is the tip speed of the rotor blade, and $u_{0}$ represents the mean rotor-induced velocity in hovering. Moreover, $d_0$ and $\rho$ denote the fuselage drag ratio and air density, respectively. $s$ and $A$ denote the rotor solidity and rotor disc area, respectively.
	
	\par The energy consumption including the UAV climbing and descending with time by using the heuristic closed-form can be measured as
	\begin{equation}
		\begin{aligned}
			E(T)\approx &\int_{0}^{T}P\left(v(t)\right)dt + mg\left(h(T)-h(0)\right)  \\& + \frac{1}{2}m\left(v(T)^{2}-v(0)^2\right),
		\end{aligned}
		\label{Equ:the energy consumption about UAV}
	\end{equation}
	
	\noindent where $T$ refers to the duration of flight time, and $v(t)$ represents the UAV speed at the time instant $t$. Moreover, $g$ and $m$ denote the gravitational acceleration and the mass of the UAV.
	
	%
	%	SubSection:Multi-Objective Optimization Problem
	%
	%\subsection{Multi-Objective Optimization Problem}
	%\label{SubSection:Multi-Objective Optimization Problem}
	%
	%Theoretically, a basic multi-objective optimization problem (MOP) is modeled as follows\cite{Alaouchiche2021}:
	%
	%\begin{equation}
	%	\begin{aligned}
		%		\textbf{optimize}\ &\enspace F(x)=[f_{1}(x), f_{2}(x), ..., f_{k}(x)]^{T}\\ 
		%		\text{s.t.}\ &\enspace g_{j}(x) \leq 0, j=1,2,...,m 
		%		\label{Equ:MOP}
		%	\end{aligned}
	%\end{equation}
	%
	%\noindent where xx is the decision variable, kk is the number of the optimization objectives, mm is the number of constraints, fi(x)f_{i}(x) is the iith objective function, and gj(x)g_{j}(x) is the jjth constraints function. In addition, \textbf{optimize} is max\max or min\min function.
	%
	%\par Different from the single-objective optimization problem, the solution of the MOP is a set, called Pareto Set \cite{Alaouchiche2022}. In the Pareto Set, each element is an optimal solution that cannot be dominated by other solutions. Moreover, Their image is described as the Pareto front.

	%%%
	%	Section:Problem Formulation and Analysis
	%%%
	\section{Problem Formulation and Analysis} 
	\label{Section:Problem Formulation and Analysis}
	
	\par In this section, the US$^2$RMOP is formulated and the corresponding analysis of the problem is presented.
	
	%
	%	SubSection:Problem Formulation
	%
	\subsection{Problem Formulation} 
	\label{SubSection:Problem Formulation and Analysis}
	
	\par In the considered scenario, $\mathcal{S}$ transmits the data to the associated IoT terminal devices with the assistance of UAV swarm as a relay. The main goal of the UAV swarm-enabled collaborative secure relay system is to guarantee the achievable rate of the associated IoT terminal devices while minimizing the achievable rate of $\mathcal{E}$. 
	
	\par Specifically, maximizing the achievable sum rate of IoT terminal devices and minimizing the achievable sum rate of eavesdropper can be realized by optimizing the beam pattern of the PAA and UVAA. According to Eqs.~\eqref{Equ:Array Factor of PAA} and \eqref{Equ:Array Factor of UVAA}, the positions and excitation current weights of the array elements can be adjusted to accomplish the better directivity of PAA and UVAA, which means that we can let the UAVs fly to better positions, and use optimal excitation current weights for the relay communication. Besides, the proper selection of UAV receiver $U_k$ in the UVAA can also increase the achievable rate of IoT terminal devices and reduce the achievable rate of $\mathcal{E}$. However, the energy consumption of the UAV elements in the UVAA will undoubtedly increase due to movement, and the positions of UAVs need to be re-tuned after communicating with an IoT terminal device since the mainlobe of the UVAA can only direct in the direction of one receiver each time. Accordingly, the multiple performances of the system should be comprehensively considered.
	
	\par Defining the optimization decision variable, i.e., the solution as $\mathbb{X} = (\mathbb{I}^{\mathrm{MN} \times \mathrm{T}}_{\mathrm{P}}, \mathbb{S}^{1 \times \mathrm{T}}_{\mathrm{P}}, \mathbb{I}^{\mathrm{K} \times \mathrm{T}}_{\mathrm{U}}, \mathbb{P}^{\mathrm{K} \times \mathrm{T}}_\mathrm{U}, \mathbb{O}^{1 \times \mathrm{T}}_{\mathrm{U}})$, which is detailed in Table \ref{Table:Statement of Variables}, and the three optimization objectives are formulated as follows.
	
	\par \textbf{\emph{Optimization objective 1}:} The first optimization objective is to maximize the achievable sum rate of IoT terminal devices, which is related to the excitation current weights of both PAA and UVAA, the position of each UAV, and the selection of UAV receiver. Therefore, the first objective function can be designed as
	\begin{equation}
		\begin{aligned}
			f_{1}(\mathbb{I}^{\mathrm{MN} \times \mathrm{T}}_{\mathrm{P}}, \mathbb{S}^{1 \times \mathrm{T}}_{\mathrm{P}}, \mathbb{I}^{\mathrm{K} \times \mathrm{T}}_{\mathrm{U}}, \mathbb{P}^{\mathrm{K} \times \mathrm{T}}_\mathrm{U}) = \sum_{i = 1}^{T} R_{\mathcal{S}2\mathcal{D}}(i). 
			\label{Equ:Optimization objective 1}
		\end{aligned}
	\end{equation}
	\begin{remark}
		Some factors of phases \uppercase\expandafter{\romannumeral1} and \uppercase\expandafter{\romannumeral2} influence the achievable sum rate of IoT terminal device. Specifically, the excitation current weights of PAA and the selection of UAV receiver $U_{k}$ affect the rate of $\mathcal{S}2U_{k}(i)$ link, and the excitation current weights and UAV positions of UVAA have an impact on the rate of $\mathcal{C}2\mathcal{D}(i)$ link.
	\end{remark}
	
	\par \textbf{\emph{Optimization objective 2}:} Minimizing the achievable sum rate of $\mathcal{E}$ is considered as the second optimization objective, and the corresponding objective function is designed as
	\begin{equation}
		f_{2}(\mathbb{I}^{\mathrm{MN} \times \mathrm{T}}_{\mathrm{P}}, \mathbb{S}^{1 \times \mathrm{T}}_{\mathrm{P}}, \mathbb{I}^{\mathrm{K} \times \mathrm{T}}_{\mathrm{U}}, \mathbb{P}^{\mathrm{K} \times \mathrm{T}}_\mathrm{U}) = \sum_{i = 1}^{T} R_{\mathcal{E}}(i). 
		\label{Equ:Optimization objective 2}
	\end{equation}
	
	\begin{remark}
		The achievable sum rate of $\mathcal{E}$ is closely related to three phases of the relay communication process. Specifically, the excitation current weights of PAA affect the rate of $\mathcal{S}2\mathcal{E}(i)$ link, the selection of UAV receiver $U_{k}$ has an impact on the rate of $U_{k}2\mathcal{E}(i)$ link, and the excitation current weights and UAV positions of UVAA influence the rate of $\mathcal{C}2\mathcal{E}(i)$ link.
	\end{remark}
	
	\par \textbf{\emph{Optimization objective 3}:} The third optimization objective is to minimize the energy consumption of UAV swarm, and this optimization objective is relevant to both the user association order\footnote{User association order is defined as the serving order of the MBS to IoT terminal devices.} of remote IoT terminal devices and the positions of UAV swarm. Thus, the specific objective function is designed as
	\begin{equation}
		f_{3}(\mathbb{P}^{\mathrm{K} \times \mathrm{T}}_{\mathrm{C}}, \mathbb{O}^{1 \times \mathrm{T}}) = \sum_{i=1}^{T}\sum_{k=1}^{K}E_{k}(i), 
		\label{Equ:Optimization objective 3}
	\end{equation}
	
	\noindent where $E_{k}(i)$ represents the motion energy consumption of $k$th UAV for communicating in the $i$th communication process. Moreover, trajectory design style, speed control strategy and the adopted model of each UAV are the same as \cite{Li2021}.
	
	\begin{remark}
		The hovering energy consumption of UAVs is not taken into account since it is positively correlated with the hovering time \cite{Babu2021}. Furthermore, the hovering time is related to the communication rate and data transfer volume. Apparently, the communication rate has been considered in optimization objective 1, and data transfer volume for each associated IoT terminal device is decided by user behaviors. Thus, the hovering energy consumption of UAVs is not necessary to be optimized separately.
	\end{remark}
	
	\par In summary, considering the three optimization objectives mentioned above, the US$^2$RMOP can be formulated as follows.
	\begin{equation}
		\label{UASCMOP}
		\begin{aligned}
			\textbf{P1}:\ & \underset{\{\mathbb{X}\}}{\text{min}}\  F=\{-f_1,f_2,f_3\}, \\
			\text{s.t.}\ &C1:0 \leq I_{m,i}^{P} \leq 1, \forall m \in \{1,...,M \times N\}, \forall i \in \{1,..., T\},\\
			&C2:1 \leq S_{i} \leq K, \forall i \in \{1,...,T\}, \\
			&C3:0 \leq I_{k,i}^{U} \leq 1, \forall k \in \{1,...,K\}, \forall i \in \{1,...,T\}, \\
			&C4:X_{min} \leq X_{k}^{U} \leq X_{max}, \forall k \in \{1,..., K\},\\
			&C5:Y_{min} \leq Y_{k}^{U} \leq Y_{max}, \forall k \in \{1,..., K\},\\
			&C6:Z_{min} \leq Z_{k}^{U} \leq Z_{max}, \forall k\in \{1,..., K\},\\
			&C7:1 \leq O_{i} \leq T, \forall i \in \{1,...,T\}, \\
			&C8:O_{i_{1}} \neq O_{i_{2}}, \forall i_{1} \neq i_{2}, \\
			&C9:\Vert P_{k_{1},i}, P_{k_{2},i} \Vert \ge D_{min}^{U}, \forall k_{1}, k_{2} \in \{1,...,K\},
		\end{aligned}
		\nonumber
	\end{equation}
	
	\noindent where $C1$ and $C3$ indicate the range of excitation current weights of PAA and UVAA, $C2$ denotes the selection range of UAV receiver, and the 3D movement area of the UAV is limited by $C4$, $C5$ and $C6$, respectively. Moreover, $C7$ and $C8$ ensure the service fairness for each associated IoT terminal device. Moreover, the collision constraint between UAVs is expressed as $C9$ where $\Vert P_{k_1,i}, P_{k_2,i} \Vert$ represents the distance between the $k_{1}$th UAV and the $k_{2}$th UAV for serving the $i$th associated IoT terminal device.
	
	%
	%	SubSection:Problem Analysis
	%
	\subsection{Problem Analysis}
	\label{SubSection:Problem Analysis}
	
	\par In this section, the formulated US$^2$RMOP is analyzed.
	
	\begin{proposition}
		The US$^2$RMOP is an NP-hard and non-convex optimization problem.
	\end{proposition}
	\begin{proof}
		To simplify this discussion, we only consider the objective $f_{1}$ without $f_{2}$ and $f_3$ in the context of serving a single IoT terminal device. The US$^2$RMOP can be reduced as follow:
		
		\begin{equation}
			\label{Equ:Reduce1}
			\begin{aligned}
				\textbf{P2}:\ &\underset{\{\mathbb{X}_1\}}{\text{min}}\ -f_1, \\
				&\text{s.t.}\ C1-C6, \text{and}\  C9,
			\end{aligned}
			\nonumber
		\end{equation}
		
		\noindent where $\mathbb{X}_1 = (\mathbb{I}^{\mathrm{MN}}_{\mathrm{P}}, \mathbb{S}^{1}_{\mathrm{P}}, \mathbb{I}^{\mathrm{K}}_{\mathrm{C}}, \mathbb{P}^{\mathrm{K}}_\mathrm{C})$ is the part of decision variable $\mathbb{X}$. As can be seen, $\textbf{P2}$ can be specified as a Mixed-Integer Nonlinear Programming (MINLP) problem, which is a typical NP-hard and non-convex optimization problem \cite{Burer2012}. Clearly, \textbf{P1} is more difficult to be solved than \textbf{P2} since it adds the coupling to the optimization objectives $f_2$ and $f_3$. Thus, the formulated US$^2$RMOP is an NP-hard and non-convex optimization problem.
	\end{proof}
	
	\begin{proposition}
		The formulated US$^2$RMOP is a large-scale optimization problem.
	\end{proposition}
	\begin{proof}
		The solution space of US$^2$RMOP is composed of the excitation current weights $\mathbb{I}_{\mathrm{P}}$, the selection of UAV receiver $\mathbb{S}_{\mathrm{P}}$, the excitation current weight distribution of UVAA $\mathbb{I}_{\mathrm{U}}$, the position of UAV swarm $\mathbb{P}_{\mathrm{U}}$ and the associated order of remote IoT terminal devices $\mathbb{O}$. Thus, the decision space dimension of US$^2$RMOP is $\left(\left(M \times N + 2 + 4 \times K\right) \times T\right)$. As the numbers of PAA elements, UAVs and IoT terminal devices increase, the decision space dimension will increase accordingly. On this basis, the formulated US$^2$RMOP is a large-scale optimization problem \cite{Cao2020}.
	\end{proof}
	
	% 表格跨栏
	\begin{table*}[htbp]    
		\renewcommand\arraystretch{2}
		\centering 
		\caption{{Description of Variables}}  
		\label{Table:Statement of Variables}  
		\begin{tabular}{l l l l}  \toprule[1.5pt]
			\multicolumn{1}{m{1.5cm}}{Variable Set}&                  
			\multicolumn{1}{m{5.5cm}}{Variable Element}&              
			\multicolumn{1}{m{5cm}}{Physical Meaning}&         
			\multicolumn{1}{m{4.5cm}}{A Case in Point}\\ 
			\toprule[1pt]
			
			\multicolumn{1}{m{1.5cm}}{$\mathbb{I}^{\mathrm{MN} \times \mathrm{T}}_{\mathrm{P}}$}& \multicolumn{1}{m{5.5cm}}{$\{I_{m,i}^{P} |  m \in \{1,...,M \times N\}, i \in \{1,...,T\}\}$}&   
			\multicolumn{1}{m{5cm}}{$\mathbb{I}^{\mathrm{MN} \times \mathrm{T}}_{\mathrm{P}}$ represents the excitation current weights of PAA, while $I_{m,i}^{P}$ is the excitation current weight of the $m$th row array element for serving the $i$th associated IoT terminal device in PAA.}&    
			\multicolumn{1}{m{4.5cm}}{ $I_{1,1}^{P}=1.0$ means the excitation current weight of the first array element is 1.0 in PAA for serving the first associated IoT terminal device.}\\ 
			\hline
			
			\multicolumn{1}{m{1.5cm}}{$\mathbb{S}^{1 \times \mathrm{T}}_{\mathrm{P}}$}& 
			\multicolumn{1}{m{5.5cm}}{$\{S_{i} | i\in \{1,...,T\}\}$} & 
			\multicolumn{1}{m{5cm}}{$\mathbb{S}^{1 \times \mathrm{T}}_{\mathrm{P}}$ represents the selection of UAV receiver, while $S_{i}$ is the receiver of PAA for serving the $i$th associated IoT terminal device .}&   
			\multicolumn{1}{m{4.5cm}}{$S_{1} = 1$ means the first UAV is scheduled to acting as a receiver of PAA for serving the first associated IoT terminal device.} \\ 
			\hline
			
			\multicolumn{1}{m{1.5cm}}{$\mathbb{I}^{\mathrm{K} \times \mathrm{T}}_{\mathrm{U}}$}& 
			\multicolumn{1}{m{5.5cm}}{$\{I_{k,i}^{{U}} |k \in \{1,...,K\},i \in \{1,...,T\} \}$}&
			\multicolumn{1}{m{5cm}}{$\mathbb{I}^{\mathrm{K} \times \mathrm{T}}_{\mathrm{U}}$ represents the excitation current weight of UVAA, while $I_{k,i}^{U}$ is the excitation current weight of the $k$th UAV in UVAA for serving $i$th associated IoT terminal device.}&
			\multicolumn{1}{m{4.5cm}}{$I_{1,1}^{U}=1.0$ means the excitation current weight of the first UAV element is 1.0 in UVAA for serving the first associated IoT terminal device.}\\ \hline
			
			\multicolumn{1}{m{1.5cm}}{$\mathbb{P}^{\mathrm{K} \times \mathrm{T}}_{\mathrm{U}}$}&  
			\multicolumn{1}{m{5.5cm}}{\{$P_{k,i} |k \in \{1,..., K\}, i \in \{1,..., T\}$\}}&  
			\multicolumn{1}{m{5cm}} {$\mathbb{P}^{\mathrm{K} \times \mathrm{T}}_{\mathrm{U}}$ represents the position of UAV swarm, while $P_{k,i}$ is the position of the $k$th UAV for serving the $i$th associated IoT terminal device.}& 
			\multicolumn{1}{m{4.5cm}}{$P_{1,1}=(300,300,100)$ means the position of the first UAV for serving the first associated IoT terminal device is $(300,300,100)$.}\\ 
			\hline
			
			\multicolumn{1}{m{1.5cm}}{$\mathbb{O}^{1 \times \mathrm{T}}$}& 
			\multicolumn{1}{m{5.5cm}}{\{$O_{i}| i \in \{1,...,T\}$\}}  & 
			\multicolumn{1}{m{5cm}}{$\mathbb{O}^{1 \times \mathrm{T}}$ represents the association order of IoT terminal devices, while $O_{i}$ is the ID of the $i$th associated IoT terminal device.}   & 
			\multicolumn{1}{m{4.5cm}}{$O_{1}=1$ means UVAA first serves the IoT terminal device with ID 1.}\\
			\toprule[1.5pt]
		\end{tabular}  
	\end{table*}
	
	%%%
	%	Section:Proposed Algorithm
	%%%
	\section{Proposed Algorithm} 
	\label{Section:Proposed Algorithm}
	\par The approaches to address the formulated US$^2$MOP can be broadly categorized into three groups, i.e., convex optimization methods, reinforcement learning and evolutionary algorithms (EA). Specifically, due to involving complex constraints, solving US$^2$MOP through relaxation and duality using convex optimization techniques is difficult. Likewise, the curse of dimensionality can impact reinforcement learning due to the presence of a large number of decision variables, resulting in extensive state and action spaces. 
	\par EA are classical stochastic search methods that simulate the natural selection and evolution of creatures. Compared to the other two categories of algorithms analyzed above, EA have strong robustness, global search capability, and adaptability, making them effective for solving the non-convex, NP-hard and larger-scale optimization problems. Among EA, the performances of grasshopper optimization algorithm (GOA) and its corresponding multi-objective grasshopper optimization algorithm (MOGOA) \cite{Mirjalili2018} are effective and they have been applied for solving problems in different areas such as the financial stress prediction \cite{Luo2018}, trajectory optimization \cite{Wu2017a}, and training neural network\cite{Heidari2019}, etc. Thus, we intend to use them as basic algorithm frameworks to deal with the formulated US$^2$RMOP.
	
	%
	%	SubSection:Conventional GOA and MOGOA
	%
	\subsection{Conventional GOA and MOGOA} 
	\label{SubSection:Conventional GOA and MOGOA}
	
	\par GOA is enlightened by the behavior of the grasshopper swarm, where the position of each individual grasshopper in the population stands for a feasible solution to the given optimization problem. As shown in Fig.~\ref{Figure:GOA}, grasshoppers exhibit interactive behavior consisting of both attractive and repulsive forces. Mathematically, the resultant force is expressed as $ s(r) =  fe^{\frac{-r}{l}}-e^{-r}$, where $f$ and $l$ is the intensity of attraction and the attractive length scale, respectively. Specifically, the solution update strategy is described as
	\begin{equation}
		\label{Equ:MODOA_update}
		x_i^{d} = c \left(\sum_{j = 1, j \neq i}^{N_{pop}} c \frac{ub_{d} - lb_{d}}{2}s\left(\left|x_j^d - x_i^d\right|\frac{x_j - x_i}{d_{ij}}\right) \right) + T_d,
	\end{equation}
	
	\noindent where $x_i^d$ is the $d$th dimension of the $i$th grasshopper, $N_{pop}$ represents the population size, $d_{ij}$ denotes the distance between the $i$th and $j$th grasshoppers, and $ub_{d}$ and $lb_d$ are the upper and lower boundaries of $d$th dimension, respectively. Moreover, $T_d$ is the value of the best solution found so far in the $d$th dimension. Notably, $c$ is the linear decreasing coefficient to control the size of a comfort zone, which can be expressed as
	\begin{equation}
		\label{Equ:Parameter c}
		c = c_{max} - iter \times \frac{c_{max} - c_{min}}{iter_{max}},
	\end{equation}
	
	\noindent where $c_{max}$ and $c_{min}$ are the maximum and minimum values, respectively, $iter$ represents the current iteration, and $iter_{max}$ is the maximum number of iterations. In Eq.~\eqref{Equ:MODOA_update}, the inner parameter $c$ decreases the attractive/repulsive forces between grasshoppers proportionally to the iteration, whereas the outer parameter $c$ diminishes the search area around the objective with the increasing of iterations.
	
	\begin{figure}[!t]
		\centering
		\includegraphics[width=0.8\linewidth,scale=1.00]{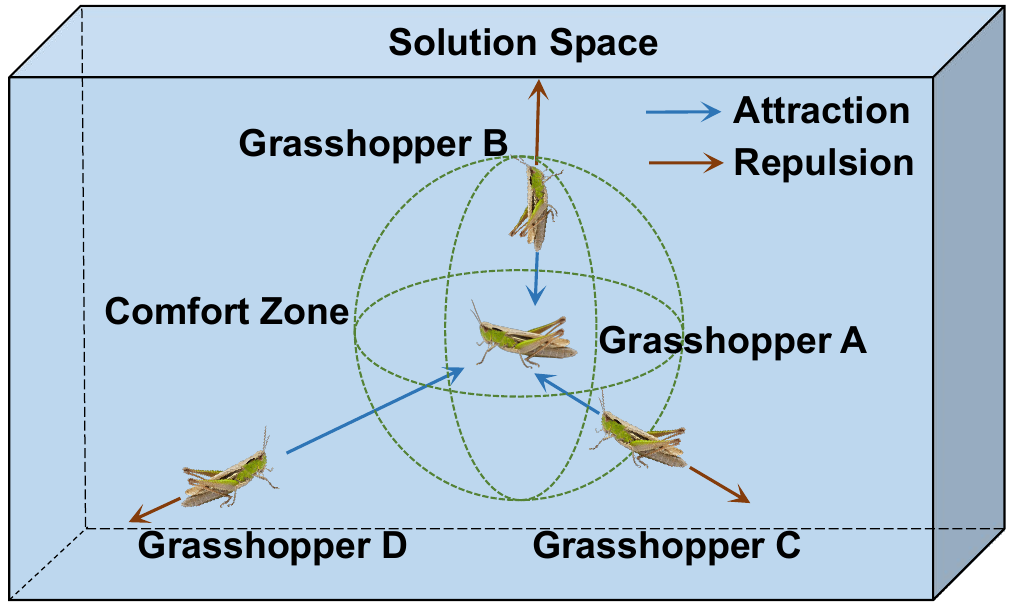}
		\caption{An illustrative example about interaction behavior between grasshoppers, where the resultant force exerted by grasshopper A through two distinct forces on other grasshoppers is categorized into three types as follows. (i) attraction force $>$ repulsion force: Overall effect is attraction, e.g., grasshoppers A to D. (ii) attraction force $=$ repulsion force: Overall effect is neither attraction nor repulsion under the condition of comfort zone, e.g., grasshoppers A to C. (iii) attraction force $<$ repulsion force: Overall effect is repulsion, e.g., grasshoppers A to B.}
		\label{Figure:GOA}
	\end{figure}
	
	\par As illustrated in Fig.~\ref{Figure:MOGOA}, the MOGOA employs archive, crowded neighborhood, and roulette wheel selection to address multi-objective optimization problems in an effective manner. These adaptations allow for better management of the search space, resulting in more optimal solutions. Specifically, the archive stores the best non-domination solutions found so far, while the crowded neighborhood prevents overcrowding of these solutions. Moreover, the roulette wheel selection is used to probabilistically select the target grasshopper for update population, ensuring diversity and convergence in the population. 
	
	\begin{figure*}[htbp]
		\centering
		\includegraphics[width=\linewidth,scale=1.00]{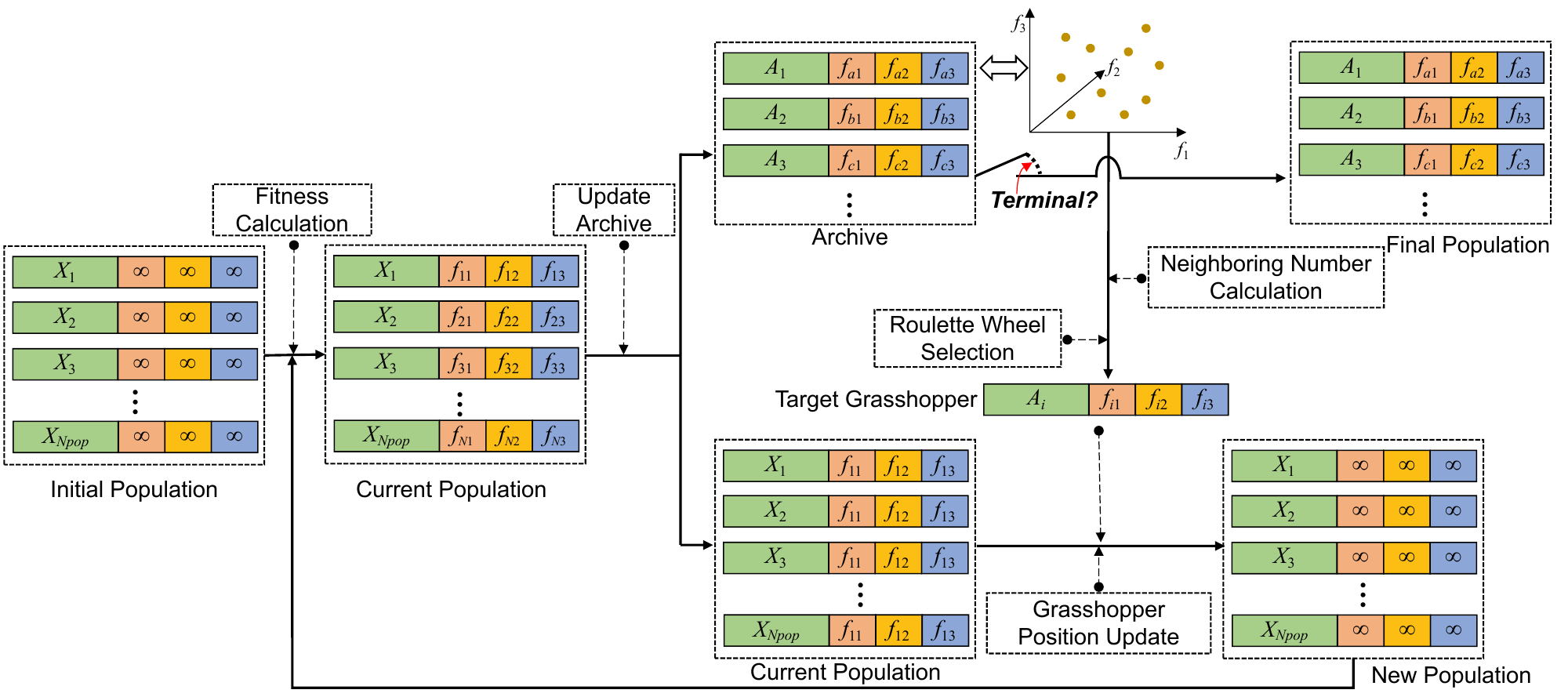}
		\caption{The framework of MOGOA, where the rectangles filled in light green and other colors represent solution and objective values part of population, respectively. Moreover, the dashed arrows with solid dots at the end indicate the operators applied to the populations.}
		\label{Figure:MOGOA}
	\end{figure*}
	
	\par However, the traditional MOGOA faces many challenges for solving the formulated US$^2$RMOP due to the following reasons. 
	
	\begin{itemize}
		\item The mixture of continuous and discrete solution spaces caused by the presence of discrete part ($\mathbb{S}_{\mathrm{P}}$, $\mathbb{O}$) cannot be addressed by conventional MOGOA. 
		\item The probability of finding the global optimal solution is reduced by the random initialization of MOGOA. 
		\item The relationship between exploration and exploitation in a large-scale solution space cannot be efficiently balanced by linear decreasing coefficient $c$. 
	\end{itemize}
	
	\par Thus, we propose the IMOGOA to improve the adaptability of MOGOA for solving the formulated US$^2$RMOP, and the details are as follows.
	
	%
	%	SubSection:IMOGOA
	%
	\subsection{IMOGOA} 
	\label{SubSection:IMOGOA}
	
	\par In this section, IMOGOA with several improvements is presented for solving US$^2$RMOP. Specifically, IMOGOA can achieve a better performance with special designs about the population initialization, solution update and archive update. The general framework of IMOGOA is shown in Algorithm \ref{Algorithm 1}, and the improved strategies are described in detail as follows.
	
	\begin{figure}[!t]
		\removelatexerror
		\begin{algorithm}[H]
			\caption{IMOGOA}
			\label{Algorithm 1}
			\LinesNumbered
			\KwIn{ iteration number $iter_{max}$, population size $N_{pop}$;}
			\KwOut{the final archive $Archive$;}
			$Archive$ $\gets$ $\varnothing$, $P_0$ $\gets$ $\varnothing$;\\
			Initialize Population $P_0$ by using \textbf{H$^3$C Strategy};\\
			Calculate the fitness value of $P_0$ and filtering the non-dominated set $S_{0}$;\\
			$Archive \gets S_0 \cup Archive$;\\
			Update $Archive$ using \textbf{Algorithm 3};\\
			\For{$i = 1$ to $iter_{max}$}
			{  
				Select a grasshopper in $Archive$ by roulette wheel as target grasshopper position $\mathbb{X}_{target}$;\\
				\For{$j = 1$ to $N_{pop}$}
				{  
					Update the solution of the $j$th grasshopper using \textbf{Algorithm 2};\\
				}
				Calculate the fitness value of current population;\\
				Update $Archive$ using \textbf{Algorithm 3};\\
			}	
			Return $Archive$;
		\end{algorithm}
	\end{figure}
	
	%%	SubSubSection:Population Initialization
	\subsubsection{Population Initialization} 
	\label{SubSubSection:Population Initialization}
	
	\par Random initialization of the population in the traditional MOGOA may reduce the probability of finding the global optima. In response to this, we adopt a half-Halton-half-chaos (H$^{3}$C) strategy to enhance diversity of initial population, and the generation method can be described as follow:
	\begin{equation}
		\label{Equ:Population intialization}
		x_{i}^{d} = \left\{
		\begin{array}{ll}
			h_{i,d}\times (ub_{d} - lb_{d}) + lb_{d}, & i \leq N_{pop} / 2 \\
			c_{i,d}\times(ub_{d} - lb_{d}) + lb_{d}, & otherwise
		\end{array}
		\right 
		.,
	\end{equation}
	
	\noindent where $h_{i,d}$ and $c_{i,d}$ are the sequences between 0 and 1 which are generated by the Halton method and chaotic map, respectively. Specifically, the calculation process for Halton sequence can be expressed as 
	\begin{equation}
		\label{Equ:Halton Sequence2}
		h_{i,d}=\left(\sum_{l=0}^{L}b_{l}(i)(p_{d})^{-l-1}\right),
	\end{equation}
	
	\noindent where $p_{d}$ is a random prime number in the range of 0 to 10, and the value $b_{0}(i), b_{1}(i), \dots, b_{l}(i)$ needs to satisfy the condition, i.e., $ \sum_{l=0}^{L}b_{l}(i)(p_{d})^{l} = i$. Accordingly, the calculation process for a chaotic sequence can be expressed as
	\begin{equation}
		\label{Equ:Chaos}
		c_{i,d}=\text{mod}\left(c_{i, d - 1} + b - \left(a/2\pi\right) \times \text{sin}(2\pi c_{i,d-1}), 1\right),
	\end{equation}
	
	\noindent where $a$ and $b$ are two constant variables. Moreover, $\mod(\cdot)$ is a mathematical division operator.
	
	%%	SubSubSection:Non-linear Decreasing Coefficient
	\subsubsection{Non-linear Decreasing Coefficient} 
	\label{SubSubSection:Non-linear Decreasing Coefficient}
	
	\par To increase the probability of finding the globally optimal solution, IMOGOA must balance the process of exploring new solution space and exploiting known information. Relying too heavily on known information may trap the algorithm in a local optimal solution and prevent it from finding the global optimal solution. Conversely, relying too much on random searches may cause the algorithm to waste time exploring low-quality solution spaces. To address this issue, a non-linear decreasing coefficient is adopted and expressed as
	\begin{equation}
		\label{Equ:Non-linear decreasing coefficient}
		c = c_{max} - c_{min} - \text{sin}\left(\frac{1}{2} \times \pi \times \left(\frac{iter}{iter_{max}}\right)^{\frac{1}{2}}\right).
	\end{equation}
	
	\begin{figure}[!h]
		\includegraphics[width=\linewidth,scale=1.0]{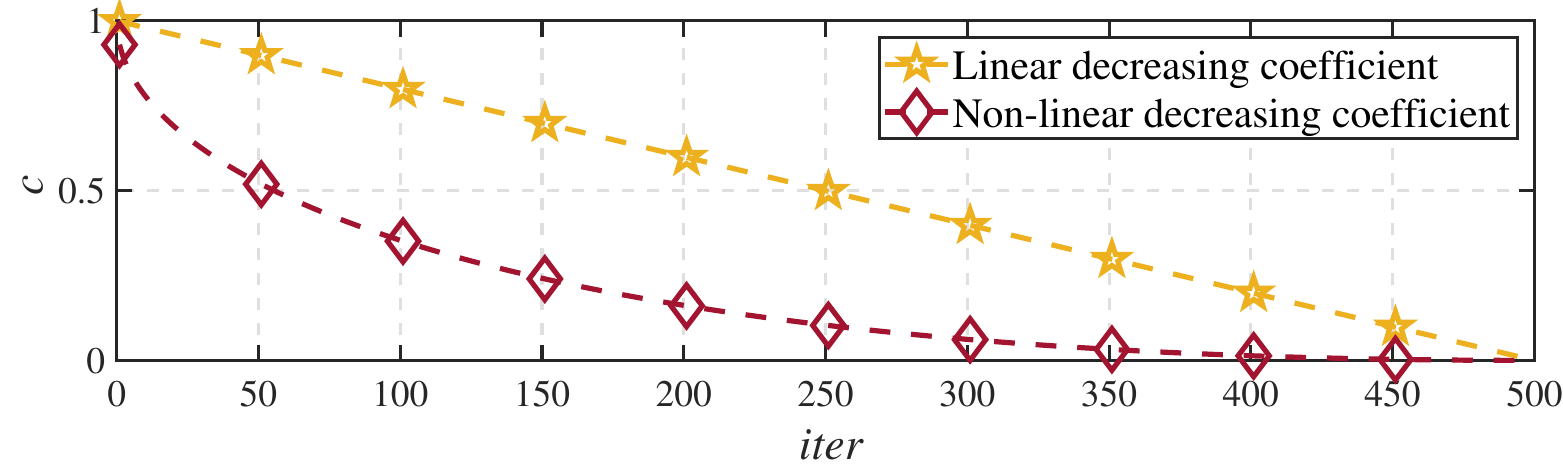}
		\caption{Comparison between non-linear and linear decreasing coefficients. We set $c_{max}=1$ and $c_{min}=0.0004$.}
		\label{Figure:Non-linear decreasing coefficient}
	\end{figure}
	
	\par The non-linear decreasing coefficient, as illustrated in Fig.~\ref{Figure:Non-linear decreasing coefficient}, enables the algorithm to prioritize exploring unknown solution space in the early stages of the search. Over time, the algorithm's ability to exploit known information is gradually strengthened, leading to better balance and higher quality solutions. By enabling a better balance between exploring new solution spaces and exploiting known information, this approach helps IMOGOA find high-quality solutions in a shorter timeframe.
	
	%%	SubSubSection:Solution Update
	\subsubsection{Solution Update} 
	\label{SubSubSection:Solution update}
	
	\par The solution update of population is an important step for efficiently searching the solution space in IMOGOA. Since the formulated US$^2$RMOP has both continuous solution part $(\mathbb{I}_{\mathrm{P}},\mathbb{I}_{\mathrm{V}},\mathbb{P}_{\mathrm{C}})$ and discrete solution part $(\mathbb{S}_{\mathrm{P}}$, $\mathbb{O})$, it is necessary to consider these two update operators separately. 
	
	\par \textbf{\textit{For the continuous part of solution:}} Although the solution can be updated directly by using conventional MOGOA, it can easily become trapped in local optima. To improve the search capability in the large-scale solution space, we employ L\'evy flight in the process of solution update. Specifically, the continuous solution update method can be expressed as
	\begin{equation}
		\label{Equ:Continuous solution update}
		\mathbb{X}_{i}^{CNew} = \mathbb{X}_{i}^{CGOA} + \alpha_{1} \otimes L\acute{e}vy(\beta),
	\end{equation}
	
	\noindent where $\mathbb{X}_{i}^{CGOA}$ is the continuous solution obtained by the traditional MOGOA after introducing the nonlinear decreasing coefficient, and $\alpha_{1}$ is the step size scaling factor. Moreover, $\otimes$ is a mathematical
	Hadamard product operator, and $L\acute{e}vy(\beta)$ is the L\'evy random search path, which is calculated as follow:
	\begin{equation}
		\label{Equ:Levy Flight}
		L\acute{e}vy(\beta) = \frac{\mu}{\lvert w \rvert^{-\beta}} ,
	\end{equation}
	
	\noindent where $\mu$ is the normal distribution matrix with mean 0 and variance $\sigma_{u}^{2}$, wherein $\sigma_{u} = \left[\frac{\Gamma(1+\beta) \sin \left(\frac{\pi \beta}{2}\right)}{\Gamma\left(\frac{1+\beta}{2}\right) \beta \times 2^{\frac{\beta-1}{2}}}\right]^{\frac{1}{\beta}}$. Besides, $w$ is the normal distribution matrix with mean 0 and variance $1$.
	
	\par \textbf{\textit{For the discrete part of solution:}} Since the conventional MOGOA cannot handle discrete solution spaces, we design the update process in detail because of the presence of discrete part $(\mathbb{S}_{\mathrm{P}}$, $\mathbb{O})$. Specifically, since the selection order of the UAV receiver denoted by $\mathbb{S}_{\mathrm{P}}$ can be duplicated and the user association order of IoT terminal devices denoted by $\mathbb{O}$ is non-duplicated, the two-points crossover (TPC) and partially-matched crossover (PMX) \cite{Goldberg2014} strategies are adopted for $\mathbb{S}_{\mathrm{P}}$ and $\mathbb{O}$, respectively. As can be seen from Fig. \ref{Figure:TPC and PMX strategy.}, the PMX strategy appends conflict remove part for eliminating the duplicate elements compared to the TPC strategy. Accordingly, the detailed update is expressed as follow:
	\begin{equation}
		\label{Equ:Discrete Solution Update}
		\left\{\begin{array}{ll}
			\left[\mathbb{X}_{i}^{SNew1}, \mathbb{X}_{i}^{SNew2}\right] = \text{TPC}\left(\mathbb{X}_{i}^{SCur}, \mathbb{X}_{TS}\right),& for \  \mathbb{S}_{\mathrm{P}} \\
			\left[\mathbb{X}_{i}^{ONew1}, \mathbb{X}_{i}^{ONew2}\right] = \text{PMX}\left(\mathbb{X}_{i}^{OCur}, \mathbb{X}_{TO}\right), & for \  \mathbb{O}
		\end{array}\right.,
	\end{equation}
	
	\noindent where $\left(\mathbb{X}_{i}^{SCur},\mathbb{X}_{i}^{OCur} \right)$ and $\left( \mathbb{X}_{TS}, \mathbb{X}_{TO}\right)$ are the discrete parts of the $i$th grasshopper and target grasshopper, respectively. Moreover, $\left(\mathbb{X}_{i}^{SNew1}, \mathbb{X}_{i}^{ONew1}\right)$ and $\left(\mathbb{X}_{i}^{SNew2}, \mathbb{X}_{i}^{ONew2}\right)$ represent the two new offspring generated by the TPC and PMX strategies, respectively.
	
	\par Accordingly, the update strategy of the solution is detailed in \textbf{Algorithm 2}.
	
	\begin{figure}[!t]
		\includegraphics[width=\linewidth,scale=1.00]{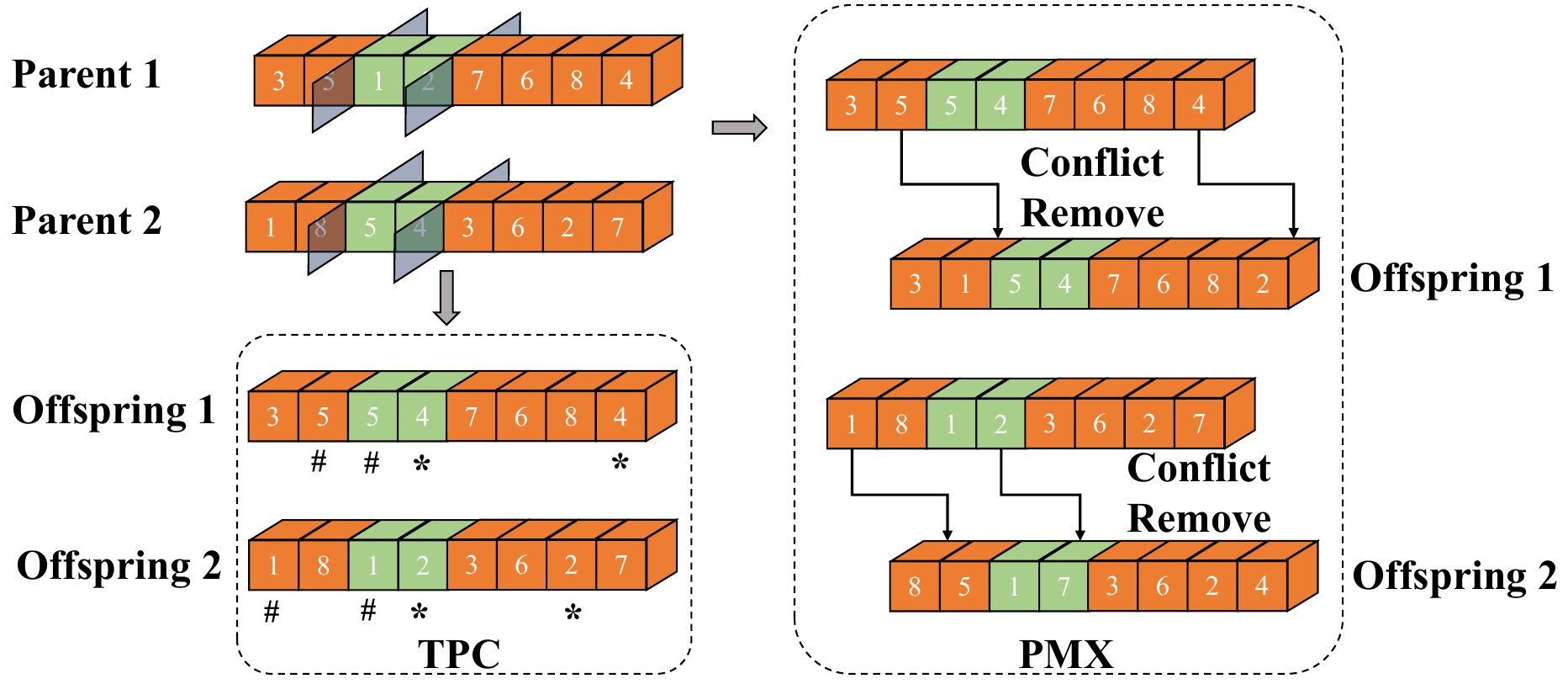}
		\caption{An illustrative example of TPC and PMX strategies, wherein TPC strategy selects randomly two crossover points to execute the exchange of elements, pounds and stars point out the repeated location after performing the TPC strategy. On the basis of TPC, PMX strategy further performs conflict remove operation.}
		\label{Figure:TPC and PMX strategy.}
	\end{figure}
	
	\begin{figure}[!t]
		\removelatexerror
		\begin{algorithm}[H]
			\caption{Solution Update}
			\label{Algorithm 2}
			\KwIn{the current solution of the $i$th grasshopper $\mathbb{X}_{i}$, the current solution of target grasshopper $\mathbb{X}_{target}$;}
			\KwOut{updated solution of the $i$th grasshopper $\mathbb{X}_{New}$;}
			\LinesNumbered
			\tcc{\textbf{Continuous Part Update}}
			Using Eq.~\eqref{Equ:Continuous solution update} update continuous part $\mathbb{X}^{CCur}_{i}=(\mathbb{I}^{Cur}_{\mathrm{P}},\mathbb{I}^{Cur}_{\mathrm{V}},\mathbb{P}^{Cur}_{\mathrm{C}})$ to produce new continuous part $\mathbb{X}^{CNew}_{i}=(\mathbb{I}^{New}_{\mathrm{P}},\mathbb{I}^{New}_{\mathrm{V}},\mathbb{P}^{New}_{\mathrm{C}})$; \\
			\tcc{\textbf{Discrete Part Update}}
			Using Eq.~\eqref{Equ:Discrete Solution Update} update discrete part $\mathbb{X}_{i}^{DCur}=(\mathbb{S}_{\mathrm{P}}^{Cur}$, $\mathbb{O}^{Cur})$ to produce two new discrete offspring parts $(\mathbb{S}_{\mathrm{P}}^{New1}$, $\mathbb{O}^{New1})$ and $(\mathbb{S}_{\mathrm{P}}^{New2}$, $\mathbb{O}^{New2})$; \\
			\tcc{\textbf{Merge and Selection}}
			\textbf{Merge}: Obtain two grasshopper offspring $O_{1}=\left(\mathbb{X}^{CNew}_{i},\mathbb{S}_{\mathrm{P}}^{New1}, \mathbb{O}^{New1} \right)$ and $O_{2}=\left(\mathbb{X}^{CNew}_{i},\mathbb{S}_{\mathrm{P}}^{New2}, \mathbb{O}^{New2} \right)$ by merging continuous and discrete part of solution;\\
			\textbf{Selection}: Calculate the fitness of $O_{1}$ and $O_{2}$;\\
			\eIf{$O_{1}$ dominates $O_{2}$}
			{
				$\mathbb{X}_{New}$ $\gets$ $O_{1}$;
			}
			{
				\eIf{$O_{2}$ dominates $O_{1}$}
				{
					$\mathbb{X}_{New}$ $\gets$ $O_{2}$;
				}
				{
					\eIf{$rand < 0.5$}{
						$\mathbb{X}_{New}$ $\gets$ $O_{1}$;
					}{
						$\mathbb{X}_{New}$ $\gets$ $O_{2}$;
					}
				}            
			}
			Return $\mathbb{X}_{New}$;
		\end{algorithm}
	\end{figure}
	
	\begin{figure}[!t]
		\removelatexerror
		\begin{algorithm}[H]
			\caption{Archive Update}
			\label{Algorithm 3}
			\KwIn{the updated solution of current grasshopper population $\left\{\mathbb{X}_{1},\mathbb{X}_{2},\cdots, \mathbb{X}_{Npop}\right\}$, the current Archive $A_{current}$, the current Archive size $N_{A}$;}
			\KwOut{the updated Archive $A_{updated}$;}
			\LinesNumbered
			$A_{updated}$ $\gets$ $A_{current}$;\\
			\tcc{\textbf{Archive Mutation}}
			\For{$i = 1$ to $N_{A}$}
			{  
				\tcc{\textbf{Continuous Part Mutation}}
				Using Eq.~\eqref{Equ:Cauchy} mutate continuous part $\mathbb{X}^{CCur}_{i}=(\mathbb{I}^{Cur}_{\mathrm{P}},\mathbb{I}^{Cur}_{\mathrm{V}},\mathbb{P}^{Cur}_{\mathrm{C}})$ to produce new continuous part $\mathbb{X}^{CM}_{i}=(\mathbb{I}^{New}_{\mathrm{P}},\mathbb{I}^{New}_{\mathrm{V}},\mathbb{P}^{New}_{\mathrm{C}})$; \\
				\tcc{\textbf{Discrete Part Crossover}}
				Random select the discrete part of a grasshopper in $A_{current}$ as $(\mathbb{X}_{TS},\mathbb{X}_{TO})$ \\
				Using Eq.~\eqref{Equ:Discrete Solution Update} crossover discrete part $\mathbb{X}_{i}^{DCur}=(\mathbb{S}_{\mathrm{P}}^{Cur}$, $\mathbb{O}^{Cur})$ to produce two new discrete offspring parts $(\mathbb{S}_{\mathrm{P}}^{New1}$, $\mathbb{O}^{New1})$ and $(\mathbb{S}_{\mathrm{P}}^{New2}$, $\mathbb{O}^{New2})$; \\
				\tcc{\textbf{Merge and Selection}}
				\textbf{Merge}: Obtain two grasshopper offspring $M_{1}=\left(\mathbb{X}^{CM}_{i},\mathbb{S}_{\mathrm{P}}^{New1}, \mathbb{O}^{new1} \right)$ and $M_{2}=\left(\mathbb{X}^{CM}_{i},\mathbb{S}_{\mathrm{P}}^{New2}, \mathbb{O}^{New2} \right)$ by merging continuous and discrete part of solution;\\
				\textbf{Selection}: Calculate the fitness of $M_{1}$ and $M_{2}$;\\
				\eIf{$M_{1}$ dominates $M_{2}$}
				{
					$A_{updated}$ $\gets$ $A_{updated} \cup M_{1}$;
				}
				{
					\eIf{$M_{2}$ dominates $M_{1}$}
					{
						$A_{updated}$ $\gets$ $A_{updated} \cup M_{2}$;
					}
					{
						\eIf{$rand < 0.5$}{
							$A_{updated}$ $\gets$ $A_{updated} \cup M_{1}$;
						}{
							$A_{updated}$ $\gets$ $A_{updated} \cup M_{2}$;
						}
					}            
				}
			}
			\tcc{\textbf{Archive Determination}}
			Remove dominated solutions in $A_{updated}$;\\
			\If{$size(A_{updated}) \geq maxArchiveSize$}
			{
				Execute DCDE on $A_{updated}$;      
			}
			Return $A_{updated}$; 
			
		\end{algorithm}
	\end{figure}
	
	%%	SubSubSection:Archive Update
	\subsubsection{Archive Update} 
	\label{SubSubSection:Archive Update}
	
	\par A single update operator cannot simultaneously balance global exploration and local exploitation capabilities. At the same time, the archive has an impact on guiding the update of the grasshopper population. Thus, we introduce archive mutation and dynamic elimination-based crowding distance (DCDE) \cite{Zhao2022} to improve global searchability and solution distribution in archive. In the archive mutation stage, we employ mutation and crossover for the continuous and discrete parts, respectively. The mutation for the continuous part is described as
	\begin{equation}
		\label{Equ:Cauchy}
		\mathbb{X}_{i}^{CM} = \mathbb{X}^{C}_{i} + \alpha_{2} \otimes Cauchy(0,1),
	\end{equation}
	
	\noindent where $\mathbb{X}^{C}_{i}$ is continuous part of the $i$th solution in current archive, and $\mathbb{X}_{i}^{CM}$ represents continuous part the $i$th solution after Cauchy mutation. Moreover, $\alpha_{2}$ is the step size scaling factor, and $Cauchy(0,1)$ is the standard Cauchy distribution. Similar to the discrete solution update method, we simply treat a random grasshopper position in the archive as the target grasshopper location $\mathbb{X}_{TS}$ and $\mathbb{X}_{TO}$ in Eq.~\eqref{Equ:Discrete Solution Update}. In the archive determination stage, we replaced the traditional MOGOA archive elimination method with DCDE. The purpose of this modification is to more effectively maintain the diverse and uniform distributions of IMOGOA archive on the Pareto front. Specifically, the difference between DCDE and the traditional elimination-based crowding distance strategy is shown in Fig. \ref{Figure:DCDE}. A more detailed explanation about DCDE can be found in the \cite{Zhao2022}.
	
	\begin{figure*}[htbp]
		\centering
		\includegraphics[width=\linewidth,scale=1.00]{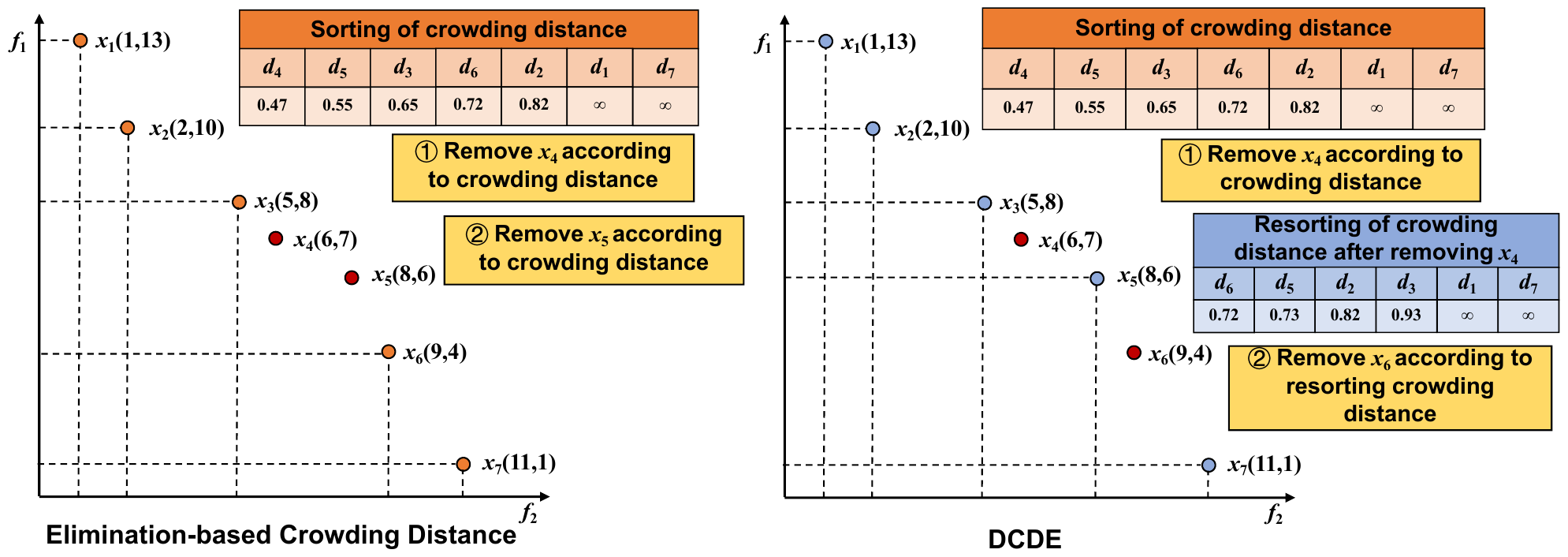}
		\caption{A comparison sketch between elimination-based crowding distance and DCDE: Elimination-based crowding distance only sorts once and then deletes the archive regardless of how many individuals it overflows, while DCDE adopts the strategy of sorting once and deleting once, which better achieves the uniformity of archive.}
		\label{Figure:DCDE}
	\end{figure*}
	
	%
	%	SubSection:Scheduling Mechanism of the IMOGOA
	%
	{\color{color}
		\subsection{Scheduling Mechanism of the IMOGOA} 
		\label{SubSection:Operating Mechanism of the IMOGOA}
		\par A simple and efficient scheduling mechanism is necessary for implementing the proposed IMOGOA in the UAV swarm-enabled collaborative secure relay communication system. Supposing that UAVs can gather the initial state information during the start stage of communication. Consider the limited computing and energy resources, to run IMOGOA on a single UAV may waste too much time and bring extra energy consumption. Due to the sufficient ability in terms of energy and computing, the MBS can be considered as a supercomputer to run IMOGOA. The main steps of the scheduling mechanism of IMOGOA are as follows.
		\par \textbf{\textit{Step 1 - Information Fusion} (at UAV Swarm):} Select UAV $U_{k}$ to broadcast the start message $M_{Start}$, where the position and network information of UAV $U_{k}$ is included. After UAV $U_{k^{-}}$ receives the message $M_{Start}$, it replies a confirm message $M_{ACK1}$, which contains its own the position and network information. Subsequently, UAV $U_{k}$ aggregates the position and network information UAV $U_{k}$ aggregates the position and network information, and transmits the fusion information $M_{Fusion}$ to MBS.
		\par \textbf{\textit{Step 2 - Optimization Algorithm Execution} (at MBS):} MBS receives the message $M_{Fusion}$ and replies a confirm message $M_{ACK2}$. Then, the supercomputer of MBS runs the proposed IMOGOA to produce the optimized solution for US$^2$RMOP. Subsequently, MBS sends the solution distribution message $M_{Solution}$ to UAV $U_{k}$, which contains the optimized solution for US$^2$RMOP.
		\par \textbf{\textit{Step 3 - Optimization Information Distribution} (at UAV Swarm):} The UAV $U_{k}$ receives the message $M_{Solution}$, and distributes optimization information to other UAVs by using the wireless channel allocation \cite{Dai2022}. Until all other UAVs reply to the message $M_{ACK3}$, the scheduling mechanism of IMOGOA is terminated.
		\par Note that the UAV $U_{k}$ will resent the message $M_{Start}$, $M_{Fusion}$ or $M_{Solution}$ every $T_{waiting}$ until it receives $M_{ACK1}$, $M_{ACK2}$ or $M_{ACK3}$ to ensure the reliability of scheduling mechanism.}

	%
	%	SubSection:Computational Complexity of the IMOGOA
	%
	\subsection{Algorithm Analysis} 
	\label{SubSection:Algorithm Analysis}
	\par In this section, the scheduling overhead and computational complexity of the proposed IMOGOA are analyzed. 
	\subsubsection{Scheduling Overhead Analysis}
	\label{SubSubSection:Scheduling Overhead Analysis}
	\par In this analysis, we assume that the maximum number of re-transmissions is $N_{re}$ for the messages $M_{Start}$, $M_{Fusion}$ and $M_{Solution}$. The bit numbers of messages $M_{Start}$, $M_{Fusion}$, $M_{Solution}$, $M_{ACK1}$ and $M_{ACK3}$ are $b_{S}$, $b_{F}$, $b_{O}$, $b_{A1} $ and $b_{A3}$, respectively. Moreover, the packet loss probabilities of links are $f$, and the transmission rate and power of the UAV are $r$ and $P_{T}$, respectively.
	\begin{proposition}
		The communication energy consumption of the UAV swarm for the scheduling mechanism of IMOGOA is
		\begin{small}
			\begin{equation}
				E_{C}  = \frac{P_{T} [(b_{S} + b_{F} + b_{O})(1 - f^{N_{re}}) + (K - 1)(b_{A1} + b_{A3})(1 - f)]}{r(1 - f)}.
			\end{equation}
		\end{small}
	\end{proposition}
	\begin{proof}
		The average re-transmission number of $M_{Start}$, $M_{Fusion}$ and $M_{Solution}$ can be expressed as
		\begin{equation}
			\begin{split}
				\begin{aligned}
					\overline{N}_{T} &= 1 \cdot \Pr(S^{1}) + 2 \cdot \Pr(F^{1}, S^{2}) + \cdots \\
					& \quad + (N_{re}-1) \cdot \Pr(F^{1}, F^{2}, \cdots, S^{N_{re}-1})\\
					& \quad + N_{re} \cdot \Pr(F^{1}, F^{2}, \cdots, F^{N_{re}-1}),
				\end{aligned}
			\end{split}
			\label{Equ:Communication Energy Consumption}
		\end{equation}
		\noindent where the successful probability of transmission at the $m$ round is denoted as $\Pr(F^{1}, F^{2}, \cdots, F^{m-1}, S^{m})$. Since the successful and failed probability of transmission are independent for each round, the average re-transmission number of $M_{Start}$, $M_{Fusion}$ and $M_{Solution}$ can further be written as
		\begin{equation}
			\begin{split}
				\begin{aligned}
					\overline{N}_{T} &= (1-f) + 2(f-f^{2}) + \cdots \\
					& \quad + (N_{re}-1) (f^{N_{re}-2} - f^{N_{re}-1}) + N_{re}f^{N_{re}-1} \\ 
					& = \frac{1-f^{N_{re}}}{1-f}.
				\end{aligned}
			\end{split}
			\label{Equ:Retransmission Number}
		\end{equation}
		\par Thus, the communication energy consumption of the UAV swarm for the scheduling mechanism of IMOGOA can be calculated as follow:
		\begin{equation}
			\begin{aligned}
				E_{C} &= E_{Step1} + E_{Step3} \\
				& = P_{T} \cdot [(b_{S} + b_{F}) \cdot \overline{N}_{T} + (K - 1) \cdot b_{A1}] / r\\
				& \quad + P_{t} \cdot [b_{O} \cdot \overline{N}_{T} +  (K - 1) \cdot b_{A3}] / r\\
				& = \frac{P_{T} \cdot [(b_{S} + b_{F} + b_{O}) \cdot \overline{N}_{T} +  (K - 1) \cdot (b_{A1} + b_{A3})]}{r}.
			\end{aligned}
			\label{Equ:Communication Energy Consumption_Middle}
		\end{equation}
		\par Substituting Eq.~\eqref{Equ:Retransmission Number} into Eq.~\eqref{Equ:Communication Energy Consumption_Middle}, the scheduling energy consumption can be re-written as Eq.~\eqref{Equ:Communication Energy Consumption}.
	\end{proof}
	
	\par The communication energy consumption of the UAV swarm for the scheduling mechanism of IMOGOA is a small enough energy overhead compared to the motion energy consumption we optimized previously for the UAV swarm. For instance, we consider a scenario where 16 UAVs serve 8 remote IoT terminal devices. Moreover, the MBS is equipped with a PAA with $6 \times 6$ array elements. The maximum number of re-transmissions is set to 3. The common format of the scheduling messages is expressed as $[Src, Dst, Data]$, where $Src$ and $Dst$ represent the source and destination addresses of the message, respectively, each of which occupies $4$ Bytes. Moreover, $Data$ indicates the specific content of message. For the five types of messages in the scheduling process, the specific data of messages are as follows. 
	\begin{itemize}
		\item $M_{Start}$ and $M_{ACK1}$ contain information about the positions of the UAV and the ground eavesdropper, which occupies $5 \times 4$ Bytes (if the ground eavesdropper cannot be detected by this UAV, $Data$ is padded with $0$).
		\item $M_{Fusion}$ contains information about the positions of each UAV and the ground eavesdropper, which occupies $(3 \times K + 2) \times 4$ Bytes.
		\item  $M_{Solution}$ contains information about the solution obtained by MBS, which occupies $\left(M \times N + 2 + 4 \times K\right) \times T \times 4$ Bytes.
		\item $M_{ACK3}$ contains information about the acknowledgment information, which occupies $1 \times 4$ Bytes.
	\end{itemize}
	\par Accordingly, $b_{S}$, $b_{F}$, $b_{O}$, $b_{A1}$ and $b_{A3}$ can be approximately computed as $28$ Bytes, $208$ Bytes, $3272$ Bytes, $28$ Bytes and $9$ Bytes, respectively. We assume that the transmission rate and power of the UAV are $1$ Mbps and $0.1$ W, respectively. As a result, the communication energy consumption of the UAV swarm for the scheduling mechanism of IMOGOA is $0.0034$ J, which is significantly less than the motion energy consumption.

	\subsubsection{Computational Complexity Analysis} 
	\par The computational complexity of IMOGOA is primarily dependent on its population size $N_{pop}$, the number of objectives $n$ and the maximum number of iterations $iter_{max}$. For simplicity, we assume that the size of the external archive is equal to the population size. The following operators represent worst-case computational complexity:
	
	\begin{itemize}
		\item{\textit{Solution Update}: Updating the continuous and discrete components accounts for the computational complexity of IMOGOA, with complexities of $O(iter_{max} \times N_{pop}^{2})$ and $O(iter_{max}\times N_{pop})$, respectively.}
		\item{\textit{Archive Update}: The archive update includes archive mutation and DCDE with complexities of $O(iter_{max} \times 2nN_{pop})$ and $O((iter_{max} \times nN_{pop}^2
			\text{log}N_{pop})$, respectively.}
	\end{itemize}
	
	\par Therefore, the overall computational complexity of IMOGOA is $O(iter_{max} \times nN_{pop}^{2}\text{log}N_{pop})$.
	
	%%%
	%	Section:Simulation Results
	%%%
	\section{Simulation Results} 
	\label{Section:Simulation Results}
	
	\par In this part, we perform simulations to verify the effectiveness and efficiency of the proposed strategy and algorithm.
	
	%
	%	SubSection:Simulation Setups
	%
	\subsection{Simulation Setups} 
	\label{SubSection:Simulation Setups}
	
	\begin{figure*}[!t]
		\centering
		\subfloat[]{
			\label{Figure:OptimizedPAARate}
			\includegraphics[width=.32\linewidth]{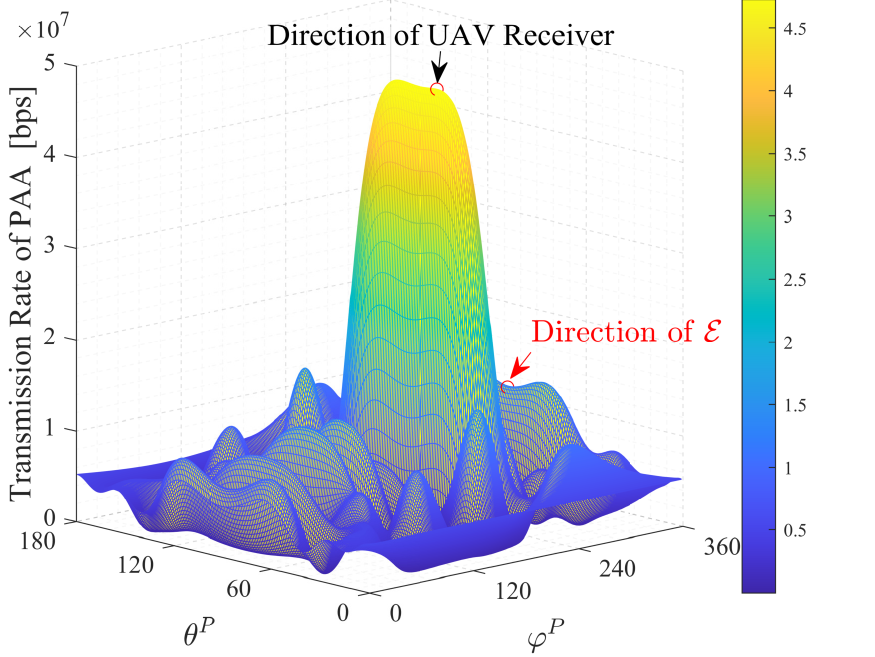}}
		\subfloat[]{
			\label{Figure:OptimizedCBRate}
			\includegraphics[width=.32\linewidth]{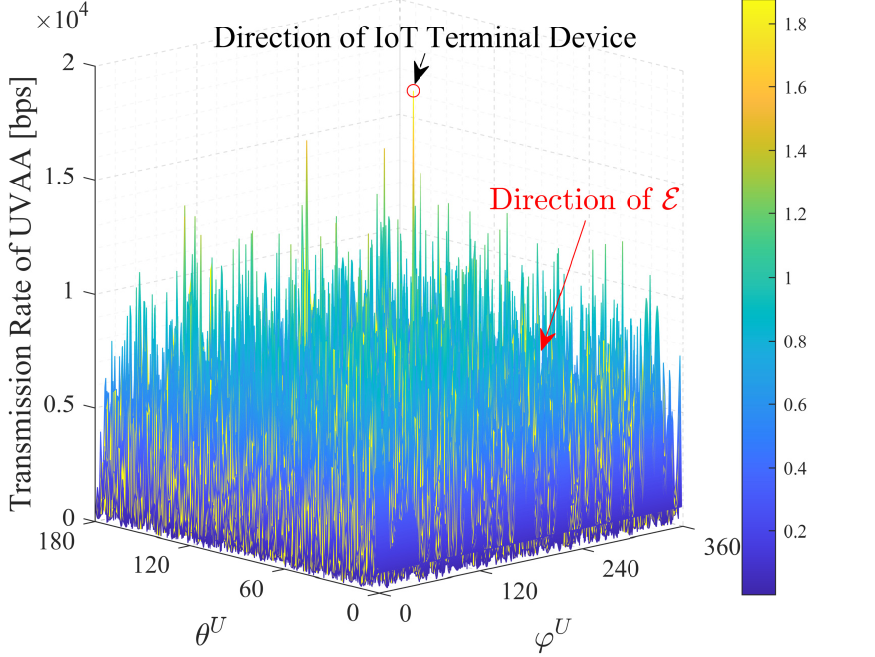}}
		\subfloat[]{
			\label{Figure:OptimizedUAVTrajectory}
			\includegraphics[width=.32\linewidth]{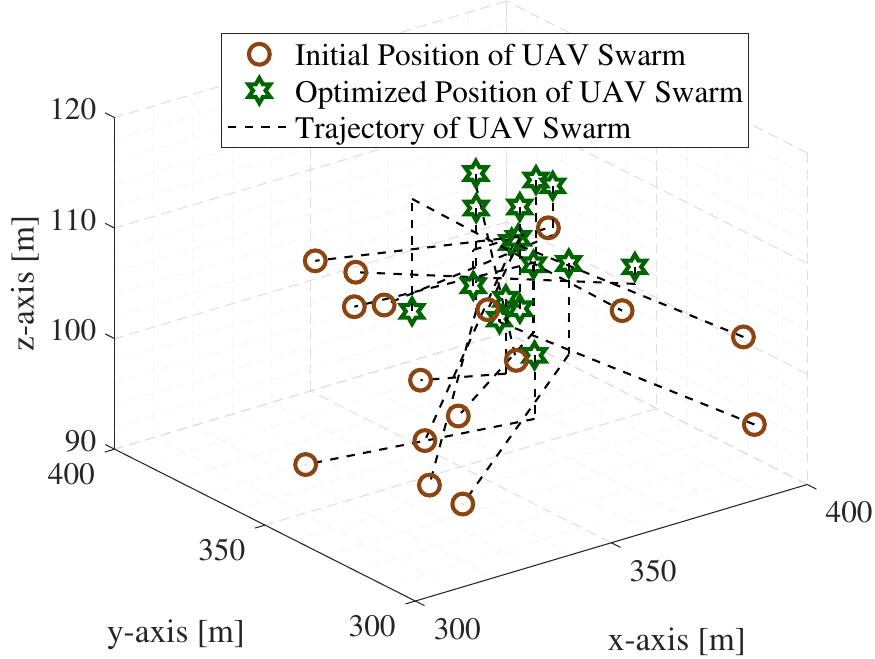}}
		\caption{Optimization results obtained by IMOGOA. (a) The optimized transmission rate distribution of PAA. (b) The optimized transmission rate distribution of UVAA. (c) The optimized trajectory of UAV swarm.}
		\label{Figure:Optimization results}
	\end{figure*}
	
	\par We consider 8 IoT terminal devices which are with the queued state for service associated with $\mathcal{S}$. The $\mathcal{S}$ is equipped with a PAA with $6 \times 6$ elements and its transmit power is set to $3.6$ W. The UAV swarm has 16 UAV individuals that can be operated and the transmit power of each UAV is set to be $0.1$ W. Additionally, other parameters about channels and UAVs are shown in Table \ref{Table:Other Simulation Parameter Settings}. Moreover, the UAV swarm is distributed in a $100$ m $\times$ $100$ m area. This study selects the carrier frequency $2.4$ GHz as it was early exploited by Wi-Fi technology and is still widely supported by many IoT terminal devices.
	
	\par On the one hand, the effectiveness of our strategy is verified by comparing it with other relay strategies. On the other hand, several multi-objective optimization algorithms, which include NSGA-II \cite{Deb2002}, multi-objective particle swarm optimization (MOPSO) \cite{Coello2002}, multi-objective grey wolf optimizer (MOGWO) \cite{Mirjalili2016}, conventional MOGOA and multi-objective multi-verse optimization (MOMVO) \cite{Mirjalili2017}, are employed to solve the US$^2$RMOP to identify the efficiency of our proposed IMOGOA.
	
	%
	%	SubSection:Metric
	%
	\subsection{Performance Metric} 
	\label{SubSection:Metric}
	
	\par For the performance of multi-objective optimization algorithms, the convergence and diversity are two critical dimensions. Hypervolume (HV), as a comprehensive evaluation metric has been used frequently to evaluate the comprehensive performance of multi-objective optimization algorithms \cite{Bader2011}, \cite{Zhang2023}. Specifically, HV is defined as the hypervolume between the estimated Pareto front $\mathcal{F}$ and the reference vector $r$, which can be expressed as
	\begin{equation}
		\label{Equ:HV}
		HV(\mathcal{F},r) = \mathscr{L}\left(\bigcup_{0 \leq i \leq F}[f_{1}(i), r_{1}] \times \cdots \times [f_{n}(i), r_{n}]\right),
	\end{equation}
	
	\noindent where the reference vector is defined as $r = [r_{1}, \cdots, r_{n}]$ , $n$ is the number of objectives, and the size of $\mathcal{F}$ is represented as $F$. Moreover, $\left[\begin{matrix}f_1\left(i\right),r_1\end{matrix}\right]\times\cdots\times\left[\begin{matrix}f_n\left(i\right),r_n\end{matrix}\right]$ represents the hypercubes by all points that are dominated by solution $i$ but not dominated by reference points. Note that $\mathscr{L}(*)$ is Lebesgue measure, which inscribes the hypervolume of all objectives' hypercubes. Specifically, $HV(\mathcal{F},r)$ is calculated by the Monte Carlo estimation proposed in \cite{Bader2011}. 
	
	\begin{table}
		\renewcommand\arraystretch{1.2}
		\setlength\tabcolsep{3pt}
		\centering 
		\caption{Other Simulation Parameter Settings}  
		\label{Table:Other Simulation Parameter Settings} 
		\scalebox{0.80}{ 
			\begin{tabular}{l l}\toprule[1.5pt]
				{Parameters} & {Values}  \\  \toprule[1pt]
				{Bandwidth} & {$B = 20$ MHz}   \\
				{Parameters of $S$-curve} & {$a = 9.61, b = 0.16$}   \\
				{\color{color}{Average channel power gain at $d_{0}=1$ m for LoS state}} & {\color{color}{$\beta_{0} = -60$ dB}}   \\
				% {Average channel power gain at d0=1d_{0}=1 m for G2G link} & {β1=−90\beta_{1} = -90 dB}  \\
				{Additional signal attenuation factor for NLoS propagation} & {$\mu = -20$ dB}  \\
				{Average path loss exponents for NLoS state} & {$\alpha_{NLoS} = 3.5$}   \\
				{Average path loss exponents for LoS state} & {$\alpha_{LoS} = 2.5$}   \\
				{Average path loss exponents for G2G link} & {$\alpha_{G} = 3.5$}   \\
				{Noise power of channel} & {$\sigma^2 = -174$ dBm/Hz}   \\
				{Mass of a UAV} & {$m = 2$ kg}   \\
				{Tip speed of the rotor blade} & {$u_{tips} = 120$ m/s}   \\
				{Mean rotor-induced velocity for hovering} & {$u_{0} = 4.03$ m/s}   \\
				{Air density} & {$\rho = 1.225$ kg/m$^3$}   \\
				{Rotor disc area} & {$A = 0.053$ m$^3$}   \\
				{Fuselage drag ratio} & {$d_{0} = 0.6$}   \\
				{Rotor solidity} & {$s = 0.05$}   \\
				\toprule[1.5pt]
		\end{tabular}}
	\end{table}
	%
	%	SubSection:Simulation Results
	%
	\subsection{Simulation Results} 
	\label{SubSection:Simulation Results}
	In this section, we first show the visualization results of IMOGOA and compare our strategy with two other relay strategies. Furthermore, several multi-objective optimization algorithms are employed to evaluate the efficiency of IMOGOA. Finally, the performance comparisons of two different situations are analyzed.

	\begin{figure}[htbp]
		\centering
		\includegraphics[width=\linewidth,scale=1.00]{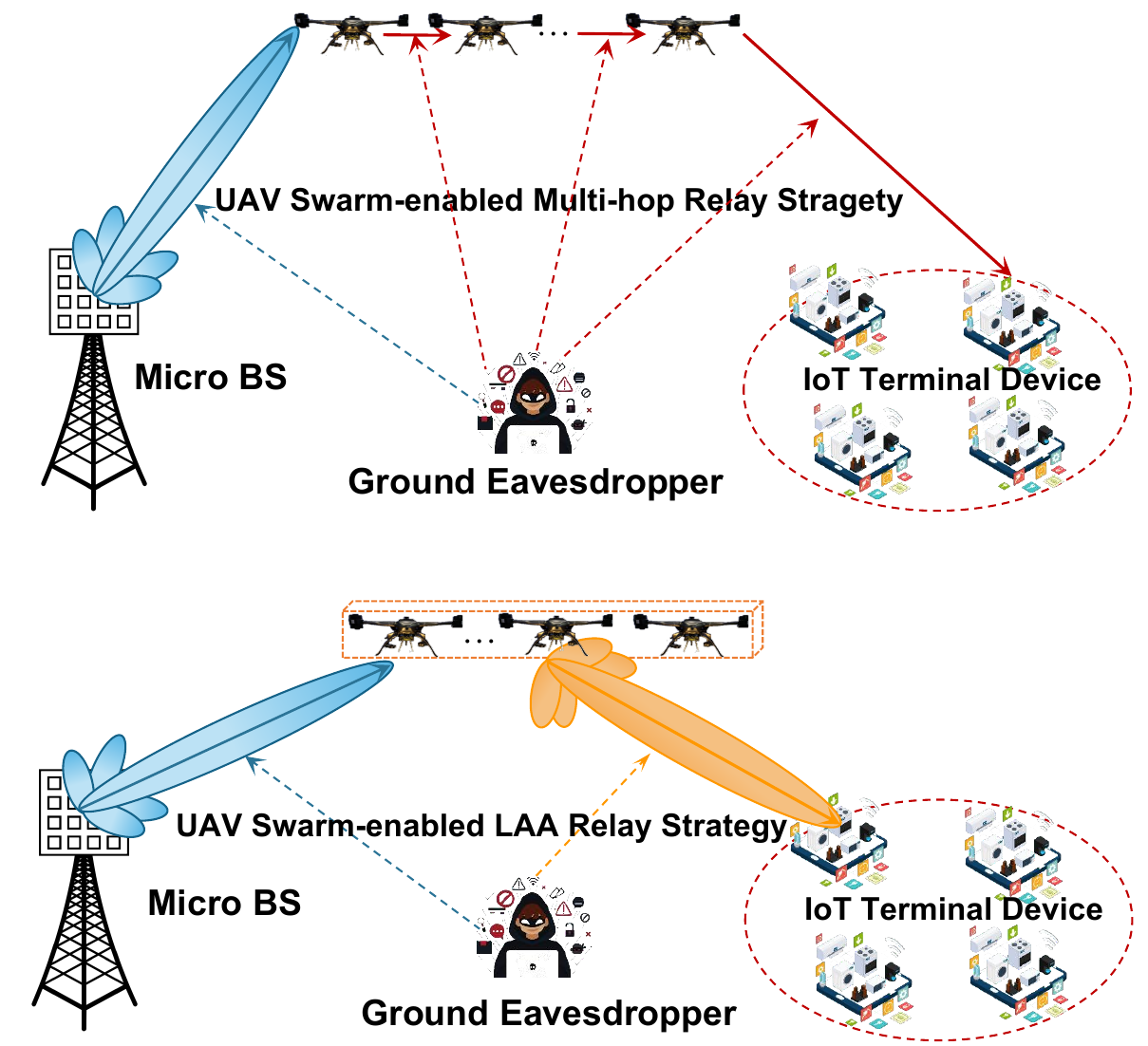}
		\caption{Schematic maps of UAV swarm-enabled multi-hop relay and UAV swarm-enabled LAA relay strategies.}
		\label{Figure:Relay}
	\end{figure}
	%%	SubSubSection:Visualization Results
	\subsubsection{Visualization Results}
	\label{SubSubSection:Visualization Results}
	
	\par Fig.~\ref{Figure:Optimization results} shows the optimization results of IMOGOA, by visualizing the optimized transmission rate distribution of the PAA, the optimized transmission rate distribution of the UVAA, and the trajectory of the UAV swarm. Due to space limitations, we only show a schematic for serving one associated IoT terminal device. From the Figs.~\ref{Figure:Optimization results}\subref{Figure:OptimizedPAARate} and \ref{Figure:Optimization results}\subref{Figure:OptimizedCBRate}, we can observe that the transmission rate is the highest in the direction of the target, regardless of whether the direction is towards the UAV receiver at the PAA or the IoT terminal device at the UVAA. Moreover, the transmission rate is relatively low in the direction of ground eavesdropper. Furthermore, Fig.~\ref{Figure:Optimization results}\subref{Figure:OptimizedUAVTrajectory} depicts the trajectory of the UAV swarm, demonstrating that the UAVs are appropriately spaced without colliding or coupling with each other, yet not overly dispersed. Accordingly, these results indicate that the proposed IMOGOA can yield significant optimization results by optimizing the three optimization objectives in the given scenario.
	
	\begin{figure}[htbp]
		\centering
		\includegraphics[width=0.8\linewidth,scale=1.00]{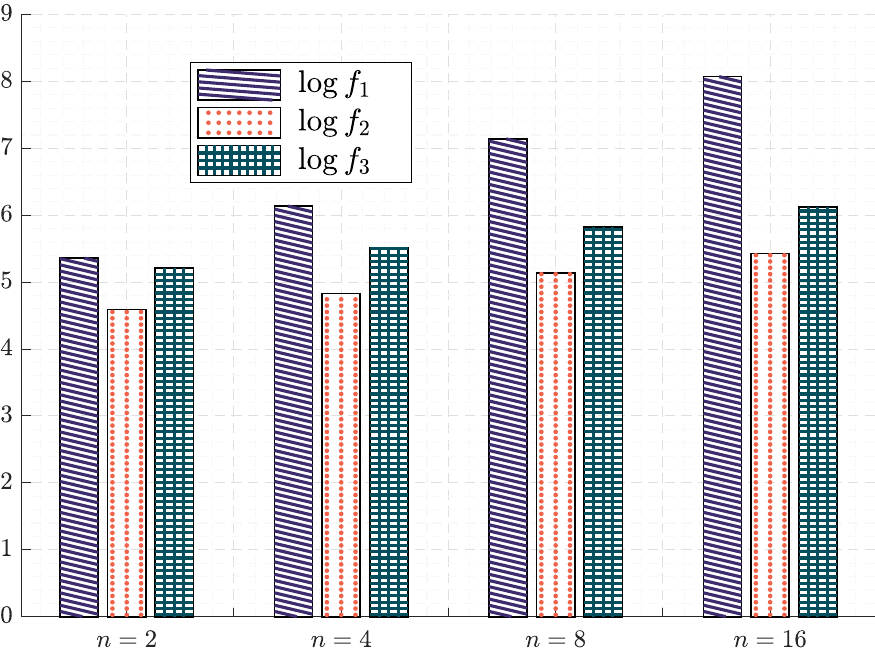}
		\caption{MRS performance with the different number of UAV relay hops.}
		\label{Figure:RelayBar}
	\end{figure}
	
	%%	SubSubSection:Comparison with Other Relay Strategy
	\subsubsection{Comparison with Other Relay Strategies}
	\label{SubSubSection:Comparison with Other Relay Strategy}

	\par In this work, we consider other two relay strategies that are the conventional UAV swarm-enabled multi-hop relay \cite{Kim2020} and UAV swarm-enabled linear antenna array (LAA) relay \cite{Mozaffari2019} for comparison, and these two scenarios are shown in Fig.~\ref{Figure:Relay} and the details are as follows.
	\begin{itemize}
		\item \textit{UAV swarm-enabled multi-hop relay strategy (MRS):} The traditional MRS in our simulation design for sequentially serving $8$ associated IoT terminal devices where $n = 2, 4, 8, 16$ UAVs are equally spaced as hop-by-hop relays.
		\item \textit{UAV swarm-enabled LAA relay strategy (LRS):} In our simulation, the traditional LRS is deployed in the center of movement space and set with different array element spacings of $1$ m, $2$ m, $3$ m, $4$ m, $5$ m.
	\end{itemize}

	\par {\color{color}Fig.~\ref{Figure:RelayBar} shows the results for all three objectives with different numbers of hops in UAV swarm-enabled MRS. Note that all the three objectives have been processed by log-transformation because of different orders of magnitude. First, it is clear that as the number of hops increases, the achievable sum rate of IoT terminal devices depicts a upward trend. This is attributed to the diminished channel loss. Specifically, the closer distance and the higher probability of LoS are attained with the increase of hops. Moreover, the achievable sum rate of $\mathcal{E}$ exhibits an upward trend with the increase of hops, which is associated with the increase of wiretap links. Additionally, the energy consumption of UAV swarm also shows an upward trend with the increase of hops. Furthermore, Table \ref{Table:Performance comparisons between multi-hop and CB strategy} provides numerical simulation results for three objectives for the UVAA relay strategy (URS) utilizing the proposed IMOGOA in the case of $16$ UAVs and MRS in the cases of $2, 4, 8, 16$ UAVs. As can be seen, our strategy has achieved a better trade-off among the three optimization objectives, which provided a lower achievable sum rate of $\mathcal{E}$ and less UAV swarm energy consumption with the similar achievable sum rate of IoT terminal devices.}

	\begin{figure}[htbp]
		\centering
		\includegraphics[width=0.8\linewidth,scale=1.00]{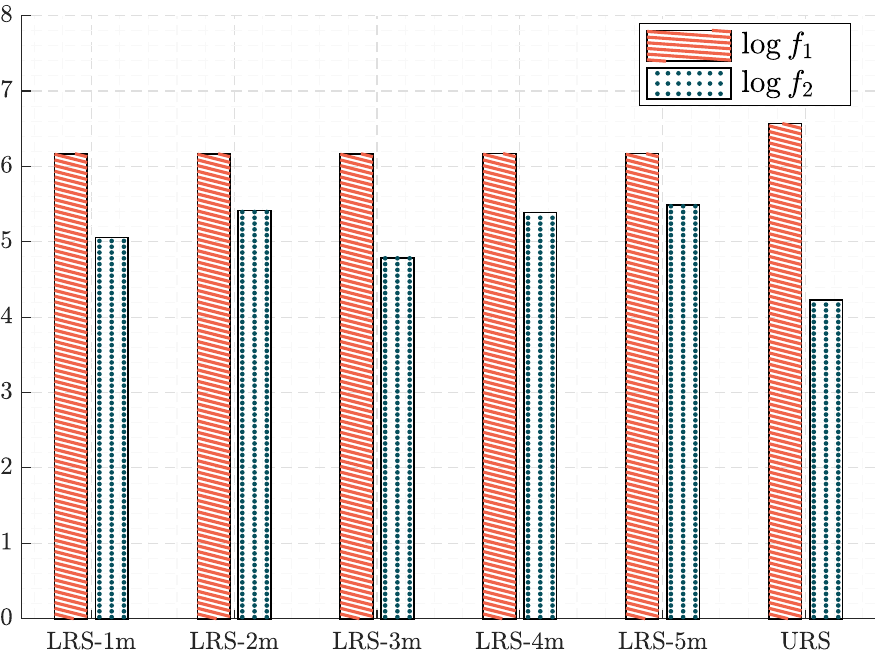}
		\caption{Performance comparison between LRS and URS.}
		\label{Figure:LAABar}
	\end{figure}
	\begin{table}
		\renewcommand\arraystretch{1.2}
		\centering 
		\caption{Performance Comparison between MRS and URS}  
		\label{Table:Performance comparisons between multi-hop and CB strategy}  
		\begin{tabular}{l l l l}\toprule[1.5pt]
			{Benchmarks} & {$f_{1}$ [bps]} & {${f_2}$ [bps]} &{$f_{3}$ [J]} \\  \toprule[1pt]
			{MRS ($2$ UAVs)} & {$2.2827 \times 10^5$} & {$3.8974\times 10^4$} &{$1.6325 \times 10^5$}  \\
			{MRS ($4$ UAVs)} & {$1.3737 \times 10^6$} & {$6.7771 \times 10^4$} &{$3.2865 \times 10^5$}  \\
			{MRS ($8$ UAVs)} & {$1.3837 \times 10^7$} & {$1.3725 \times 10^5$} &{$6.6050 \times 10^5$}  \\
			{MRS ($16$ UAVs)} & {$1.1808 \times 10^8$} & {$2.6738 \times 10^5$} &{$1.3225 \times 10^6$}  \\
			{URS ($16$ UAVs)} & {$3.7021 \times 10^6$} & {$1.6846 \times 10^4$} &{$5.7818 \times 10^4$}  \\
			\toprule[1.5pt]
		\end{tabular}  
	\end{table}
	\par Fig.~\ref{Figure:LAABar} shows the corresponding simulation results of LRS and URS, where each group of LRS is the mean calculated by simulating $100$ times to reduce the effect of randomness. Specifically, according numerical results are listed in Table~\ref{Table:Performance comparisons between LAA and CB strategy}. We can clearly observe that our proposed URS is superior to LRS, which may be due to the advantage of our strategy in the adjustable space of antenna array element positions.

	\begin{table}
		\renewcommand\arraystretch{1.2}
		\centering 
		\caption{Performance Comparison between LRS and URS}  
		\label{Table:Performance comparisons between LAA and CB strategy}  
		\begin{tabular}{l l l l}\toprule[1.5pt]
			{Benchmarks} & {$f_{1}$ [bps]} & {${f_2}$ [bps]} \\  \toprule[1pt]
			{LRS ($1$3 m)} & {$1.4743 \times 10^6$} & {$1.1358\times 10^5$}  \\
			{LRS ($2$ m)} & {$1.4712 \times 10^6$} & {$2.6128 \times 10^5$}  \\
			{LRS ($3$ m)} & {$1.4713 \times 10^6$} & {$6.0798 \times 10^4$} \\
			{LRS ($4$ m)} & {$1.4804 \times 10^6$} & {$2.4441 \times 10^5$} \\
			{LRS ($5$ m)} & {$1.4474 \times 10^6$} & {$3.0725 \times 10^5$} \\
			{URS} & {$3.7021 \times 10^6$} & {$1.6846 \times 10^4$} \\
			\toprule[1.5pt]
		\end{tabular}  
	\end{table}

	%%	SubSubSection:Comparison with Other Multi-objective Optimization Algorithms
	\subsubsection{Comparison with Other Multi-objective Optimization Algorithms}
	\label{SubSubSection:Comparison with Other Multi-objective Optimization Algorithms}
	
	\par In this section, the proposed IMOGOA is compared with 5 other multi-objective optimization algorithms mentioned in Section~\ref{SubSection:Simulation Setups}, and the hyperparameters of these algorithms are listed in Table~\ref{Table:hyperparameters of the algorithms}. Moreover, the maximum iteration number and population size of each algorithm are set to 500 and 30, respectively.
	
	\begin{table}
		\renewcommand\arraystretch{1.2}
		\centering 
		\caption{Hyperparameters of the Algorithms}  
		\label{Table:hyperparameters of the algorithms}  
		\begin{tabular}{c c}\toprule[1.5pt]
			{Algorithms} & {Hyperparameters} \\ \toprule[1pt]
			{NSGA-II} &{$p_{c}=0.9,p_{m}=0.1$} \\
			{MOPSO} & {$w=0.5,c_{1}=1,c_{2}=2,w_{damp}=0.99$} \\
			{MOGWO} & {$\alpha=0.1,\beta=4,\gamma=2$} \\
			{MOGOA} & {$c_{max}=1,c_{min}=0.0004$} \\
			{MOMVO} & {$WEP_{max}=1,WEP_{min}=0.2$} \\
			{IMOGOA} & {$c_{max}=1,c_{min}=0.0004,\alpha_{1}=0.2,\alpha_{2}=0.2$} \\
			\toprule[1.5pt]
		\end{tabular}  
	\end{table}
	
	\begin{figure}[htbp]
		\begin{minipage}[b]{\linewidth}
			\subfloat[]{
				\label{Figure:Normal}
				\includegraphics[width=\linewidth,scale=1.00]{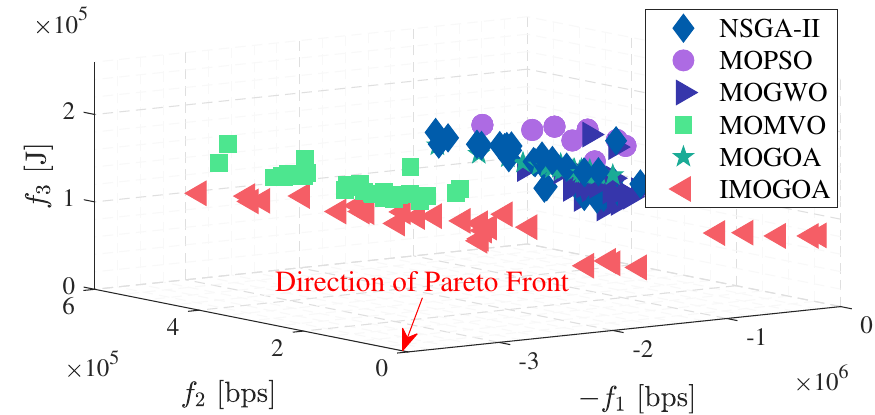}}
		\end{minipage}
		\begin{minipage}[b]{0.32\linewidth}
			\subfloat[]{
				\label{Figure:FrontView}
				\includegraphics[width=\linewidth]{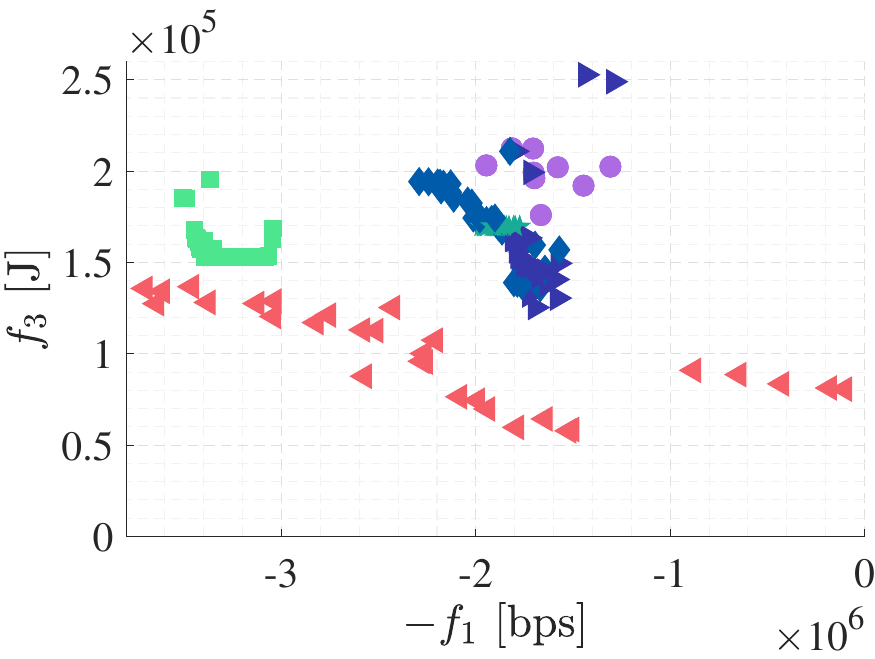}}
			\subfloat[]{
				\label{Figure:LeftView}
				\includegraphics[width=\linewidth]{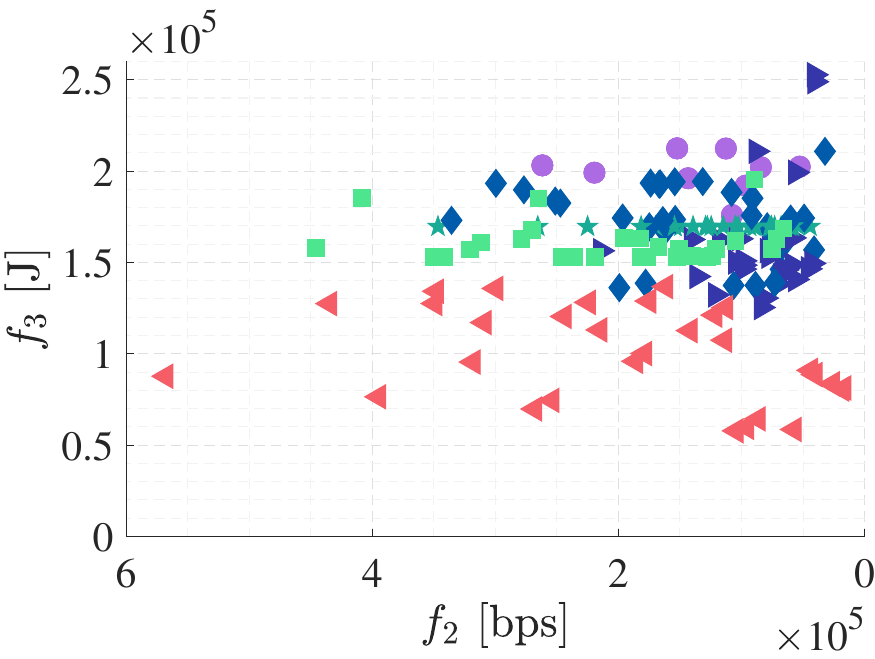}}
			\subfloat[]{
				\label{Figure:TopView}
				\includegraphics[width=\linewidth]{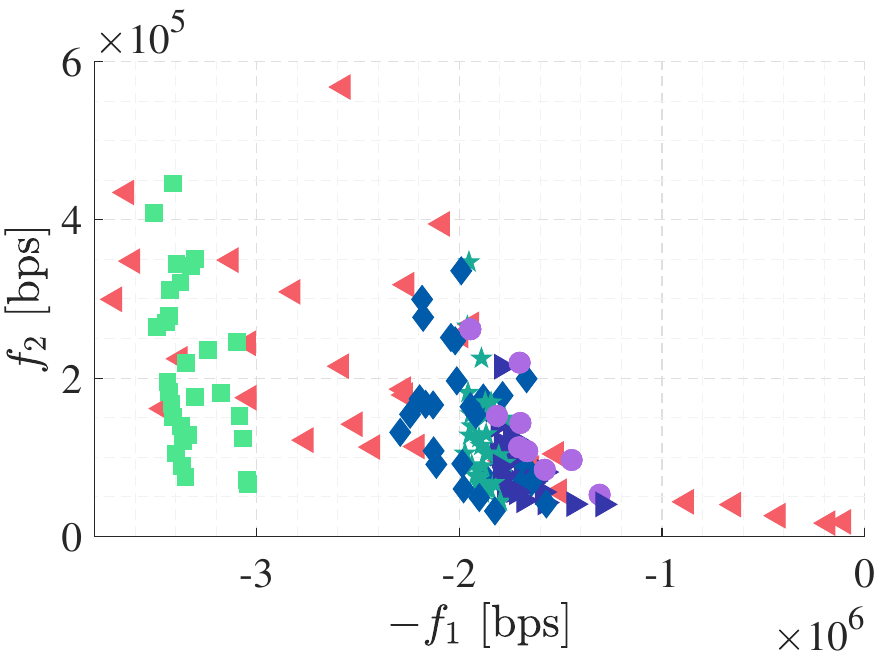}}
		\end{minipage}
		
		\caption{Solution distributions obtained by different algorithms of (a) perspective view, (b) front plan view, (c) left plan view and (d) top plan view. The direction pointed by the red arrow is the direction of true Pareto front.} 
		\label{Figure:Pareto Front}
	\end{figure}

	\par Firstly, we show numerical results obtained by these algorithms in terms of the achievable sum rate of IoT terminal devices, the achievable sum rate of $\mathcal{E}$, the energy consumption of UAV swarm and HV indicator in Table~\ref{Table:Numerical optimization results}. As can be observed, our proposed IMOGOA achieves the best results on all three objectives compared to other algorithms. Additionally, the HV indicator of IMOGOA is also larger than other algorithms, which indicates that the comprehensive performance of IMOGOA is superior among all algorithms. To further demonstrate, Fig.~\ref{Figure:Pareto Front} shows the Pareto solution distributions of these algorithms. Note that $f_{1}$ exhibits negative values because a unified optimization direction is required negative values for $f_{1}$ in Eq.~\eqref{UASCMOP}. It can be seen from the figure that the proposed IMOGOA algorithm has a closer estimated Pareto Front to the actual Pareto Front when compared to other algorithms. Specially, Fig.~\ref{Figure:Comparison of solution distributions obtained by MOGOA and IMOGOA at different iterations} depicts a more intuitive comparison of solution distributions obtained by MOGOA and IMOGOA at various iterations, which highlights that IMOGOA is closer to the pareto front than MOGOA over iteration number, indicating the contributions of some improvement factors introduced in IMOGOA. In summary, these results demonstrate that the proposed IMOGOA is more suitable for solving the US$^2$RMOP and outperforms other benchmarks in terms of performance.

	\begin{figure}[htbp]
		\begin{minipage}[b]{\linewidth}
			\subfloat[]{
				\label{Figure:NormalCompare}
				\includegraphics[width=\linewidth,scale=1.00]{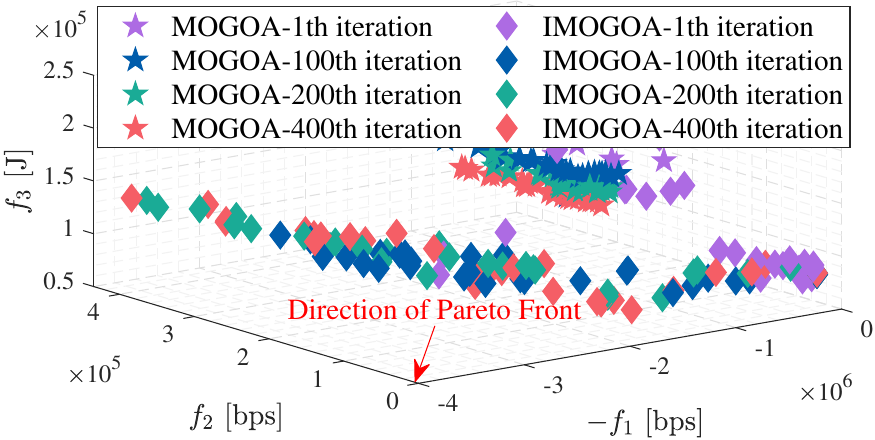}}
		\end{minipage}
		\begin{minipage}[b]{0.32\linewidth}
			\subfloat[]{
				\label{Figure:FrontViewCompare}
				\includegraphics[width=\linewidth]{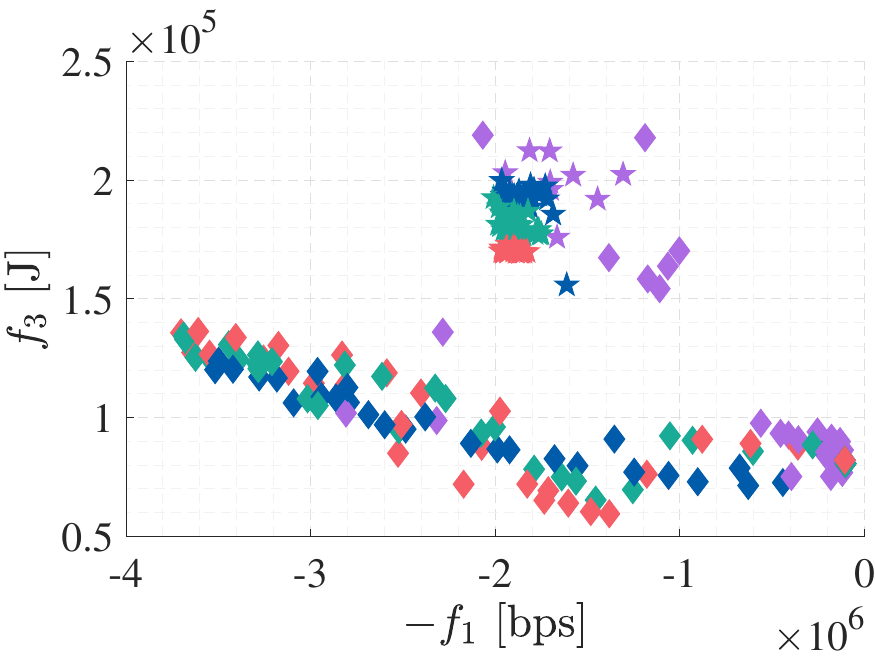}}
			\subfloat[]{
				\label{Figure:LeftViewCompare}
				\includegraphics[width=\linewidth]{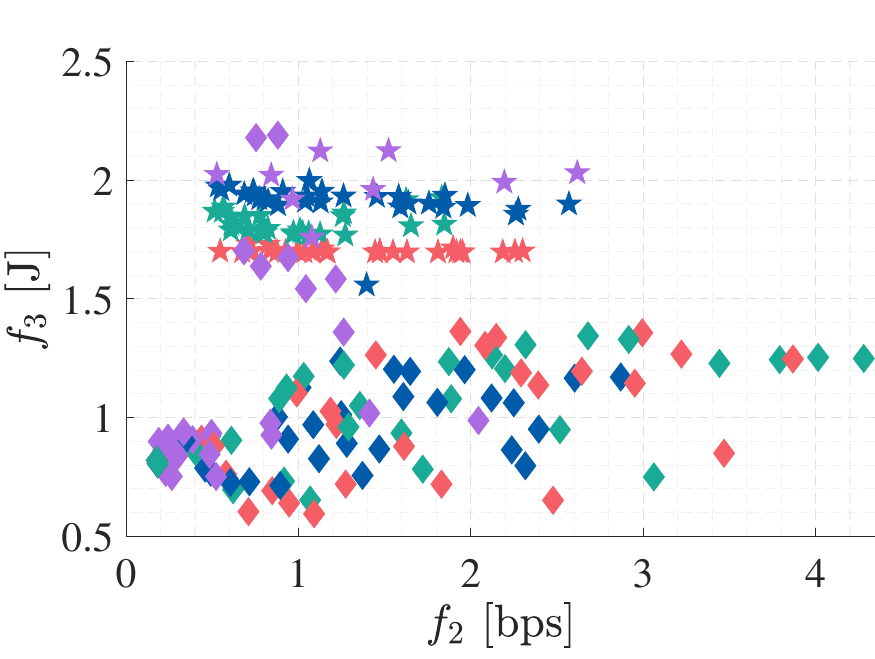}}
			\subfloat[]{
				\label{Figure:TopViewCompare}
				\includegraphics[width=\linewidth]{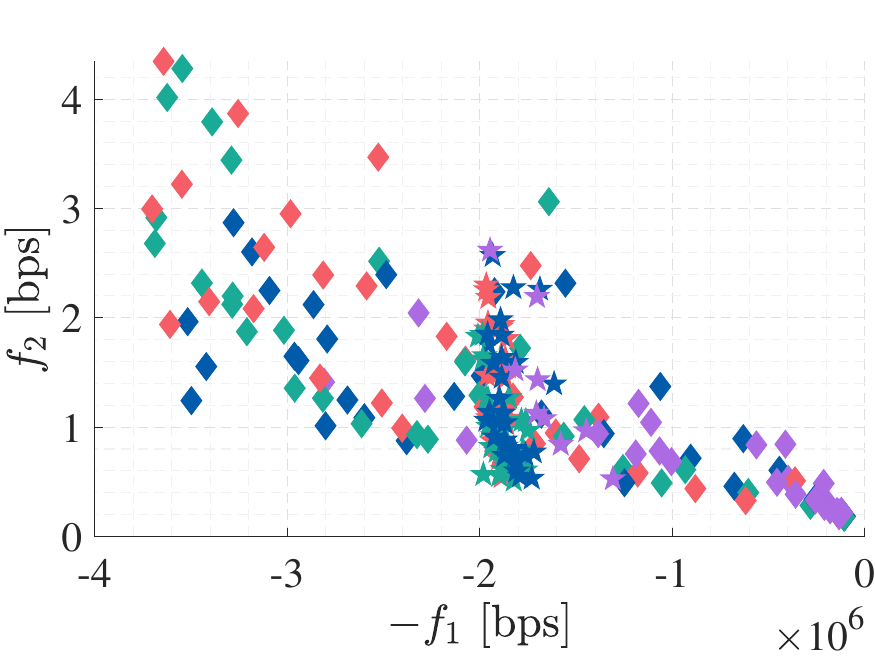}}
		\end{minipage}
		
		\caption{Comparison of solution distributions obtained by MOGOA and IMOGOA of (a) perspective view, (b) front plan view, (c) left plan view and (d) top plan view at different iterations: (i) $1$th iteration. (ii) $100$th iteration. (iii) $200$th iteration. (iv) $400$th iteration.}
		\label{Figure:Comparison of solution distributions obtained by MOGOA and IMOGOA at different iterations}
	\end{figure}
	
	\begin{table}
		\renewcommand\arraystretch{1.2}
		
		\setlength\tabcolsep{3pt}
		\centering 
		\caption{Numerical Optimization Results Obtained by Different Algorithms}  
		\label{Table:Numerical optimization results}  
		{\scriptsize{
				\begin{tabular}{l l l l l}\toprule[1.5pt]
					{Algorithms} & {$f_{1}$ [bps]} & {${f_2}$ [bps]} &{$f_{3}$ [J]} & {$HV$}\\ \toprule[1pt]
					{NSGA-II} &{$2.2921 \times 10^6$} & {$3.2145\times 10^4$} & {$1.3617 \times 10^5$} & {$0.4604$} \\
					{MOPSO} & {$1.9461 \times 10^6$} & {$5.2678 \times 10^4$} & {$1.7594 \times 10^5$} & {$0.3358$}\\
					{MOGWO} & {$1.8129\times 10^6$} & {$4.0500 \times 10^4$} & {$1.2532 \times 10^5$} & {$0.4695$}\\
					{MOMVO} & {$3.5067 \times 10^6$} & {$6.5821 \times 10^4$} & {$1.5278 \times 10^5$} & {$0.4303$}\\
					{MOGOA} & {$1.9727 \times 10^6$} & {$4.4313 \times 10^4$} & {$1.6964 \times 10^5$} & {$0.3801$}\\
					{IMOGOA} & \pmb{$3.7021 \times 10^6$} & \pmb{$1.6846 \times 10^4$} & \pmb{$5.7818 \times 10^4$} & \pmb{$0.6615$}\\
					\toprule[1.5pt]
		\end{tabular}  }}
	\end{table}
	{\color{color}
		\subsubsection{Results with Multiple Eavesdroppers}
		\label{SubSubsection: Results with Multiple Eavesdroppers}
		\par In this section, we consider two situations with different numbers of eavesdroppers as follows.
		\begin{itemize}
			\item \textit{Only collude in time domain (OCTD):} In this situation, the eavesdropper simply only colludes alone in time domain, and mutliple eavesdroppers do not cooperate with each other.
			\item \textit{Collude both in time and space domains (CTSD):} In this situation, the eavesdropper collude alone in time domain. Moreover, multiple eavesdroppers in different positions can collude in space domain with each other.
		\end{itemize}

		\par Fig.~\ref{Figure: Performance comparison with different number of eavesdroppers} shows the performance comparison under OCTD and CTSD with different numbers of eavesdroppers. First, our proposed IMOGOA outperforms these benchmarks in all three objectives and has more stable performance for the two situations. Moreover, when concentrating on $f_{2}$ under CTSD, it can be seen that IMOGOA is on a slow growth trend with the increasing of the number of eavesdroppers, which also better validates the efficiency and scalability of our improvements on MOGOA for solving the US$^2$RMOP.}
	\begin{figure}[htbp]
		\centering
		\subfloat[]{
			\label{Non-Collusive f1}
			\includegraphics[width=0.49\linewidth]{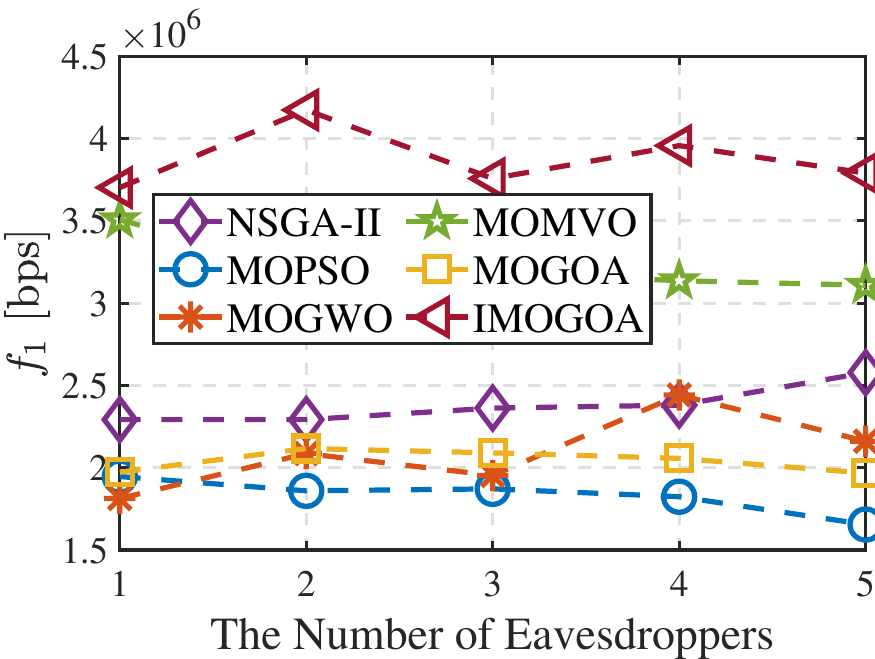}}
		\subfloat[]{
			\label{Collusive f1}
			\includegraphics[width=0.49\linewidth]{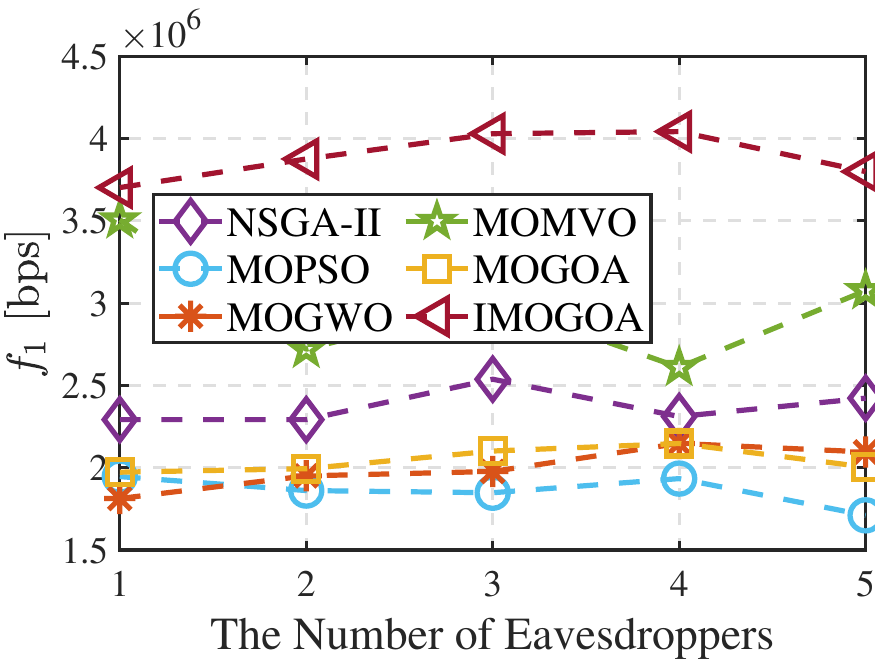}}
		\hfill
		\subfloat[]{
			\label{Non-Collusive f2}
			\includegraphics[width=0.49\linewidth]{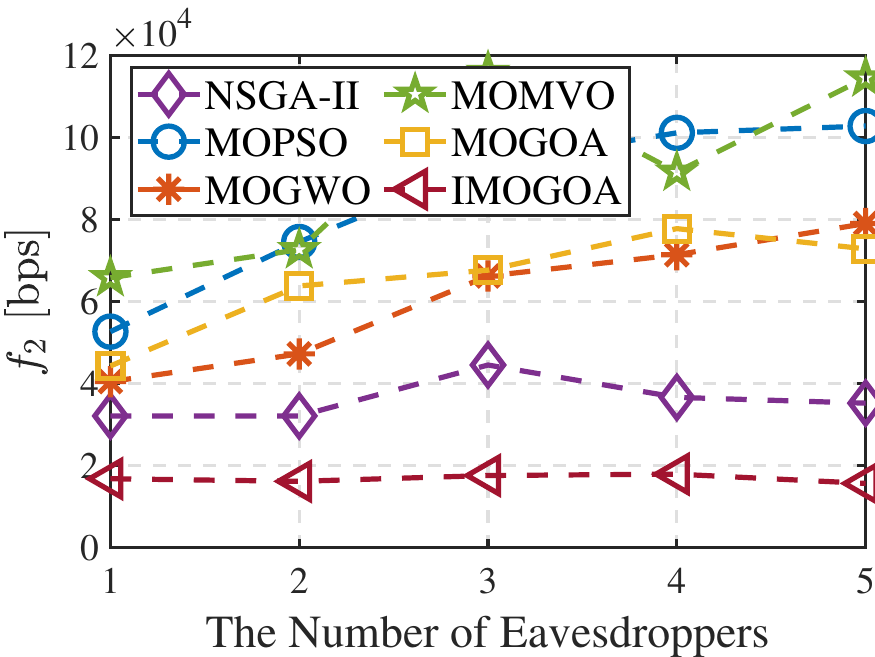}}
		\subfloat[]{
			\label{Collusive f2}
			\includegraphics[width=0.49\linewidth]{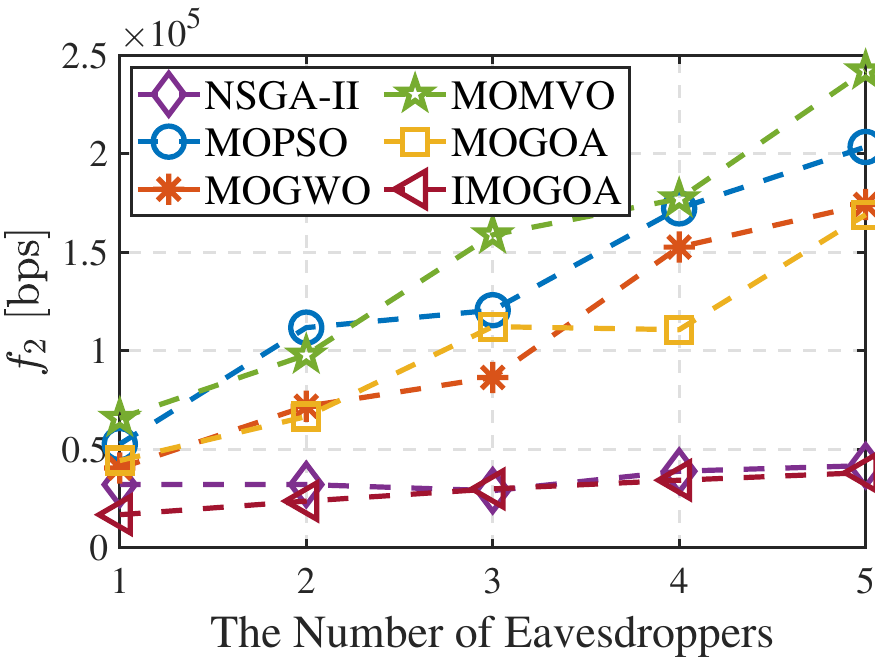}}
		\hfill
		\subfloat[]{
			\label{Non-Collusive f3}
			\includegraphics[width=0.49\linewidth]{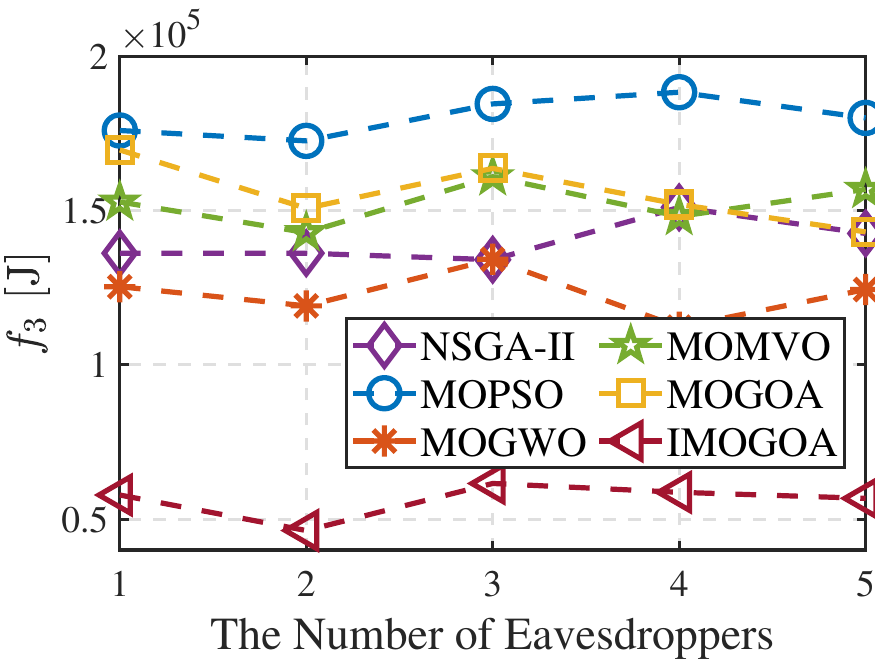}}
		\subfloat[]{
			\label{Collusive f3}
			\includegraphics[width=0.49\linewidth]{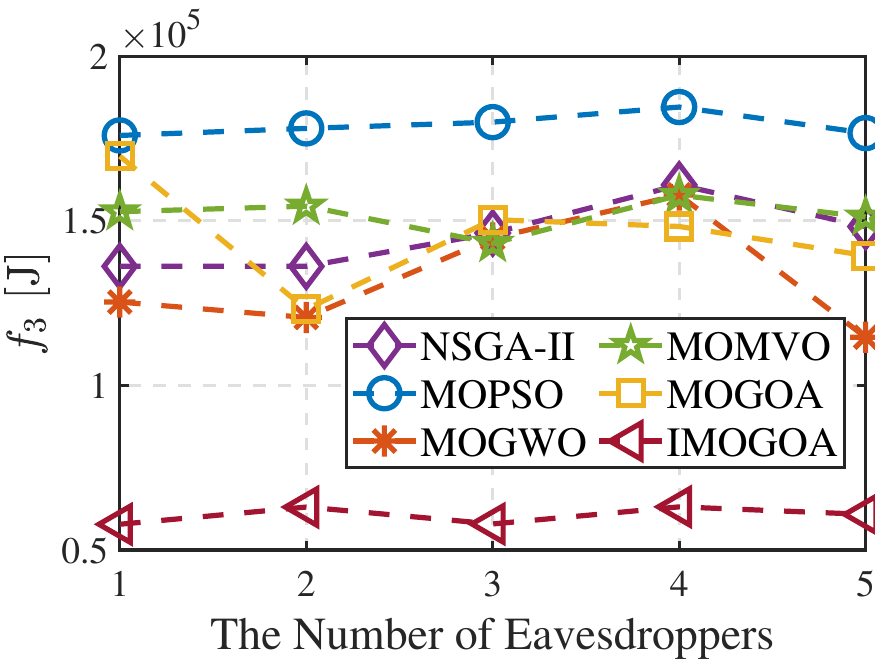}}
		\caption{Performance comparison with different number of eavesdroppers. (a) $f_{1}$ under OCTD. (b) $f_{1}$ under CTSD. (c) $f_{2}$ under OCTD. (d) $f_{2}$ under CTSD. (e) $f_{3}$ under OCTD. (f) $f_{3}$ under CTSD.}
		\label{Figure: Performance comparison with different number of eavesdroppers}
	\end{figure}
	%%%
	%	Section:Conclusion
	%%%
	\section{Conclusion} 
	\label{Section:Conclusion}
	
	\par In this paper, the UAV swarm-enabled collaborative secure relay system is proposed where a UAV swarm serves for forwarding confidential messages from the source MBS with PAA to the remote IoT terminal devices via CB so as to counteract the threat of time-domain collusive eavesdropper. Furthermore, we formulate an US$^2$RMOP to maximize the achievable sum rate of all IoT terminal devices, minimizing the achievable sum rate of the eavesdropper, and minimizing the energy consumption of UAV swarm. Subsequently, an IMOGOA with several improvements is proposed to solve US$^2$RMOP. Simulation results illustrate the effectiveness of the proposed UAV swarm-enabled collaborative secure relay system by comparing both traditional UAV swarm-enabled multi-hop relay and LAA relay strategies, and verify that IMOGOA has better performance than several other comparison algorithms. {\color{color}In addition, IMOGOA is more stable and effective under OCTD and CTSD with multiple eavesdroppers.}
	
	%%%
	%   Reference
	%%%
	\bibliographystyle{IEEEtran}
	\bibliography{cite}

\begin{thebibliography}{10}
\providecommand{\url}[1]{#1}
\csname url@samestyle\endcsname
\providecommand{\newblock}{\relax}
\providecommand{\bibinfo}[2]{#2}
\providecommand{\BIBentrySTDinterwordspacing}{\spaceskip=0pt\relax}
\providecommand{\BIBentryALTinterwordstretchfactor}{4}
\providecommand{\BIBentryALTinterwordspacing}{\spaceskip=\fontdimen2\font plus
\BIBentryALTinterwordstretchfactor\fontdimen3\font minus
  \fontdimen4\font\relax}
\providecommand{\BIBforeignlanguage}[2]{{%
\expandafter\ifx\csname l@#1\endcsname\relax
\typeout{** WARNING: IEEEtran.bst: No hyphenation pattern has been}%
\typeout{** loaded for the language `#1'. Using the pattern for}%
\typeout{** the default language instead.}%
\else
\language=\csname l@#1\endcsname
\fi
#2}}
\providecommand{\BIBdecl}{\relax}
\BIBdecl

\bibitem{CSCWD}
C.~Zhang, G.~Sun, J.~Li, and X.~Zheng, ``Bi-objective optimization for {UAV}
  swarm-enabled relay communications via collaborative beamforming,'' in
  \emph{Proc. IEEE 26th Int. Conf. Comput. Supported Cooperat. Work Design.
  {(CSCWD)}}, 2023, pp. 984--989.

\bibitem{Zeng2019}
Y.~Zeng, Q.~Wu, and R.~Zhang, ``Accessing from the sky: {A} tutorial on {UAV}
  communications for 5{G} and beyond,'' \emph{Proc. {IEEE}}, vol. 107, no.~12,
  pp. 2327--2375, Dec. 2019.

\bibitem{Xu2021}
H.~Xu, W.~Huang, Y.~Zhou, D.~Yang, M.~Li, and Z.~Han, ``Edge computing resource
  allocation for unmanned aerial vehicle assisted mobile network with
  blockchain applications,'' \emph{{IEEE} Trans. Wirel. Commun.}, vol.~20,
  no.~5, pp. 3107--3121, May. 2021.

\bibitem{Liu2020}
J.~Liu, X.~Du, J.~Cui, M.~Pan, and D.~Wei, ``Task-oriented intelligent
  networking architecture for the space-air-ground-aqua integrated network,''
  \emph{{IEEE} Internet Things J.}, vol.~7, no.~6, pp. 5345--5358, Jun. 2020.

\bibitem{Xu2022}
Q.~Xu, Z.~Su, R.~Lu, and S.~Yu, ``Ubiquitous transmission service: Hierarchical
  wireless data rate provisioning in space-air-ocean integrated networks,''
  \emph{{IEEE} Trans. Wirel. Commun.}, vol.~21, no.~9, pp. 7821--7836, Sep.
  2022.

\bibitem{LiJiahui2023}
J.~Li, G.~Sun, H.~Kang, A.~Wang, S.~Liang, Y.~Liu, and Y.~Zhang,
  ``Multi-objective optimization approaches for physical layer secure
  communications based on collaborative beamforming in uav networks,''
  \emph{{IEEE/ACM} Trans. Netw.}, pp. 1--16, 2023, {E}arly {A}ccess, doi:
  {10.1109/TNET.2023.3234324}.

\bibitem{Samir2020}
M.~Samir, S.~Sharafeddine, C.~M. Assi, T.~M. Nguyen, and A.~Ghrayeb, ``{UAV}
  trajectory planning for data collection from time-constrained {I}o{T}
  devices,'' \emph{{IEEE} Trans. Wirel. Commun.}, vol.~19, no.~1, pp. 34--46,
  Jan. 2020.

\bibitem{Pan2023}
H.~Pan, Y.~Liu, G.~Sun, J.~Fan, S.~Liang, and C.~Yuen, ``Joint power and 3d
  trajectory optimization for uav-enabled wireless powered communication
  networks with obstacles,'' \emph{{IEEE} Trans. Commun.}, vol.~71, no.~4, pp.
  2364--2380, 2023.

\bibitem{Zeng2019a}
Y.~Zeng, J.~Lyu, and R.~Zhang, ``Cellular-connected {UAV:} potential,
  challenges, and promising technologies,'' \emph{{IEEE} Wirel. Commun.},
  vol.~26, no.~1, pp. 120--127, Feb. 2019.

\bibitem{Zhang2018}
S.~Zhang and J.~Liu, ``Analysis and optimization of multiple unmanned aerial
  vehicle-assisted communications in post-disaster areas,'' \emph{{IEEE} Trans.
  Veh. Technol.}, vol.~67, no.~12, pp. 12\,049--12\,060, Dec. 2018.

\bibitem{GengSun2023}
G.~Sun, X.~Zheng, Z.~Sun, Q.~Wu, J.~Li, Y.~Liu, and V.~C. Leung, ``Uav-enabled
  secure communications via collaborative beamforming with imperfect
  eavesdropper information,'' \emph{IEEE Trans. Mobile Comput.}, pp. 1--18,
  2023, {E}arly {A}ccess, doi: {10.1109/TMC.2023.3273293}.

\bibitem{Ochiai2005}
H.~Ochiai, P.~Mitran, H.~V. Poor, and V.~Tarokh, ``Collaborative beamforming
  for distributed wireless ad hoc sensor networks,'' \emph{{IEEE} Trans. Signal
  Process.}, vol.~53, no.~11, pp. 4110--4124, Nov. 2005.

\bibitem{Ahmed2009}
M.~F.~A. Ahmed and S.~A. Vorobyov, ``Collaborative beamforming for wireless
  sensor networks with gaussian distributed sensor nodes,'' \emph{{IEEE} Trans.
  Wirel. Commun.}, vol.~8, no.~2, pp. 638--643, Feb. 2009.

\bibitem{Zhang2010}
J.~Zhang and M.~C. Gursoy, ``Collaborative relay beamforming for secrecy,'' in
  \emph{Proc. {IEEE} Int. Conf. Commun. {(ICC)}}, 2010, pp. 1--5.

\bibitem{Yang2016}
W.~Yang, K.~Wang, X.~Xu, and J.~Zhou, ``Secure transmission for {AF} relaying
  spectrum-sharing systems with collaborative distributed beamforming,'' in
  \emph{Proc. 25th Wireless Opt. Commun. Conf. {(WOCC)}}, 2016, pp. 1--4.

\bibitem{Sun2022}
G.~Sun, J.~Li, A.~Wang, Q.~Wu, Z.~Sun, and Y.~Liu, ``Secure and
  energy-efficient {UAV} relay communications exploiting collaborative
  beamforming,'' \emph{{IEEE} Trans. Commun.}, vol.~70, no.~8, pp. 5401--5416,
  Aug. 2022.

\bibitem{Li2021}
J.~Li, H.~Kang, G.~Sun, S.~Liang, Y.~Liu, and Y.~Zhang, ``Physical layer secure
  communications based on collaborative beamforming for {UAV} networks: {A}
  multi-objective optimization approach,'' in \emph{Proc. 40th {IEEE} Conf.
  Comput. Commun. ({INFOCOM})}, 2021, pp. 1--10.

\bibitem{Zeng2017}
Y.~Zeng and R.~Zhang, ``Energy-efficient {UAV} communication with trajectory
  optimization,'' \emph{{IEEE} Trans. Wirel. Commun.}, vol.~16, no.~6, pp.
  3747--3760, Jun. 2017.

\bibitem{Zhang2019}
S.~Zhang, Y.~Zeng, and R.~Zhang, ``Cellular-enabled {UAV} communication: {A}
  connectivity-constrained trajectory optimization perspective,'' \emph{{IEEE}
  Trans. Commun.}, vol.~67, no.~3, pp. 2580--2604, Mar. 2019.

\bibitem{Li2018}
R.~Li, Z.~Wei, L.~Yang, D.~W.~K. Ng, N.~Yang, J.~Yuan, and J.~An, ``Joint
  trajectory and resource allocation design for {UAV} communication systems,''
  in \emph{Proc. IEEE Global Telecommun. Conf. (Globecom) Workshops}, 2018, pp.
  1--6.

\bibitem{Wu2017}
Q.~Wu, Y.~Zeng, and R.~Zhang, ``Joint trajectory and communication design for
  {UAV}-enabled multiple access,'' in \emph{Proc. IEEE Global Telecommun. Conf.
  (Globecom)}, 2017, pp. 1--6.

\bibitem{Hua2020}
M.~Hua, L.~Yang, Q.~Wu, and A.~L. Swindlehurst, ``3{D} {UAV} trajectory and
  communication design for simultaneous uplink and downlink transmission,''
  \emph{{IEEE} Trans. Commun.}, vol.~68, no.~9, pp. 5908--5923, Sep. 2020.

\bibitem{Meng2022}
K.~Meng, Q.~Wu, S.~Ma, W.~Chen, and T.~Q.~S. Quek, ``{UAV} trajectory and
  beamforming optimization for integrated periodic sensing and communication,''
  \emph{{IEEE} Wirel. Commun. Lett.}, vol.~11, no.~6, pp. 1211--1215, Jun.
  2022.

\bibitem{Yang2021}
G.~Yang, R.~Dai, and Y.~Liang, ``Energy-efficient {UAV} backscatter
  communication with joint trajectory design and resource optimization,''
  \emph{{IEEE} Trans. Wirel. Commun.}, vol.~20, no.~2, pp. 926--941, Feb. 2021.

\bibitem{Zhong2019}
C.~Zhong, J.~Yao, and J.~Xu, ``Secure {UAV} communication with cooperative
  jamming and trajectory control,'' \emph{{IEEE} Commun. Lett.}, vol.~23,
  no.~2, pp. 286--289, Feb. 2019.

\bibitem{Cai2018}
Y.~Cai, F.~Cui, Q.~Shi, M.~Zhao, and G.~Y. Li, ``Dual-{UAV}-enabled secure
  communications: Joint trajectory design and user scheduling,'' \emph{{IEEE}
  J. Sel. Areas Commun.}, vol.~36, no.~9, pp. 1972--1985, Sep. 2018.

\bibitem{Zhou2018}
Y.~Zhou, P.~L. Yeoh, H.~Chen, Y.~Li, R.~Schober, L.~Zhuo, and B.~Vucetic,
  ``Improving physical layer security via a {UAV} friendly jammer for unknown
  eavesdropper location,'' \emph{{IEEE} Trans. Veh. Technol.}, vol.~67, no.~11,
  pp. 11\,280--11\,284, Nov. 2018.

\bibitem{Sun2020}
X.~Sun, W.~Yang, and Y.~Cai, ``Secure communication in {NOMA}-assisted
  millimeter-wave {SWIPT} {UAV} networks,'' \emph{{IEEE} Internet Things J.},
  vol.~7, no.~3, pp. 1884--1897, Mar. 2020.

\bibitem{Cheng2019}
F.~Cheng, G.~Gui, N.~Zhao, Y.~Chen, J.~Tang, and H.~Sari,
  ``{UAV}-relaying-assisted secure transmission with caching,'' \emph{{IEEE}
  Trans. Commun.}, vol.~67, no.~5, pp. 3140--3153, May. 2019.

\bibitem{Na2022}
Z.~Na, C.~Ji, B.~Lin, and N.~Zhang, ``Joint optimization of trajectory and
  resource allocation in secure {UAV} relaying communications for internet of
  things,'' \emph{{IEEE} Internet Things J.}, vol.~9, no.~17, pp.
  16\,284--16\,296, Sep. 2022.

\bibitem{Ji2021}
J.~Ji, K.~Zhu, D.~Niyato, and R.~Wang, ``Joint trajectory design and resource
  allocation for secure transmission in cache-enabled {UAV}-relaying networks
  with {D2D} communications,'' \emph{{IEEE} Internet Things J.}, vol.~8, no.~3,
  pp. 1557--1571, Feb. 2021.

\bibitem{Mohanti2019}
S.~Mohanti, C.~Bocanegra, J.~Meyer, G.~Secinti, M.~Diddi, H.~Singh, and K.~R.
  Chowdhury, ``Airbeam: Experimental demonstration of distributed beamforming
  by a swarm of {UAV}s,'' in \emph{Proc. 16th {IEEE} Int. Conf. Mobile Ad Hoc
  Sensor Syst. ({MASS})}, 2019, pp. 162--170.

\bibitem{Mozaffari2019}
M.~Mozaffari, W.~Saad, M.~Bennis, and M.~Debbah, ``Communications and control
  for wireless drone-based antenna array,'' \emph{{IEEE} Trans. Commun.},
  vol.~67, no.~1, pp. 820--834, Jan. 2019.

\bibitem{Dinh2019}
P.~Dinh, T.~M. Nguyen, S.~Sharafeddine, and C.~Assi, ``Joint location and
  beamforming design for cooperative {UAV}s with limited storage capacity,''
  \emph{{IEEE} Trans. Commun.}, vol.~67, no.~11, pp. 8112--8123, Nov. 2019.

\bibitem{Zhu2018}
S.~Zhu, K.~Yang, J.~Ouyang, and Y.~Du, ``Cooperative beamforming for
  {UAV}-assisted cognitive relay networks with partial channel state
  information,'' in \emph{Proc. IEEE 4th Int. Conf. Comput. Commun. ({ICCC})},
  2018, pp. 158--162.

\bibitem{Yang2017}
L.~Yang, J.~Chen, H.~Jiang, S.~A. Vorobyov, and H.~Zhang, ``Optimal relay
  selection for secure cooperative communications with an adaptive
  eavesdropper,'' \emph{{IEEE} Trans. Wirel. Commun.}, vol.~16, no.~1, pp.
  26--42, Jan. 2017.

\bibitem{Yan2016}
S.~Yan and R.~A. Malaney, ``Location-based beamforming for enhancing secrecy in
  rician wiretap channels,'' \emph{{IEEE} Trans. Wirel. Commun.}, vol.~15,
  no.~4, pp. 2780--2791, Apr. 2016.

\bibitem{Sun2019}
X.~Sun, D.~W.~K. Ng, Z.~Ding, Y.~Xu, and Z.~Zhong, ``Physical layer security in
  {UAV} systems: Challenges and opportunities,'' \emph{{IEEE} Wirel. Commun.},
  vol.~26, no.~5, pp. 40--47, Oct. 2019.

\bibitem{Duo2021}
B.~Duo, H.~Hu, Y.~Li, Y.~Hu, and X.~Zhu, ``Robust 3{D} trajectory and power
  design in probabilistic {L}os channel for {UAV}-enabled cooperative
  jamming,'' \emph{Veh. Commun.}, vol.~32, p. 100387, Dec. 2021.

\bibitem{Balanis2016}
C.~A. Balanis, ``Fundamental parameters and figures-of-merit of antennas,'' in
  \emph{Antenna theory: Analysis and design}.\hskip 1em plus 0.5em minus
  0.4em\relax John Wiley \& Sons, 2016.

\bibitem{Ding2020}
R.~Ding, F.~Gao, and X.~S. Shen, ``3d {UAV} trajectory design and frequency
  band allocation for energy-efficient and fair communication: {A} deep
  reinforcement learning approach,'' \emph{{IEEE} Trans. Wirel. Commun.},
  vol.~19, no.~12, pp. 7796--7809, Dec. 2020.

\bibitem{Babu2021}
N.~Babu, M.~Virgili, C.~B. Papadias, P.~Popovski, and A.~J. Forsyth, ``Cost-
  and energy-efficient aerial communication networks with interleaved hovering
  and flying,'' \emph{{IEEE} Trans. Veh. Technol.}, vol.~70, no.~9, pp.
  9077--9087, Sep. 2021.

\bibitem{Burer2012}
S.~Burer and A.~N. Letchford, ``Non-convex mixed-integer nonlinear programming:
  A survey,'' \emph{Surv. Oper. Res. Manage. Sci.}, vol.~17, pp. 97--106, Jul.
  2012.

\bibitem{Cao2020}
B.~Cao, S.~Fan, J.~Zhao, P.~Yang, K.~Muhammad, and M.~Tanveer,
  ``Quantum-enhanced multiobjective large-scale optimization via parallelism,''
  \emph{Swarm Evol. Comput.}, vol.~57, p. 100697, Sep. 2020.

\bibitem{Mirjalili2018}
S.~Z. Mirjalili, S.~Mirjalili, S.~Saremi, H.~Faris, and I.~Aljarah,
  ``Grasshopper optimization algorithm for multi-objective optimization
  problems,'' \emph{Appl. Intell.}, vol.~48, no.~4, pp. 805--820, Aug. 2018.

\bibitem{Luo2018}
J.~Luo, H.~Chen, Q.~zhang, Y.~Xu, H.~Huang, and X.~Zhao, ``An improved
  grasshopper optimization algorithm with application to financial stress
  prediction,'' \emph{Appl. Math. Model.}, vol.~64, pp. 654--668, Dec. 2018.

\bibitem{Wu2017a}
J.~Wu, H.~Wang, N.~Li, P.~Yao, Y.~Huang, Z.~Su, and Y.~Yu, ``Distributed
  trajectory optimization for multiple solar-powered {UAV}s target tracking in
  urban environment by adaptive grasshopper optimization algorithm,''
  \emph{Aerosp. Sci. Technol.}, vol.~70, pp. 497--510, Nov. 2017.

\bibitem{Heidari2019}
A.~A. Heidari, H.~Faris, I.~Aljarah, and S.~Mirjalili, ``An efficient hybrid
  multilayer perceptron neural network with grasshopper optimization,''
  \emph{Soft Comput.}, vol.~23, no.~17, pp. 7941--7958, Jul. 2019.

\bibitem{Goldberg2014}
D.~E. Goldberg and R.~Lingle, ``Alleles, loci, and the traveling salesman
  problem,'' in \emph{Proc. 1st Int. Conf. Genetic Algorithms Their Appl.},
  2014, pp. 154--159.

\bibitem{Zhao2022}
W.~Zhao, Z.~Zhang, S.~Mirjalili, L.~Wang, N.~Khodadadi, and S.~M. Mirjalili,
  ``An effective multi-objective artificial hummingbird algorithm with dynamic
  elimination-based crowding distance for solving engineering design
  problems,'' \emph{Comput. Methods Appl. Mech. Eng.}, vol. 398, p. 115223,
  Aug. 2022.

\bibitem{Dai2022}
M.~Dai, T.~H. Luan, Z.~Su, N.~Zhang, Q.~Xu, and R.~Li, ``Joint channel
  allocation and data delivery for {UAV}-assisted cooperative transportation
  communications in post-disaster networks,'' \emph{{IEEE} Trans. Intell.
  Transp. Syst.}, vol.~23, no.~9, pp. 16\,676--16\,689, Sep. 2022.

\bibitem{Deb2002}
K.~Deb, S.~Agrawal, A.~Pratap, and T.~Meyarivan, ``A fast and elitist
  multiobjective genetic algorithm: {NSGA-II},'' \emph{{IEEE} Trans. Evol.
  Comput.}, vol.~6, no.~2, pp. 182--197, Apr. 2002.

\bibitem{Coello2002}
C.~A.~C. Coello and M.~S. Lechuga, ``{MOPSO:} a proposal for multiple objective
  particle swarm optimization,'' in \emph{Proc. Congr. Evol. Comput. ({CEC})},
  2002, pp. 1051--1056.

\bibitem{Mirjalili2016}
S.~Mirjalili, S.~Saremi, S.~M. Mirjalili, and L.~dos Santos~Coelho,
  ``Multi-objective grey wolf optimizer: {A} novel algorithm for
  multi-criterion optimization,'' \emph{Expert Syst. Appl.}, vol.~47, pp.
  106--119, Apr. 2016.

\bibitem{Mirjalili2017}
S.~Mirjalili, P.~Jangir, S.~Z. Mirjalili, S.~Saremi, and I.~N. Trivedi,
  ``Optimization of problems with multiple objectives using the multi-verse
  optimization algorithm,'' \emph{Knowl. Based Syst.}, vol. 134, pp. 50--71,
  Oct. 2017.

\bibitem{Bader2011}
J.~Bader and E.~Zitzler, ``Hype: An algorithm for fast hypervolume-based
  many-objective optimization,'' \emph{Evol. Comput.}, vol.~19, no.~1, pp.
  45--76, Mar. 2011.

\bibitem{Zhang2023}
J.~Zhang, Z.~Ning, R.~H. Ali, M.~Waqas, S.~Tu, and I.~Ahmad, ``A many-objective
  ensemble optimization algorithm for the edge cloud resource scheduling
  problem,'' \emph{IEEE Trans. Mobile Comput.}, pp. 1--18, Jan. 2023, {E}arly
  {A}ccess, doi: {10.1109/TMC.2023.3235064}.

\bibitem{Kim2020}
T.~Kim and D.~Qiao, ``Energy-efficient data collection for {I}o{T} networks via
  cooperative multi-hop {UAV} networks,'' \emph{{IEEE} Trans. Veh. Technol.},
  vol.~69, no.~11, pp. 13\,796--13\,811, Nov. 2020.

\end{thebibliography}
	
	\vspace{10pt}
	
	\begin{IEEEbiography}[{\includegraphics[width=1in,height=1.25in,clip,keepaspectratio]{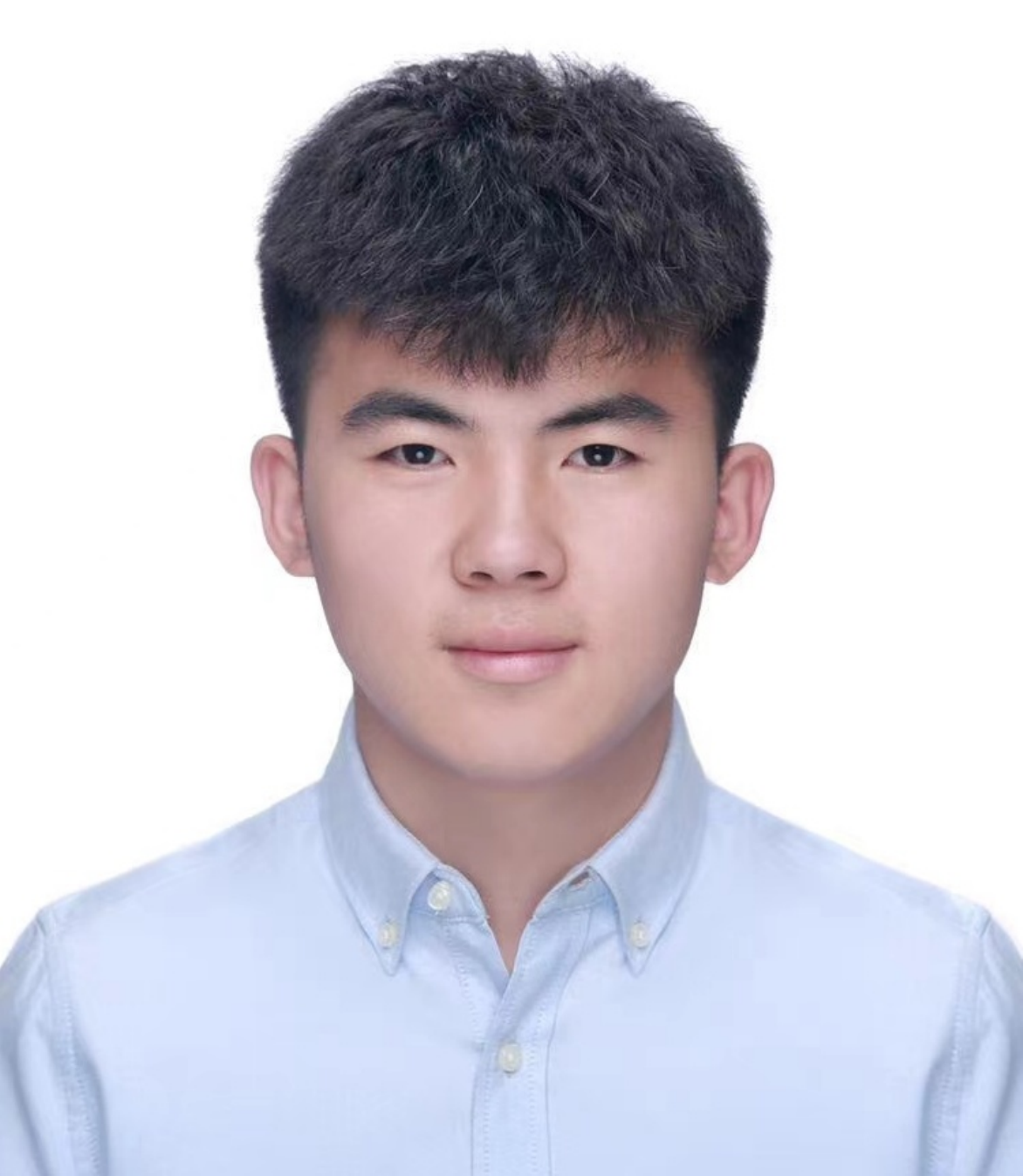}}]{Chuang Zhang} received the B.S. degree in computer science and technology from Jilin University, Changchun, China, in 2021, where he is currently pursuing the Ph.D. degree with the College of Computer Science and Technology.
		\par His current research interests include UAV communications, distributed beamforming and multi-objective optimization.
	\end{IEEEbiography}
	
	\begin{IEEEbiography}[{\includegraphics[width=1in,height=1.25in,clip,keepaspectratio]{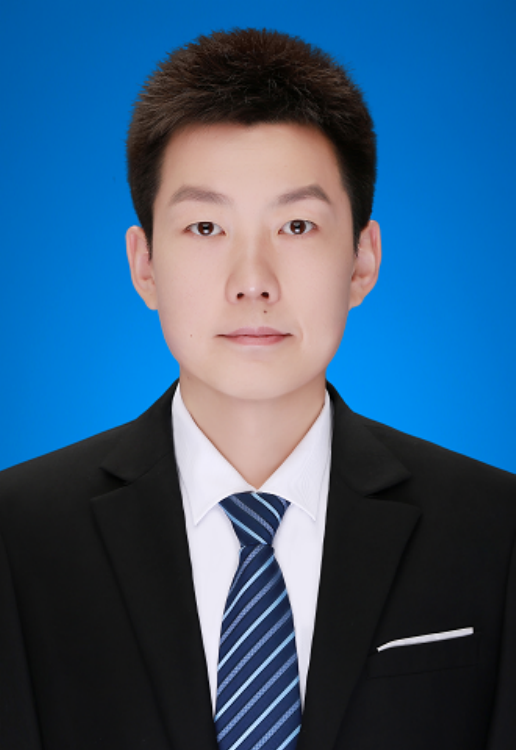}}]{Geng Sun} (S'17-M'19) received the B.S. degree in communication engineering from Dalian Polytechnic University, and the Ph.D. degree in computer science and technology from Jilin University, in 2011 and 2018, respectively. He was a Visiting Researcher with the School of Electrical and Computer Engineering, Georgia Institute of Technology, USA. He is an Associate Professor in College of Computer Science and Technology at Jilin University, and His research interests include wireless networks, UAV communications, collaborative beamforming and optimizations.
	\end{IEEEbiography}
	
	\begin{IEEEbiography}[{\includegraphics[width=1in,height=1.25in,clip,keepaspectratio]{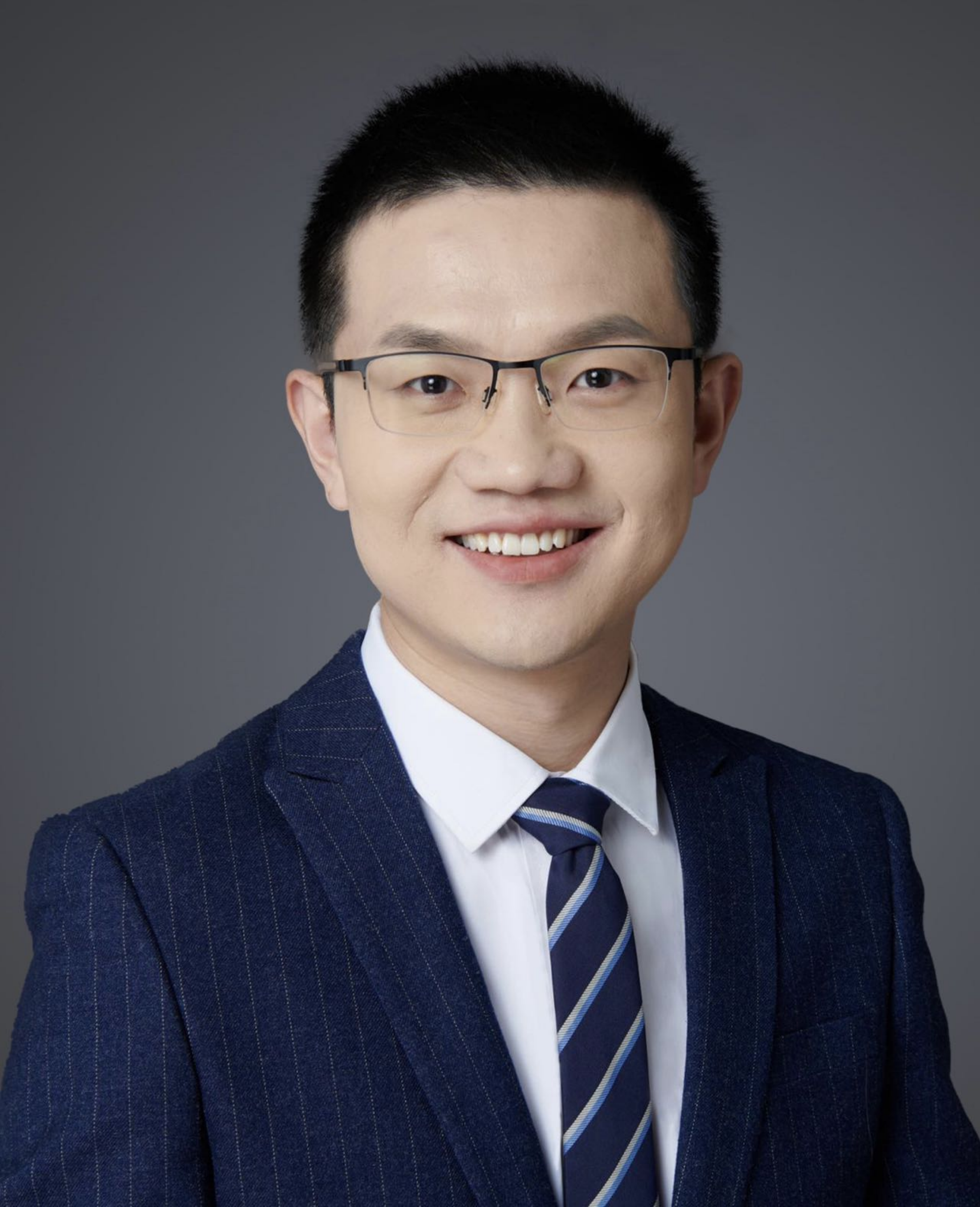}}]{Qingqing Wu} (S’13-M’16-SM’21) received the B.Eng. and the Ph.D. degrees in Electronic Engineering from South China University of Technology and Shanghai Jiao Tong University (SJTU) in 2012 and 2016, respectively. From 2016 to 2020, he was a Research Fellow in the Department of Electrical and Computer Engineering at National University of Singapore. He is currently an Associate Professor with Shanghai Jiao Tong University. His current research interest includes intelligent reflecting surface (IRS), unmanned aerial vehicle (UAV) communications, and MIMO transceiver design. He has coauthored more than 100 IEEE journal papers with 26 ESI highly cited papers and 8 ESI hot papers, which have received more than 18,000 Google citations. He was listed as the Clarivate ESI Highly Cited Researcher in 2022 and 2021, the Most Influential Scholar Award in AI-2000 by Aminer in 2021 and World’s Top 2\% Scientist by Stanford University in 2020 and 2021.
		\par He was the recipient of the IEEE Communications Society Asia Pacific Best Young Researcher Award and Outstanding Paper Award in 2022, the IEEE Communications Society Young Author Best Paper Award in 2021, the Outstanding Ph.D. Thesis Award of China Institute of Communications in 2017, the Outstanding Ph.D. Thesis Funding in SJTU in 2016, the IEEE ICCC Best Paper Award in 2021, and IEEE WCSP Best Paper Award in 2015. He was the Exemplary Editor of IEEE Communications Letters in 2019 and the Exemplary Reviewer of several IEEE journals. He serves as an Associate Editor for IEEE Transactions on Communications, IEEE Communications Letters, IEEE Wireless Communications Letters, IEEE Open Journal of Communications Society (OJ COMS), and IEEE Open Journal of Vehicular Technology (OJVT). He is the Lead Guest Editor for IEEE Journal on Selected Areas in Communications on “UAV Communications in 5G and Beyond Networks”, and the Guest Editor for IEEE OJVT on “6G Intelligent Communications” and IEEE OJ-COMS on “Reconfigurable Intelligent Surface-Based Communications for 6G Wireless Networks”. He is the workshop co-chair for IEEE ICC 2019-2022 workshop on “Integrating UAVs into 5G and Beyond”, and the workshop co-chair for IEEE GLOBECOM 2020 and ICC 2021 workshop on “Reconfigurable Intelligent Surfaces for Wireless Communication for Beyond 5G”. He serves as the Workshops and Symposia Officer of Reconfigurable Intelligent Surfaces Emerging Technology Initiative and Research Blog Officer of Aerial Communications Emerging Technology Initiative. He is the IEEE Communications Society Young Professional Chair in Asia Pacific Region.
	\end{IEEEbiography}
	
	\begin{IEEEbiography}[{\includegraphics[width=1in,height=1.25in,clip,keepaspectratio]{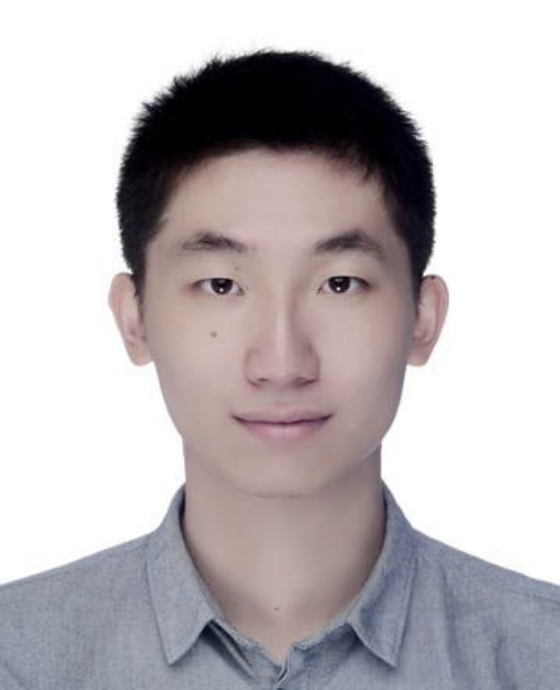}}]{Jiahui Li} (S'21) received a BS degree in Software Engineering, and an MS degree in Computer Science and Technology from Jilin University, Changchun, China, in 2018 and 2021, respectively. He is currently studying Computer Science at Jilin University to get a Ph.D. degree, and also a visiting Ph. D. at Singapore University of Technology and Design (SUTD), Singapore. His current research focuses on UAV networks, antenna arrays, and optimization.
	\end{IEEEbiography}

	\begin{IEEEbiography}[{\includegraphics[width=1in,height=1.25in,clip,keepaspectratio]{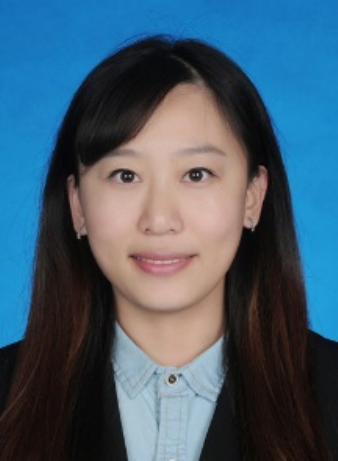}}]{Shuang Liang} received the B.S. degree in Communication Engineering from Dalian Polytechnic University, China in 2011, the M.S. degree in Software Engineering from Jilin University, China in 2017, and the Ph.D. degree in Computer Science from Jilin University, China in 2022. She is a post-doctoral in the School of Information Science and Technology, Northeast Normal University, and her research interests focus on wireless communication and UAV networks.
	\end{IEEEbiography}
	
	\begin{IEEEbiography}[{\includegraphics[width=1in,height=1.25in,clip,keepaspectratio]{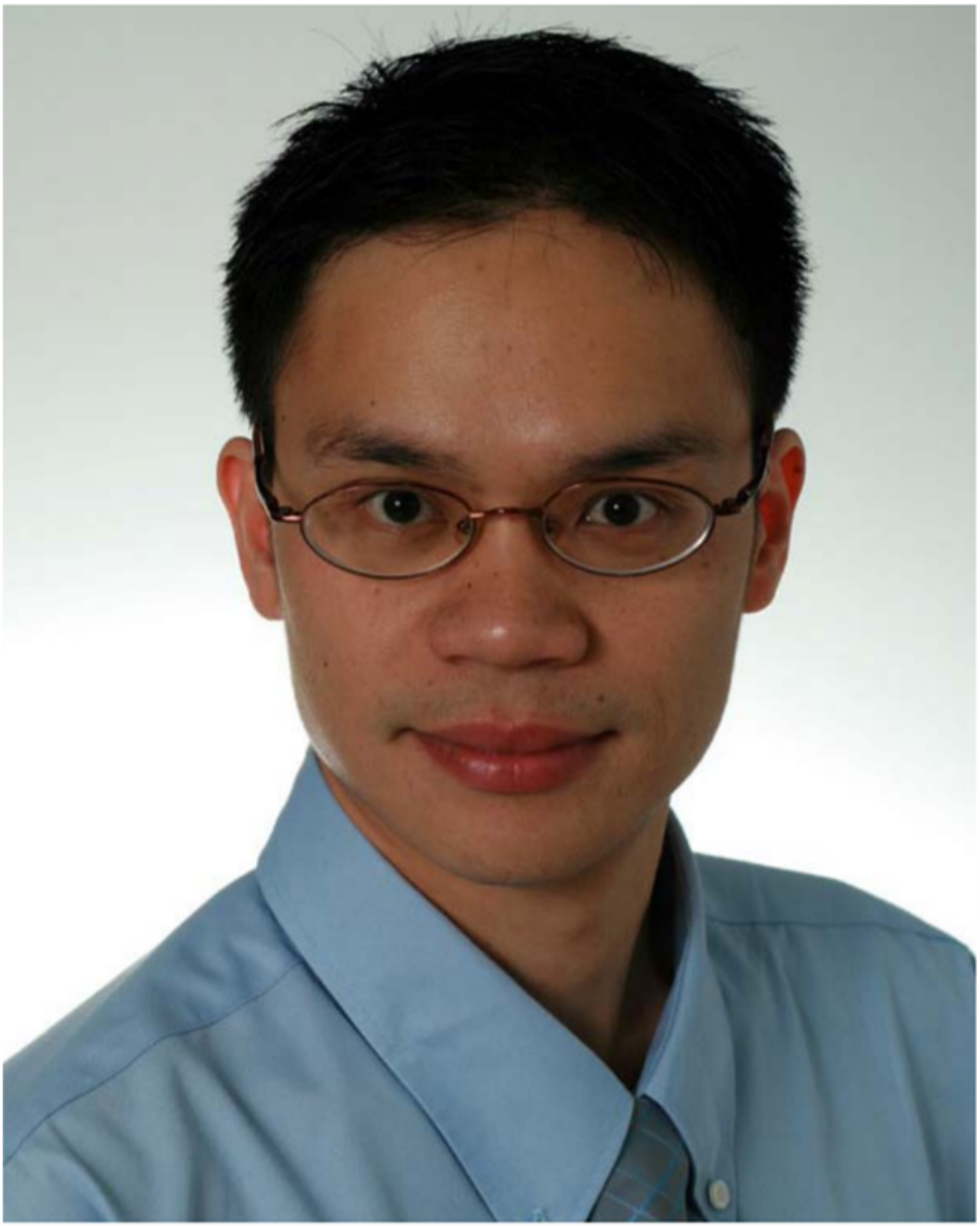}}]{Dusit-Niyato} (Fellow, IEEE) received the B.Eng. degree from the King Mongkuts Institute of Technology Ladkrabang (KMITL), Thailand, in 1999, and the Ph.D. degree in electrical and computer engineering from the University of Manitoba, Canada, in 2008. He is currently a Professor with the School of Computer Science and Engineering, Nanyang Technological University, Singapore. His research interests include the Internet of Things (IoT), machine learning, and incentive mechanism design.
	\end{IEEEbiography}
	
	\begin{IEEEbiography}[{\includegraphics[width=1in,height=1.25in,clip,keepaspectratio]{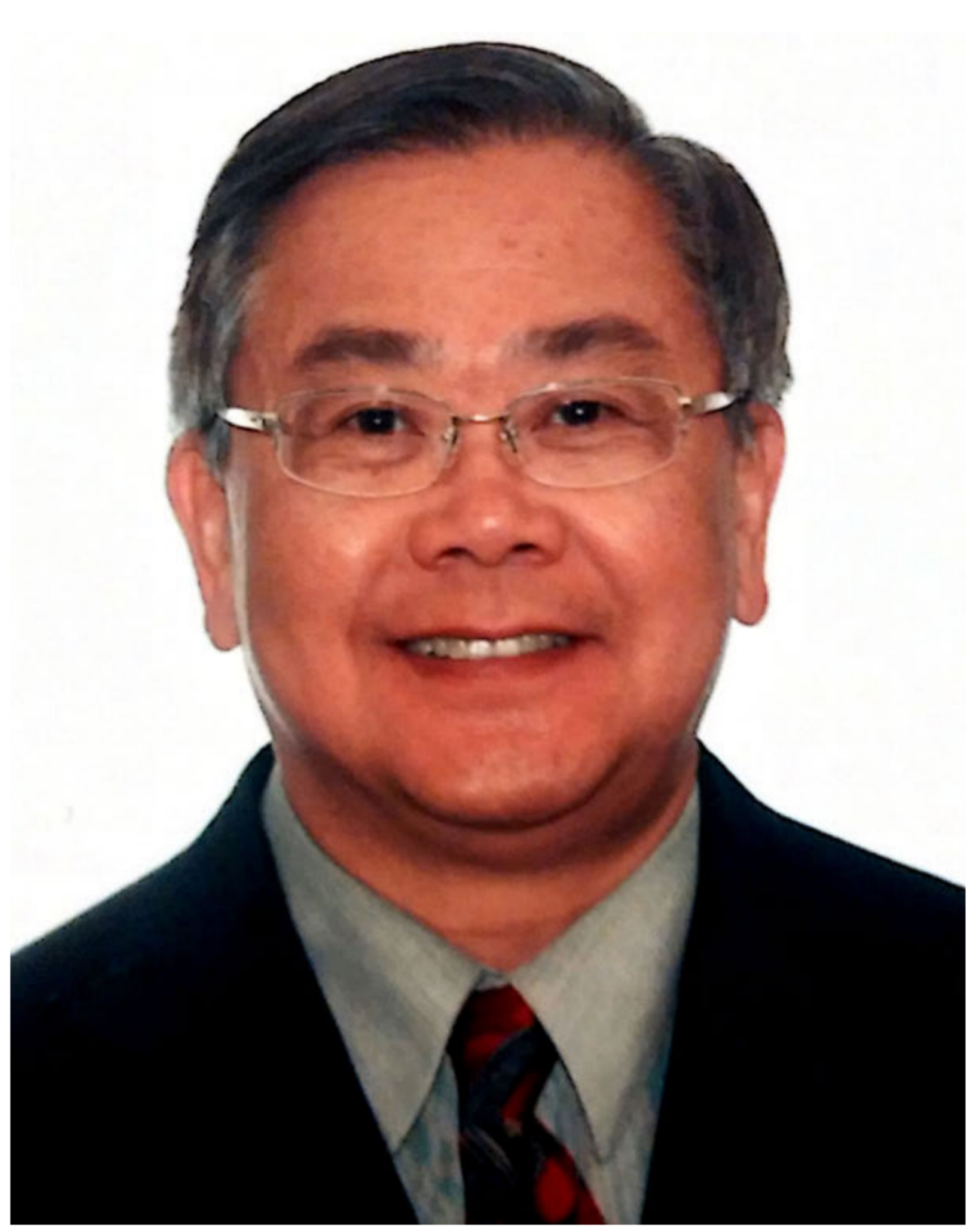}}]{Victor C. M. Leung} (Life Fellow, IEEE) is a Distinguished Professor of computer science and software engineering with Shenzhen University, China. He is also an Emeritus Professor of electrial and computer engineering and the Director of the Laboratory for Wireless Networks and Mobile Systems at the University of British Columbia (UBC). His research is in the broad areas of wireless networks and mobile systems. He has co-authored more than 1300 journal/conference papers and book chapters. Dr. Leung is serving on the editorial boards of IEEE Transactions on Green Communications and Networking, IEEE Transactions on Cloud Computing, IEEE Access, and several other journals. He received the IEEE Vancouver Section Centennial Award, 2011 UBC Killam Research Prize, 2017 Canadian Award for Telecommunications Research, and 2018 IEEE TCGCC Distinguished Technical Achievement Recognition Award. He co-authored papers that won the 2017 IEEE ComSoc Fred W. Ellersick Prize, 2017 IEEE Systems Journal Best Paper Award, 2018 IEEE CSIM Best Journal Paper Award, and 2019 IEEE TCGCC Best Journal Paper Award. He is a Life Fellow of IEEE, and a Fellow of the Royal Society of Canada, Canadian Academy of Engineering, and Engineering Institute of Canada. He is named in the current Clarivate Analytics list of Highly Cited Researchers.
	\end{IEEEbiography}
	\vspace{11pt}
	
	\vfill
	
\end{document}